\documentclass[12pt]{article}
\usepackage{amssymb}
\usepackage{amsfonts}
\usepackage{amsmath}
\usepackage{eurosym}
\usepackage{makeidx}
\usepackage{amssymb}
\usepackage{times}
\usepackage{anysize}
\usepackage{cite}
\usepackage{lineno}

\marginsize{2cm}{2cm}{1cm}{1cm}
\newtheorem{theorem}{Theorem}[section]

\newtheorem{corollary}[theorem]{Corollary}

\newtheorem{definition}[theorem]{Definition}

\newtheorem{lemma}[theorem]{Lemma}
\newtheorem{notation}[theorem]{Notation}

\newtheorem{proposition}[theorem]{Proposition}
\newtheorem{remark}[theorem]{Remark}

\newenvironment{proof}[1][Proof]{\noindent\textbf{#1.} }{\ \rule{0.5em}{0.5em}}
\input{tcilatex}
\begin{document}

\title{Classical Dynamics Generated by Long-Range Interactions for Lattice
Fermions and Quantum Spins}
\author{J.-B. Bru \and W. de Siqueira Pedra}
\date{\today }
\maketitle

\begin{abstract}
\bigskip We study the macroscopic dynamical properties of fermion and
quantum-spin systems with long-range, or mean-field, interactions. The
results obtained are far beyond previous ones and require the development of
a mathematical framework to accommodate the macroscopic long-range dynamics,
which corresponds to an intricate combination of classical and short-range
quantum dynamics. In this paper we focus on the classical part of the
long-range, or mean-field, macroscopic dynamics, but we already introduce
the full framework. The quantum part of the macroscopic dynamics is studied
in a subsequent paper. We show that the classical part of the macroscopic
dynamics results from self-consistency equations within the (quantum) state
space. As is usual, the classical dynamics is driven by Liouville's
equation.\bigskip

\noindent \textit{Dedicated to V.A. Zagrebnov for his important
contributions to the mathematics of quantum many-body theory.}\bigskip
\bigskip

\bigskip

\end{abstract}

\tableofcontents%

\section{Introduction}

More than seventy years ago, Bogoliubov proposes an ansatz, widely known as
the Bogoliubov approximation, which corresponds to replace, in many-boson%
\footnote{%
Quantum particles carry an intrinsic form of angular momentum, the so-called
spin. It is related to a spin quantum number $\mathfrak{s}\in \mathbb{N}/2$
(depending on the nature of a given particle) which fixes a finite spin set $%
\mathrm{S}\doteq \left\{ -\mathfrak{s},-\mathfrak{s}+1,\ldots \mathfrak{s}-1,%
\mathfrak{s}\right\} $. If $\mathfrak{s}$ is a half-integer, i.e. $\mathfrak{%
s}\notin \mathbb{N}$ , then the corresponding particles are fermions while
the case $\mathfrak{s}\in \mathbb{N}$ correspond to bosons. By the
celebrated spin-statistics theorem, fermionic wave functions are
antisymmetric with respect to permutations of particles, whereas the bosonic
ones are symmetric.} Hamiltonians, the annihilation and creation operators
of zero-momentum particles with complex numbers to be determined \emph{%
self-consistently}.\ See \cite[Section 1.1]{BruZagrebnov8} for more details.
His motivation comes from the observation that these (unbounded) operators
almost commute with any other operator in the thermodynamic limit, leading
to some \emph{classical} field in the macroscopic Bose system. In 1968,
Ginibre \cite{Ginibre} shows the exactness of the approximation in which
concerns the thermodynamic pressure of superstable Bose gases. However, even
nowadays, the mathematical validity of this approximation with respect to
the primordial dynamics of (stable) many-boson Hamiltonians with usual
two-body interactions is an open problem.

In the context of many-fermion systems, ten years after Bogoliubov's ansatz,
a similar approximation is used in the BCS theory of (conventional)
superconductivity, as explained by Bogliubov in 1958 \cite{Bogoliubov58} and
Haag in 1962 \cite{haag62}. In 1966, this approximation is shown \cite%
{Bogjunior} to be exact at the level of the thermodynamic pressure for
fermion systems that are similar to the BCS model. See also the so-called
approximating Hamiltonian method used on the level of the pressure of
fermionic systems in \cite%
{Bogjunior,approx-hamil-method0,approx-hamil-method,approx-hamil-method2}.

The validity of the approximation with respect to the primordial dynamics
was an open question that Thirring and Wehrl \cite{T1,T2} solve in 1967 for
an exactly solvable permutation-invariant fermion model. An attempt to
generalize Thirring and Wehrl's results to a general class of fermionic
models, including the BCS theory, has been done in 1978 \cite{Hemmen78}, but
at the cost of technical assumptions that are difficult to verify in
practice.

In 1973, Hepp and Lieb \cite{Hepp-Lieb73} made explicit, for the first time,
the existence of Poisson brackets in some (commutative) algebra of
functions, related to the classical effective dynamics. This is done for a
permutation-invariant quantum-spin system with mean-field interactions. This
research direction has been strongly developed by many authors until 1992,
see \cite%
{Bona75,Sewell83,Rieckers84,Morchio87,Bona87,Duffner-Rieckers88,Bona88,Bona89,Bona90,Unnerstall90,Unnerstall90b,Unnerstall90-open,Bona91,Duffield1991,BagarelloMorchio92,Duffield-Werner1,Duffield-Werner2,Duffield-Werner3,Duffield-Werner4,Duffield-Werner5}%
. All these papers study dynamical properties of \emph{permutation-invariant}
quantum-spin systems with mean-field interactions. Even if we are rather
interested in mean-field dynamics, note meanwhile that equilibrium
properties of such quantum systems are also extensively studied in the same
period, see for instance \cite%
{QuaegebeurVerbeure1980,RaggioWerner1,RaggioWerner2} and references therein.

Thereafter, the mathematical research activity on this subject considerably
decreases until the early 2000s when emerges, within the mathematical
physics community, a new interest in such quantum systems, partially because
of new experiments like those on ultracold atoms (via laser and evaporative
coolings). See, for instance, the paper \cite{non-semiclassique1bisbis} on
mean-field dynamics, published in the year 2000. In 2005, Ginibre's result
on Bose gases at thermodynamical equilibrium is generalized \cite{LSY05}.
There is also an important research activity on the mathematical foundation
of the Gross-Pitaevskii\footnote{%
The so-called GP limit is not really a mean-field limit, but it looks
similar.} (GP) or Hartree theories, starting after 1998. For more details on
the GP theory and mean-field dynamics for indistinguishable particles
(bosons), see \cite{LSSY05,Shlein2016,GolseMonhot2016,JP18,Ammari2018} and
references therein. In which concern lattice-fermion\footnote{%
The fermionic nature of particles is encoded in Canonical Anti-commutation
Relations (CAR) which lead, in the algebraic formulation, to a so-called CAR
algebra associated with a one-particle Hilbert space $\mathfrak{h}$. The
term \textquotedblleft lattice-fermion\textquotedblright\ only refers in
this case to the special choice $\mathfrak{h}=\ell ^{2}(\mathbb{Z}^{d}\times 
\mathrm{S})$, where $\mathrm{S}$ is some finite set of spins and $d\in 
\mathbb{N}$.} or quantum-spin systems\footnote{%
Quantum spin systems refer to infinite tensor products of a fixed matrix
algebra, enconding a spin variable at any fixed lattice site. Their $C^{\ast
}$ algebras are constructed in a similar way as those for lattice fermions
(in fact, these infinite tensor products are even isomorphic, as $C^{\ast }$
algebras, to CAR algebras). Results of this paper hold true also in this
context, via obvious modifications.} with long-range\footnote{%
By long-range interactions, we basically mean interactions that are \emph{not%
} decaying at large interparticle distances. Note however that the attribute
\textquotedblleft long-range\textquotedblright\ is used in the literature
also for interactions that slowly decay with respect to interparticle
distances. Our notion of \textquotedblleft long-range\textquotedblright\ is
therefore stronger, but still includes important physical models like some
of those related to conventional superconductivity. It refers to a type of
mean-field interactions widely used in theoretical and mathematical physics.}%
, or mean-field\footnote{%
The attribute \textquotedblleft mean-field\textquotedblright\ means here
that the (interparticle) interaction is averaged over the number of sites of
given boxes.}, interactions at equilibrium, see, e.g., \cite%
{Petz2008,BruPedra1,monsieurremark,BruPedra2}. Concerning the dynamics of
fermion systems in the continuum with mean-field interactions, see \cite%
{non-semiclassique1,non-semiclassique1bis,Shlein2004,non-semiclassique2,Shlein2014bis,Schlein2014bisbis,Shlein2015,non-semiclassique3,Breteaux2016,Shlein2017}%
, as well as \cite[Sections 6 and 7]{Shlein2016}. Such mean-field problems
are even related to other academic disciplines, like mathematical economics,
via the so-called mean-field game theory \cite{MFgame} developed since 2006
by Lasry and Lions. Mean-field theory in its extended sense is, in fact, a
major research field of mathematics, even nowadays, and is still studied in
physics, see, e.g., \cite{extra-ref0000} and references therein. The current
paper belongs to this research field, since we hereby study the dynamical
properties of fermion and quantum-spin systems with long-range, or
mean-field, interactions. \smallskip

All mathematical studies \cite%
{Shlein2016,non-semiclassique1,non-semiclassique1bis,Shlein2004,non-semiclassique2,Shlein2014bis,Schlein2014bisbis,Shlein2015,non-semiclassique3,Breteaux2016,Shlein2017}
of the 2000s on mean-field fermionic dynamics are on the continuum and use
approximating \emph{quasi-free}\footnote{%
In some papers, only (approximating) Slater determinants are considered.}
states (or a mixture of them) as initial states to derive the Hartree-Fock
equation, which is originally based on the assumption that the many-fermion
wavefunction is a Slater determinant. Quoting \cite[p. 79]{Shlein2016}:
\textquotedblleft \textit{Slater determinants are relevant at zero
temperature because they provide (or at least they are expected to provide)
a good approximation to the fermionic ground state of Hamilton operators
like (6.1) in the mean-field limit. At positive temperature, equilibrium
states are mixed; in the mean-field regime, they are expected to be
approximately quasi-free mixed states, like the Gibbs state of a
non-interacting gas.}\textquotedblright\ 

These arguments are probably true in the mean-field regimes considered in
these studies because the non-mean-field part of the model is always
associated with a one-particle Hamiltonian, but one cannot expect this
property to hold true for general fermion models with long-range, or
mean-field, interactions. Equilibrium states of the usual BCS\ theory is a
mixture of quasi-free states, even at zero temperature, in presence of a
superconducting phase, because of gauge invariance. Adding a Hubbard
interaction to such BCS-type models directly destroys the quasi-free
property of equilibrium states, even in the sense of a mixture. Indeed, a
general many-fermion wavefunction cannot be expressed as a single
determinant, even at zero temperature, and, consequently, we cannot expect
the Hartree-Fock equation to be generally correct. For instance, this method
usually overestimates the full (ground state) energy. To solve that problem
in computational chemistry and condensed matter physics one can use either
extensions of the Hartree-Fock method, or the Density Functional Theory%
\footnote{%
A not necessarly better, but computationally more efficient approximation.},
which is based on some energy functional for the one-fermion density only,
via the Hohenberg-Kohn theorems. See for instance \cite{DFT}.\medskip

In contrast with results \cite%
{Shlein2016,non-semiclassique1,non-semiclassique1bis,Shlein2004,non-semiclassique2,Shlein2014bis,Schlein2014bisbis,Shlein2015,non-semiclassique3,Breteaux2016,Shlein2017}
of the 2000s, we consider fermions on the lattice, not in the continuum, and
we meanwhile study quantum-spin systems. Our analysis of the corresponding
mean-field dynamics does \emph{not} require to have approximating quasi-free
states (or a mixture of them) as initial states. \medskip

The recent studies \cite%
{non-semiclassique1,non-semiclassique1bis,Shlein2004,non-semiclassique2,Shlein2014bis,Schlein2014bisbis,Shlein2015,non-semiclassique3,Shlein2016,Breteaux2016,Shlein2017}
of the 2000s never extract effective classical dynamics within the quantum
one. The reason, for most of them \cite%
{Shlein2004,Shlein2014bis,Schlein2014bisbis,Shlein2015,non-semiclassique3,Shlein2016,Shlein2017}%
, is that their mean-field scales reveal a \emph{semi-classical structure},
similar to what is first done in 1981 to derive the Vlasov hierarchy in \cite%
{vlasov81,vlasov81bis} from quantum dynamics. In order to see the intricate
combination of a classical dynamics and a\ quantum one, one has to go back
to previous results \cite%
{Hepp-Lieb73,Bona75,Sewell83,Rieckers84,Morchio87,Bona87,Duffner-Rieckers88,Bona88,Bona89,Bona90,Unnerstall90,Unnerstall90b,Unnerstall90-open,Bona91,Duffield1991,BagarelloMorchio92,Duffield-Werner1,Duffield-Werner2,Duffield-Werner3,Duffield-Werner4,Duffield-Werner5}%
, initiated by Hepp and Lieb. In our opinion, the most elaborated and
interesting results in that context are obtained by B\'{o}na \cite%
{Bona87,Bona88,Bona89,Bona90} in 1987-1990, who studies in detail the
dynamics of a large class of permutation-invariant quantum-spin systems with
mean-field interactions. B\'{o}na formalizes a new view point on quantum
mechanics in 1991 \cite{Bona91} and later in a mature textbook published in
2000 and revised in 2012 \cite{Bono2000}, describing what he names
\textquotedblleft extended quantum mechanics\textquotedblright : It is an
intricate combination of classical and quantum dynamics. B\'{o}na's approach
does not seem to be incorporated by the physics community, until now.
\medskip

In contrast with the results \cite%
{Hepp-Lieb73,Bona75,Sewell83,Rieckers84,Morchio87,Bona87,Duffner-Rieckers88,Bona88,Bona89,Bona90,Unnerstall90,Unnerstall90b,Unnerstall90-open,Bona91,Duffield1991,BagarelloMorchio92,Duffield-Werner1,Duffield-Werner2,Duffield-Werner3,Duffield-Werner4,Duffield-Werner5}
of the 70s-90s (1973-1992), our study does \emph{not} require the
permutation invariance of quantum-spin systems. In fact, we are able to
study dynamical properties of a very general class of lattice-fermion, or
quantum-spin, systems with long-range, or mean-field, \emph{%
translation-invariant}\footnote{%
This can be easily extended to periodically invariant systems, by redefining
the finite spin set.} interactions. In this context, the initial state is
only assumed to be periodic with respect to space translations. In fact, by
Proposition \ref{density of periodic states}, the set of all initial states
allowed in our study is weak$^{\ast }$-dense within the set of even states,
the physically relevant ones. Even for permutation-invariant systems, our
results go beyond previous ones, since the class of long-range, or
mean-field, interactions we are able to handle is much larger.

Our study reveals an entanglement of classical and quantum dynamics, as
explained in \cite{Bru-pedra-MF-I}, which revisits B\'{o}na's approach \cite%
{Bono2000} on \textquotedblleft extended quantum mechanics\textquotedblright
. In \cite{Bru-pedra-MF-I}, we do not really follow B\'{o}na's path but
highlight the importance of self-consistency equations, instead. This study
leads to a new theory which strongly differs from the so-called
quantum-classical hybrid theories of theoretical and mathematical physics,
including the recent mathematical results \cite{1,2,3}. See below Remark \ref%
{Quantum-classical hybrid theories}. In this paper and in the subsequent one 
\cite{Bru-pedra-MF-III}, we show that this mathematical framework is
imperative to describe the dynamical properties of lattice-fermion or
quantum-spin systems with long-range, or mean-field, interactions, because
of the necessity of coupled classical and quantum evolution
equations.\medskip

Here, we focus on the classical part of the long-range, or mean-field,
dynamics of the quantum systems under consideration, even if we have to
consider the full algebraic framework, as described in \cite{Bru-pedra-MF-I}%
. The quantum part of the dynamics, shortly discussed in Section \ref%
{Quantum Part}, is studied in a companion paper \cite{Bru-pedra-MF-III}
because it additionally involves the theory of direct integrals of
measurable families of Hilbert spaces, operators, von Neumann algebras and $%
C^{\ast }$-algebra representations, as presented, for instance, in the
monograph \cite{Niesen-direct integrals}.

The classical part of the dynamics of lattice-fermion, or quantum-spin,
systems with long-range, or mean-field, interactions is shown to result from
the solution of \ a self-consistency equation within the (quantum) state
space. In general, it is a \emph{non-linear} dynamics generated by a Poisson
bracket. In other words, it is driven by some Liouville's equation, similar
to what has been recently observed \cite{Ammari2018} by Ammari, Liard, and
Rouffort for Bose gases in the mean-field limit. As soon as the classical
part of this dynamics is concerned, there is no need to assume any
particular property, neither on initial states, nor on local Hamiltonians.
Translation invariance, or periodicity, is only required for the analysis of
the quantum part of the dynamics. So, these particular spatial features are
only pivotal in \cite{Bru-pedra-MF-III}.

It would have been interesting to add applications of our general theory,
but we still have to present (in \cite{Bru-pedra-MF-III}) the quantum part
of such dynamics before doing that. In fact, several applications on
quasi-free models as well as permutation-invariant systems will be presented
in separated papers. See, e.g., \cite{Bru-pedra-proceeding,Bru-pedra-MF-IV}.
We will also discuss in a forthcoming paper such mean-field dynamics within
the representation of an arbitrary (generalized) equilibrium state\footnote{%
This class of states is studied in detail in \cite{BruPedra2}.} of the
quantum system under consideration.

\begin{remark}[Quantum-classical hybrid theories]
\label{Quantum-classical hybrid theories}\mbox{
}\newline
Many important models of quantum mechanics represent systems of quantum
particles in interaction with classical fields. For example, a quantum
particle interacting with an external electromagnetic field is commonly
studied via the magnetic Laplacian. In other words, these models implicitly
combine quantum and classical mechanics. This simplification is physically
justified by the huge numbers of photons giving origin to (effective
classical) macroscopic fields, in the spirit of the correspondence
principle. The very recent paper \cite{3} gives a mathematical justification
of such a procedure for three important quantum models: the Nelson,
Pauli-Fierz and Polaron models. Mathematically, it refers to a
semi-classical analysis of the bosonic degrees of freedom, leading to a new
(equivalent) quantum model interacting with a classical field. Thanks to 
\cite{1,2}, the authors are able to handle the semi-classical limit for very
general (possibly entangled, i.e., not factorized) states. These
mathematical results are highly non-trivial. However, such a procedure does
not entangle classical and quantum dynamics in the way it occurs in presence
of long-range interactions. E.g., the relation between the phase space of
the classical dynamics and the state space of the $C^{\ast }$-algebra where
the quantum dynamics runs is not manifest. In fact, the situation presented
in this paper is essentially different than any quantum-classical hybrid
theory.
\end{remark}

Our main results are the following:

\begin{itemize}
\item Theorem \ref{theorem sdfkjsdklfjsdklfj} shows that the
self-consistency equations are well-posed for any initial state. Such
equations are central to get a complete description of the macroscopic
dynamics of long-range models.

\item Proposition \ref{lemma poisson copy(2)} states the existence of
evolution groups while Corollary \ref{Closed Poissonian symmetric
derivations copy(1)} shows that their generators can be represented via
Poisson brackets and a very natural classical energy function on the state
space of the CAR algebra.

\item Theorem \ref{classical dynamics I} demonstrates the existence of
non-autonomous classical dynamics on the state space. It is a highly
non-trivial mathematical statement in the theory of non-auto%
\-%
nomous evolution equations \cite{Katobis,Caps,Schnaubelt1,Pazy,Bru-Bach}. It
results from Lieb-Robinson bounds on multi-commutators of high orders
derived in 2017 \cite{brupedraLR}.

\item Corollary \ref{classical dynamics I copy(1)} reduces the classical
dynamics to the autonomous situation, which leads to the usual (autonomous)
dynamics of classical mechanics written in terms of Poisson brackets, i.e., 
\emph{Liouville's equation}.
\end{itemize}

\noindent The paper is organized as follows: Section \ref{Section FERMI}
explains the well-established dynamical properties of lattice-fermion
systems with \emph{short-range}\footnote{%
Interactions are said to be short-range when they decay sufficiently fast at
large distances. This refers to a sufficient polynomial decay such that the
celebrated Lieb-Robinson bounds hold true for the corresponding dynamics.}
interactions. In this context, the celebrated Lieb-Robinson bounds are
pivotal and are also an important tool to obtain the full dynamics of
systems with long-range, or mean-field, interactions. Section \ref%
{Long-Range systems} defines what we name \textquotedblleft long-range
interactions\textquotedblright\ and we make explicit the problem of the
thermodynamic limit\ of their associated dynamics. From now on, like in \cite%
{BruPedra2} we use the term \textquotedblleft long-range\textquotedblright ,
only,\ instead of \textquotedblleft mean-field\textquotedblright . In fact,
we prefer the first rather than the second attribute, because the latter can
refer to different scalings, in particular in the recent works \cite%
{non-semiclassique1,non-semiclassique1bis,Shlein2004,non-semiclassique2,Shlein2014bis,Schlein2014bisbis,Shlein2015,non-semiclassique3,Shlein2016,Breteaux2016,Shlein2017}%
. Note, additionally, that \textquotedblleft Bogoliubov's
school\textquotedblright\ also uses the term \textquotedblleft
long-range\textquotedblright\ for models containing interactions that are
mean-field in the sense of the present paper. In Section \ref{sect
Lieb--Robinson copy(3)}, we present the self-consistency equations as well
as the classical part of long-range dynamics. This requires the mathematical
framework of \cite{Bru-pedra-MF-I}, which has thus to be presented in the
first subsections of Section \ref{sect Lieb--Robinson copy(3)}. Note that we
shortly explain what the quantum part of the long-range dynamics is in
Section \ref{Quantum Part}, even if this will be done in detail in \cite%
{Bru-pedra-MF-III}. Finally, Section \ref{Well-posedness sect} gives the
proofs of Theorems \ref{theorem sdfkjsdklfjsdklfj} and \ref{classical
dynamics I}. Sections \ref{Long-range models} and \ref{Section applications}
are two appendices relating our current notion of long-range interactions to
the one of \cite{BruPedra2} and discussing possible applications in physics.

In this paper, we only focus on lattice-fermion systems which are, from a
technical point of view, slightly more difficult than quantum-spin systems,
because of a non-commutativity issue at different lattice sites. However,
all the results presented here and in the subsequent papers hold true for
quantum-spin systems, via obvious modifications.

\begin{notation}
\label{remark constant}\mbox{
}\newline
\emph{(i)} A norm on a generic vector space $\mathcal{X}$ is denoted by $%
\Vert \cdot \Vert _{\mathcal{X}}$ and the identity mapping of $\mathcal{X}$
by $\mathbf{1}_{\mathcal{X}}$. The space of all bounded linear operators on $%
(\mathcal{X},\Vert \cdot \Vert _{\mathcal{X}}\mathcal{)}$ is denoted by $%
\mathcal{B}(\mathcal{X})$. The unit element of any algebra $\mathcal{X}$ is
denoted by $\mathfrak{1}$, provided it exists. \newline
\emph{(ii)} For any topological space $\mathcal{X}$ and normed space $(%
\mathcal{Y},\Vert \cdot \Vert _{\mathcal{Y}}\mathcal{)}$, $C\left( \mathcal{X%
};\mathcal{Y}\right) $ denotes the space of continuous mappings from $%
\mathcal{X}$ to $\mathcal{Y}$. If $\mathcal{X}$ is a locally compact
topological space, then $C_{b}\left( \mathcal{X};\mathcal{Y}\right) $
denotes the Banach space of bounded continuous mappings from $\mathcal{X}$
to $\mathcal{Y}$ along with the topology of uniform convergence.\newline
\emph{(iii)} The notion of an automorphism depends on the structure of the
corresponding domain. In algebra, a ($\ast $-) automorphism acting on a $%
\ast $-algebra is a bijective $\ast $-homomorphism from this algebra to
itself. In topology, an automorphism acting on a topological space is a
self-homeomorphism, that is, a homeomorphism of the space to itself.
\end{notation}

\section{Algebraic Formulation of Lattice Fermion Systems\label{Section
FERMI0}}

\subsection{Background Lattice}

Fix once and for all $d\in \mathbb{N}$, the dimension of the lattice. Choose
also a subset $\mathfrak{L}\subseteq \mathbb{Z}^{d}$ (which will be omitted
in our notation, unless it is important to specify it). The main example we
have in mind is $\mathfrak{L}=\mathbb{Z}^{d}$, as is done in \cite{BruPedra2}%
, but it is convenient to keep this set as being any arbitrary subset of the 
$d$-dimensional lattice. Let $\mathcal{P}_{f}\equiv \mathcal{P}_{f}\left( 
\mathfrak{L}\right) \subseteq 2^{\mathfrak{L}}$ be the set of all non-empty
finite subsets of $\mathfrak{L}$. In order to define the thermodynamic
limit, for simplicity, we use the cubic boxes 
\begin{equation}
\Lambda _{L}\doteq \{(x_{1},\ldots ,x_{d})\in \mathbb{Z}^{d}:|x_{1}|,\ldots
,|x_{d}|\leq L\}\cap \mathfrak{L}\ ,\qquad L\in \mathbb{N}\ ,
\label{eq:def lambda n}
\end{equation}%
as a so-called van Hove sequence.

\subsection{The CAR Algebra}

For any $\Lambda \in \mathcal{P}_{f}$, $\mathcal{U}_{\Lambda }$ is the
finite dimensional unital $C^{\ast }$-algebra generated by elements $\{a_{x,%
\mathrm{s}}\}_{x\in \Lambda ,\mathrm{s}\in \mathrm{S}}$ satisfying the
canonical anti-commutation relations (CAR), $\mathrm{S}$ being some finite
set of spins: 
\begin{equation}
a_{x,\mathrm{s}}a_{y,\mathrm{t}}+a_{y,\mathrm{t}}a_{x,\mathrm{s}}=0\ ,\qquad
a_{x,\mathrm{s}}a_{y,\mathrm{t}}^{\ast }+a_{y,\mathrm{t}}^{\ast }a_{x,%
\mathrm{s}}=\delta _{\mathrm{s},\mathrm{t}}\delta _{x,y}\mathfrak{1}
\label{CARbis}
\end{equation}%
for all $x,y\in \mathfrak{L}$ and $\mathrm{s},\mathrm{t}\in \mathrm{S}$.
Observe that $(\mathcal{U}_{\Lambda _{L}})_{L\in \mathbb{N}}$ is an
increasing sequence of $C^{\ast }$-algebras. Hence, the set%
\begin{equation}
\mathcal{U}_{0}\doteq \bigcup_{L\in \mathbb{N}}\mathcal{U}_{\Lambda _{L}}
\label{simple}
\end{equation}%
of so-called local elements is a normed $\ast $-algebra with $\left\Vert
A\right\Vert _{\mathcal{U}_{0}}=\left\Vert A\right\Vert _{\mathcal{U}%
_{\Lambda _{L}}}$for all $A\in \mathcal{U}_{\Lambda _{L}}$ and $L\in \mathbb{%
N}$. The CAR $C^{\ast }$-algebra $\mathcal{U\equiv U}_{\mathfrak{L}}$ of the
full system, whose norm is denoted by $\Vert \cdot \Vert _{\mathcal{U}}$, is
by definition the completion of the normed $\ast $-algebra $\mathcal{U}_{0}$%
. It is separable, by finite dimensionality of $\mathcal{U}_{\Lambda }$ for $%
\Lambda \in \mathcal{P}_{f}$. Equivalently, the $C^{\ast }$-algebra $%
\mathcal{U}$ is the universal $C^{\ast }$-algebra \cite[Section II.8.3]%
{Blackadar} associated with the relations (\ref{CARbis}) for all $x,y\in 
\mathfrak{L}$ and $\mathrm{s},\mathrm{t}\in \mathrm{S}$. The (real) Banach
subspace of all self-adjoint elements of $\mathcal{U}$ is denoted by $%
\mathcal{U}^{\mathbb{R}}\varsubsetneq \mathcal{U}$.

\subsection{Important $\ast $-Automorphisms of the\ CAR Algebra\label{even}}

\noindent \underline{Parity:} Observe that the condition%
\begin{equation}
\sigma (a_{x,\mathrm{s}})=-a_{x,\mathrm{s}},\qquad x\in \Lambda ,\ \mathrm{s}%
\in \mathrm{S}\ ,  \label{automorphism gauge invariance}
\end{equation}%
defines a unique $\ast $-automorphism $\sigma $ of the $C^{\ast }$-algebra $%
\mathcal{U}$. Elements $A_{1},A_{2}\in \mathcal{U}$ satisfying $\sigma
(A_{1})=A_{1}$ and $\sigma (A_{2})=-A_{2}$ are respectively called \emph{even%
} and \emph{odd}. Note that the set%
\begin{equation}
\mathcal{U}^{+}\equiv \mathcal{U}_{\mathfrak{L}}^{+}\doteq \{A\in \mathcal{U}%
:A=\sigma (A)\}\subseteq \mathcal{U}  \label{definition of even operators}
\end{equation}%
of all even elements is a $\ast $-algebra. By continuity of $\sigma $, $%
\mathcal{U}^{+}$ is additionally norm-closed and, hence, a $C^{\ast }$%
-subalgebra of $\mathcal{U}$. Even elements are crucial here since, for any
finite subsets $\Lambda ^{(1)},\Lambda ^{(2)}\in \mathcal{P}_{f}$ with $%
\Lambda ^{(1)}\cap \Lambda ^{(2)}=\emptyset $, 
\begin{equation}
\left[ A_{1},A_{2}\right] \doteq A_{1}A_{2}-A_{2}A_{1}=0\ ,\qquad A_{1}\in 
\mathcal{U}_{\Lambda ^{(1)}}\cap \mathcal{U}^{+},\ A_{2}\in \mathcal{U}%
_{\Lambda ^{(2)}}\ .  \label{commutatorbis}
\end{equation}%
However, this relation is wrong in general when $A_{1}$ is not an even
element. For instance, the CAR (\ref{CARbis}) trivially yields $[a_{x,%
\mathrm{s}},a_{y,\mathrm{t}}]=2a_{x,\mathrm{s}}a_{y,\mathrm{t}}\neq 0$ for
any $x,y\in \mathfrak{L}$ and $\mathrm{s},\mathrm{t}\in \mathrm{S}$ such
that $(x,\mathrm{s})\neq (y,\mathrm{t})$.

The condition (\ref{commutatorbis}) is the expression of the local causality
in quantum field theory. Using well-known constructions\footnote{%
More precisely, the so-called sector theory of quantum field theory.}, the $%
C^{\ast }$-algebra $\mathcal{U}$, generated by anticommuting elements, can
be recovered from $\mathcal{U}^{+}$. As a consequence, the $C^{\ast }$%
-algebra $\mathcal{U}^{+}$ should thus be seen as more fundamental than $%
\mathcal{U}$ in Physics. In fact, $\mathcal{U}$ corresponds in this context
to the so-called local field algebra. See, e.g., \cite[Sections 4.8 and 6]%
{Araki-livre}.

Note that any element $A\in \mathcal{U}$ can be written as a sum of even and
odd elements, respectively denoted by $A^{+}$ and $A^{-}$:%
\begin{equation}
A=A^{+}+A^{-}\qquad \text{with}\qquad A^{\pm }\doteq \frac{1}{2}\left( A\pm
\sigma \left( A\right) \right) \ .  \label{decomposition even odd}
\end{equation}%
This directly follows from the fact that $\sigma $ is an involution (i.e., $%
\sigma \circ \sigma =\mathbf{1}_{\mathcal{U}}$). \bigskip

\noindent \underline{Translations:} $\mathfrak{L}=\mathbb{Z}^{d}$ is an
important case here, because invariant states under the action of the group $%
(\mathbb{Z}^{d},+)$ of lattice translations\ on $\mathcal{U\equiv U}_{%
\mathbb{Z}^{d}}$ are pivotal in the full description of macroscopic
long-range dynamics. The translations in $\mathcal{U}$ refer to the group
homomorphism $x\mapsto \alpha _{x}$ from $\mathbb{Z}^{d}$ to the group of $%
\ast $-automorphisms of $\mathcal{U}$, which is uniquely defined by the
condition%
\begin{equation}
\alpha _{x}(a_{y,\mathrm{s}})=a_{y+x,\mathrm{s}}\ ,\quad y\in \mathbb{Z}%
^{d},\;\mathrm{s}\in \mathrm{S}\ .  \label{transl}
\end{equation}%
This group homomorphism is used below to define the notion of (space)
periodicity of states of lattice-fermion systems (Section \ref{Even States}%
), as well as the notion of translation-invariance of interactions (Section %
\ref{Section Banach space interaction}).\bigskip

\noindent \underline{Permutations:} Let $\Pi $ be the set of all bijective
mappings from $\mathfrak{L}$ to itself which leave all but finitely many
elements invariant. It is a group with respect to the composition of
mappings. The condition 
\begin{equation}
\mathfrak{p}_{\pi }:a_{x,\mathrm{s}}\mapsto a_{\pi (x),\mathrm{s}},\quad
x\in \mathfrak{L},\;\mathrm{s}\in \mathrm{S}\ ,
\label{definition perm automorphism}
\end{equation}%
defines a group homomorphism $\pi \mapsto \mathfrak{p}_{\pi }$ from $\Pi $
to the group of $\ast $-automorphisms of $\mathcal{U}$.

\subsection{State Space}

For any $\Lambda \in \mathcal{P}_{f}$, we denote by%
\begin{equation}
E_{\Lambda }\doteq \{\rho _{\Lambda }\in \mathcal{U}_{\Lambda }^{\ast }:\rho
_{\Lambda }\geq 0,\ \rho _{\Lambda }(\mathfrak{1})=1\}=\{\rho _{\Lambda }\in 
\mathcal{U}_{\Lambda }^{\ast }:\Vert \rho _{\Lambda }\Vert _{\mathcal{U}%
_{\Lambda }^{\ast }}=\rho _{\Lambda }(\mathfrak{1})=1\}  \label{local states}
\end{equation}%
the set of all states on $\mathcal{U}_{\Lambda }$. By finite dimensionality
of $\mathcal{U}_{\Lambda }$ for $\Lambda \in \mathcal{P}_{f}$, the set $%
E_{\Lambda }$ is a norm-compact convex subset of the dual space $\mathcal{U}%
_{\Lambda }^{\ast }$. It is not a simplex, being affinely equivalent to the
set of states on a matrix algebra. It does not have of course a dense set of
extreme points. In comparison, the structure of the set of states for
infinite systems is more subtle.

The state space associated with $\mathcal{U}$ is defined by%
\begin{equation}
E\doteq \{\rho \in \mathcal{U}^{\ast }:\rho \geq 0,\ \rho (\mathfrak{1}%
)=1\}=\{\rho \in \mathcal{U}^{\ast }:\Vert \rho \Vert _{\mathcal{U}^{\ast
}}=\rho (\mathfrak{1})=1\}\ .  \label{states CAR}
\end{equation}%
In particular, if $\mathfrak{L}\in \mathcal{P}_{f}$ then $E=E_{\mathfrak{L}}$%
. In any case, for $\mathcal{U}$ is a separable Banach space, $E$ is a
metrizable and weak$^{\ast }$-compact convex subset of the dual space $%
\mathcal{U}^{\ast }$, by \cite[Theorems 3.15 and 3.16]{Rudin}. It is also
the state space of the classical dynamics \cite[Definition 2.1]%
{Bru-pedra-MF-I} introduced later on.

By the Krein-Milman theorem \cite[Theorem 3.23]{Rudin}, $E$ is the weak$%
^{\ast }$-closure of the convex hull of the (non-empty) set $\mathcal{E}(E)$
of its extreme points: 
\begin{equation}
E=\overline{\mathrm{co}\mathcal{E}\left( E\right) }\ .
\label{closure of the convex hull}
\end{equation}%
When $\mathfrak{L}$ is an infinite set, $\mathcal{U}$ is antiliminal and
simple, because it is a UHF (uniformly hyperfinite) algebra. See, e.g., \cite%
[Section 8]{Bru-pedra-MF-I}. Therefore, by \cite[Lemma 8.5]{Bru-pedra-MF-I},
the state space $E$ has a weak$^{\ast }$-dense subset of extreme points:%
\begin{equation}
E=\overline{\mathcal{E}(E)}\ .  \label{density extreme states}
\end{equation}%
This fact is well-known and already discussed in \cite[p. 226]{Bratteliquasi}%
. As a matter of fact, the property of a convex weak$^{\ast }$-compact set
having a weak$^{\ast }$-dense set of extreme points is not accidental, but 
\emph{generic}\footnote{%
More precisely, by \cite[Theorem 2.4]{Bru-pedra-MF-I}, the set of all convex
weak$^{\ast }$-compact subsets of the dual $\mathcal{X}^{\ast }$ of an
infinite-dimensional separable Banach space $\mathcal{X}$ is a weak$^{\ast }$%
-Hausdorff-dense $G_{\delta }$ subset of the space of convex weak$^{\ast }$%
-compact subsets of $\mathcal{X}^{\ast }$.} in infinite dimension, by \cite[%
Theorem 2.4]{Bru-pedra-MF-I}.

For any $\Lambda \in \mathcal{P}_{f}$, the restriction of a state $\rho \in
E $ to $\mathcal{U}_{\Lambda }$ yields a state in $E_{\Lambda }$, always
denoted by $\rho _{\Lambda }\in E_{\Lambda }$. Conversely, for any $\Lambda
\in \mathcal{P}_{f}$, a state $\rho _{\Lambda }\in E_{\Lambda }$ can be seen
as a state $\tilde{\rho}$ on $\mathcal{U}$ by setting\footnote{%
(\ref{Gibbs state H_nbis}) is well defined because of \cite[Theorem 11.2]%
{Araki-Moriya}, the tracial state $\mathrm{tr}$ being even.} 
\begin{equation}
\tilde{\rho}(AB)=\rho _{\Lambda }(A)\mathrm{tr}(B)\ ,\text{\qquad }A\in 
\mathcal{U}_{\Lambda },\;B\in \mathcal{U}_{\mathcal{Z}},\ \mathcal{Z}\in 
\mathcal{P}_{f},\ \mathcal{Z}\cap \Lambda =\emptyset \ ,
\label{Gibbs state H_nbis}
\end{equation}%
where $\mathrm{tr}$ is the unique tracial state on $\mathcal{U}$, also named
normalized trace, see \cite[Section 4.2]{Araki-Moriya}. In particular,
similar to (\ref{simple}), the set 
\begin{equation*}
E_{0}\doteq \bigcup_{L\in \mathbb{N}}E_{\Lambda _{L}}
\end{equation*}%
can be seen as a weak$^{\ast }$-dense subset of $E$, by density of $\mathcal{%
U}_{0}\subseteq \mathcal{U}$.

\subsection{Even States\label{Even States}}

As explained below Equation (\ref{definition of even operators}), recall
that the $C^{\ast }$-algebra $\mathcal{U}^{+}$ should thus be seen as more
fundamental than $\mathcal{U}$ in Physics, because of the local causality in
quantum field theory. As a consequence, states on the $C^{\ast }$-algebra $%
\mathcal{U}^{+}$ should be seen as being the physically relevant ones. As it
is explicitly shown in the proof of Proposition \ref{density of periodic
states copy(3)}, the set of states on $\mathcal{U}^{+}$ is canonically
identified with the set of \emph{even} states on $\mathcal{U}$, defined by%
\begin{equation}
E^{+}\doteq \left\{ \rho \in E:\rho \circ \sigma =\rho \right\} \ ,
\label{gauge invariant states}
\end{equation}%
$\sigma $ being the unique $\ast $-automorphism of $\mathcal{U}$ satisfying (%
\ref{automorphism gauge invariance}). $E^{+}$ is again a metrizable and weak$%
^{\ast }$-compact convex set and has a (non-empty) set $\mathcal{E}(E^{+})$
of extreme points such that 
\begin{equation*}
E^{+}=\overline{\mathrm{co}\mathcal{E}\left( E^{+}\right) }\ ,
\end{equation*}%
by the Krein-Milman theorem \cite[Theorem 3.23]{Rudin}. Like for the space $%
E $ of \emph{all} states (cf. (\ref{density extreme states})), if $\mathfrak{%
L} $ is an infinite set then $\mathcal{E}(E^{+})$ is a \emph{dense} subset
of the set of all even states:

\begin{proposition}[Weak$^{\ast }$-density of extreme even states]
\label{density of periodic states copy(3)}\mbox{
}\newline
Let $\mathfrak{L}\subseteq \mathbb{Z}^{d}$ be an infinite set. The set $%
\mathcal{E}(E^{+})$ of extreme points of $E^{+}$ is a weak$^{\ast }$-dense
subset of $E^{+}$, i.e., $E^{+}=\overline{\mathcal{E}(E^{+})}$.
\end{proposition}

\begin{proof}
Let $\tilde{E}^{+}$ be the set of all states on $\mathcal{U}^{+}$.
Obviously, $E^{+}\subseteq \tilde{E}^{+}$ by seeing any state on $E$\ as a
state on $\mathcal{U}^{+}$, by restriction. Conversely, for any state $%
\tilde{\rho}\in \tilde{E}^{+}$, we define the linear functional 
\begin{equation*}
\rho \doteq \tilde{\rho}\circ \left( \frac{\sigma +\mathbf{1}_{\mathcal{U}}}{%
2}\right) \in \mathcal{U}^{\ast }\ .
\end{equation*}%
Note that $\Vert \sigma \Vert _{\mathcal{B}(\mathcal{U}^{\ast })}=1$ because 
$\sigma $\ is a $\ast $-automorphism of a $C^{\ast }$-algebra. Since $\rho (%
\mathfrak{1})=1$ and $\Vert \rho \Vert _{\mathcal{U}^{\ast }}\leq 1$, we
deduce from \cite[Proposition 2.3.11]{BrattelliRobinsonI} that $\rho \in E$.
By construction, the restriction of $\rho $ to $\mathcal{U}^{+}$ is $\tilde{%
\rho}$ while $\rho (A)=0$ for all odd elements $A\in \mathcal{U}$. In other
words, $\rho \in E^{+}$ and one can identify $\tilde{E}^{+}$ with $E^{+}$.
It easy to check that the mapping $\tilde{\rho}\mapsto \rho $ is an affine
weak$^{\ast }$ homeomorphism from $\tilde{E}^{+}$ to $E^{+}$. By \cite%
{stormer}, if $\mathfrak{L}$ is an infinite set then $\mathcal{U}^{+}$ is $%
\ast $-isomorphic to $\mathcal{U}$. Thus, by (\ref{density extreme states}),
the assertion follows.
\end{proof}

\begin{remark}[States of quantum-spin systems]
\label{density of periodic states copy(1)}\mbox{
}\newline
As already mentioned, all results presented here can be extended to
quantum-spin systems on a lattice, the corresponding $C^{\ast }$-algebra
being an infinite tensor product of finite-dimensional matrix algebras
attached to each lattice site $x\in \mathfrak{L}=\mathbb{Z}^{d}$. In this
case, there is no even property to consider, which corresponds to take $%
E^{+}=E$ in all the discussions.
\end{remark}

Important examples of even states are the so-called periodic\emph{\ }states:
Let $\mathfrak{L}=\mathbb{Z}^{d}$. Consider the sub-groups $(\mathbb{Z}_{%
\vec{\ell}}^{d},+)\subseteq (\mathbb{Z}^{d},+)$, $\vec{\ell}\in \mathbb{N}%
^{d}$, where%
\begin{equation*}
\mathbb{Z}_{\vec{\ell}}^{d}\doteq \ell _{1}\mathbb{Z}\times \cdots \times
\ell _{d}\mathbb{Z}\ \ .
\end{equation*}%
Any state $\rho \in E$ satisfying $\rho \circ \alpha _{x}=\rho $ for all $%
x\in \mathbb{Z}_{\vec{\ell}}^{d}$ is called $\mathbb{Z}_{\vec{\ell}}^{d}$%
\emph{-invariant} on $\mathcal{U}$ or $\vec{\ell}$\emph{-periodic}, $\alpha
_{x}$ being the unique $\ast $-automorphism of $\mathcal{U}$ satisfying (\ref%
{transl}). Translation-invariant states refer to $(1,\cdots ,1)$-periodic
states. The set of all periodic states is denoted by%
\begin{equation}
E_{\mathrm{p}}\doteq \bigcup_{\vec{\ell}\in \mathbb{N}^{d}}E_{\vec{\ell}}\ ,
\label{set of periodic states}
\end{equation}%
where, for any $\vec{\ell}\in \mathbb{N}^{d}$, 
\begin{equation}
E_{\vec{\ell}}\doteq \left\{ \rho \in E:\rho \circ \alpha _{x}=\rho \quad 
\text{for}\ \text{all}\ x\in \mathbb{Z}_{\vec{\ell}}^{d}\right\} \ .
\label{periodic invariant states}
\end{equation}%
By \cite[Lemma 1.8]{BruPedra2}, periodic states are even and form a weak$%
^{\ast }$-dense subset of even states:

\begin{proposition}[Weak$^{\ast }$-density of periodic states]
\label{density of periodic states}\mbox{
}\newline
Let $\mathfrak{L}=\mathbb{Z}^{d}$. The set $E_{\mathrm{p}}$ of periodic
states is a weak$^{\ast }$-dense set of $E^{+}$, i.e., $E^{+}=\overline{E_{%
\mathrm{p}}}$.
\end{proposition}

\begin{proof}
For any $\rho \in E^{+}$ and $n\in \mathbb{N}$, we define the state $\tilde{%
\rho}_{n}$ to be some $(2n+1,\ldots ,2n+1)$-periodic state for which $(%
\tilde{\rho}_{n})_{\Lambda _{n}}=\rho _{\Lambda _{n}}$ in the cubic box $%
\Lambda _{n}$ (\ref{eq:def lambda n}). Such a state always exists because of 
\cite[Theorem 11.2]{Araki-Moriya}, since $\rho $ is, by definition, even.
Clearly, $\{\tilde{\rho}_{n}\}_{n\in \mathbb{N}}$ converges, as $%
n\rightarrow \infty $, towards $\rho \in E^{+}$ with respect to the weak$%
^{\ast }$-topology, by density of $\mathcal{U}_{0}\subseteq \mathcal{U}$.
\end{proof}

The sets $E_{\vec{\ell}}$, $\vec{\ell}\in \mathbb{N}^{d}$, of $\vec{\ell}$%
-periodic states all share the \emph{same} peculiar geometrical structure:
Like the set $E$ of all states for $\mathfrak{L}=\mathbb{Z}^{d}$, they are
metrizable and weak$^{\ast }$-compact convex sets with a weak$^{\ast }$%
-dense set $\mathcal{E}(E_{\vec{\ell}})$ of extreme points: 
\begin{equation}
E_{\vec{\ell}}=\overline{\mathrm{co}\mathcal{E}(E_{\vec{\ell}})}=\overline{%
\mathcal{E}(E_{\vec{\ell}})}\ ,\qquad \vec{\ell}\in \mathbb{N}^{d}\ .
\label{cov heull l perio}
\end{equation}%
Compare with Equations (\ref{closure of the convex hull})-(\ref{density
extreme states}). In fact, up to an affine homeomorphism, for any $\vec{\ell}%
\in \mathbb{N}^{d}$, $E_{\vec{\ell}}$ is the so-called \emph{Poulsen simplex}
\cite[Theorem 1.12]{BruPedra2}. This property is well-known and also holds
true for quantum-spin systems \cite[p. 405-406, 464]{BrattelliRobinsonI}.
The fact that all $E_{\vec{\ell}}$, $\vec{\ell}\in \mathbb{N}^{d}$, have the
same topological structure is not so surprising since, for any fixed $\vec{%
\ell}\in \mathbb{N}^{d}$, we can redefine the spin set $\mathrm{S}=\mathrm{S}%
_{\vec{\ell}}$ and, as a consequence, the CAR algebra $\mathcal{U}=\mathcal{U%
}_{\vec{\ell}}$ to see any $\vec{\ell}$-periodic state $\rho \in E_{\vec{\ell%
}}$ as a translation-invariant state on the new CAR algebra $\mathcal{U}_{%
\vec{\ell}}$. By \cite[Proposition 6.14]{Bru-pedra-MF-I}, $E^{+}$ can also
be seen as the \emph{weak}$^{\ast }$\emph{-Hausdorff limit} of the
increasing sequence $\{E_{\vec{\ell}}\}_{\vec{\ell}\in \mathbb{N}^{d}}$ of
weak$^{\ast }$-compact \emph{simplices} with weak$^{\ast }$-dense set of
extreme points.

By \cite[Theorem 1.16]{BruPedra2} note that the set of all ($\vec{\ell}$-) 
\emph{ergodic} states, as defined by \cite[Definition 1.15]{BruPedra2} for
any $\vec{\ell}\in \mathbb{N}^{d}$, is equal to 
\begin{equation}
\mathcal{E}_{\mathrm{p}}\doteq \bigcup_{\vec{\ell}\in \mathbb{N}^{d}}%
\mathcal{E}\left( E_{\vec{\ell}}\right) \subseteq E_{\mathrm{p}}\subseteq
E^{+}\ .  \label{ergodic states}
\end{equation}%
By Proposition \ref{density of periodic states}, $\mathcal{E}_{\mathrm{p}}$
is also a weak$^{\ast }$-dense set of $E^{+}$.

The set $E_{\mathrm{p}}$ is important because, in all cyclic representations
of $\mathcal{U}$ associated with a state $\rho \in E_{\mathrm{p}}$, the
infinite-volume dynamics of interacting lattice fermions with long-range
interactions exists. The ergodicity of states of $\mathcal{E}_{\mathrm{p}}$
plays a crucial r\^{o}le in this context.

Instead of the set $E_{\mathrm{p}}$ of periodic states, previous studies
extracting classical dynamics for lattice-fermion and quantum-spin systems
with long-range interactions (cf. \cite%
{T1,T2,Hepp-Lieb73,Bona75,Sewell83,Rieckers84,Morchio87,Bona87,Duffner-Rieckers88,Bona88,Bona89,Bona90,Unnerstall90,Unnerstall90b,Unnerstall90-open,Bona91,Duffield1991,BagarelloMorchio92,Duffield-Werner1,Duffield-Werner2,Duffield-Werner3,Duffield-Werner4,Duffield-Werner5}%
) use classical flows within the subset 
\begin{equation}
E_{\Pi }\doteq \left\{ \rho \in E:\rho \circ \mathfrak{p}_{\pi }=\rho \quad 
\text{for}\ \text{all}\ \pi \in \Pi \right\}  \label{permutation inv states}
\end{equation}%
of \emph{permutation-invariant} states, $\mathfrak{p}_{\pi }$ being the
unique $\ast $-automorphism of $\mathcal{U}$ satisfying (\ref{definition
perm automorphism}). This set has a much simpler structure than the set $E_{%
\mathrm{p}}$ of periodic states:%
\begin{equation*}
E_{\Pi }\subseteq \bigcap_{\vec{\ell}\in \mathbb{N}^{d}}E_{\vec{\ell}%
}\subseteq E_{\mathrm{p}}\ .
\end{equation*}%
Compare this assertion with (\ref{set of periodic states}). $E_{\Pi }$ is a
closed metrizable face of $E_{\vec{\ell}}$ for all $\vec{\ell}\in \mathbb{N}%
^{d}$ and a Bauer simplex, i.e., a (Choquet) simplex whose set of extreme
points is closed. Indeed, by St{\o }rmer's theorem \cite[Theorem 5.2]%
{BruPedra2}, extreme points of $E_{\Pi }$ are so-called \emph{product states}%
. See \cite[Section 5.1]{BruPedra2} for more details on
permutation-invariant states.

\section{Lattice Fermions with Short-Range Interactions\label{Section FERMI}}

\subsection{Banach Spaces of Short-Range Interactions\label{Section Banach
space interaction}}

A (complex) \emph{interaction} is a mapping $\Phi :\mathcal{P}%
_{f}\rightarrow \mathcal{U}^{+}$ such that $\Phi _{\Lambda }\in \mathcal{U}%
_{\Lambda }$ for any $\Lambda \in \mathcal{P}_{f}$. The set of all
interactions can be naturally endowed with the structure of a complex vector
space as follows: 
\begin{equation}
(\Phi +\tilde{\Phi})_{\Lambda }\doteq \Phi _{\Lambda }+\tilde{\Phi}_{\Lambda
}\text{\qquad and\qquad }(\lambda \Phi )_{\Lambda }\doteq \lambda \Phi
_{\Lambda }  \label{vector space interaction}
\end{equation}%
for all interactions $\Phi ,\tilde{\Phi}$\ and $\lambda \in \mathbb{C}$. The
mapping%
\begin{equation}
\Phi \mapsto \Phi ^{\ast }\doteq (\Phi _{\Lambda }^{\ast })_{\Lambda \in 
\mathcal{P}_{f}}  \label{involution}
\end{equation}%
is a natural involution on the vector space of interactions. Self-adjoint
interactions are interactions $\Phi $ satisfying $\Phi =\Phi ^{\ast }$. The
set of all self-adjoint interactions forms a real subspace of the space of
all interactions. Any interaction $\Phi $ can be decomposed into its real
and imaginary parts, which are self-adjoint interactions respectively
defined by 
\begin{equation}
\mathrm{Re}\left\{ \Phi \right\} \doteq \frac{1}{2}\left( \Phi ^{\ast }+\Phi
\right) \qquad \text{and}\qquad \mathrm{Im}\left\{ \Phi \right\} \doteq 
\frac{1}{2i}\left( \Phi -\Phi ^{\ast }\right) \ .
\label{real-im part interaction}
\end{equation}

We now define a Banach space $\mathcal{W}$ of short-range interactions by
introducing a norm for interactions that take into account their spacial
decay. To this end, we use a positive-valued symmetric function $\mathbf{F}:%
\mathfrak{L}^{2}\rightarrow (0,1]$ with maximum value $\mathbf{F}\left(
x,x\right) =1$ for all $x\in \mathfrak{L}$. Like for instance in \cite%
{brupedraLR}, we impose the following conditions on $\mathbf{F}$:

\begin{itemize}
\item \emph{Summability on }$\mathfrak{L}$\emph{.} 
\begin{equation}
\left\Vert \mathbf{F}\right\Vert _{1,\mathfrak{L}}\doteq \underset{y\in 
\mathfrak{L}}{\sup }\sum_{x\in \mathfrak{L}}\mathbf{F}\left( x,y\right) \in %
\left[ 1,\infty \right) \ .  \label{(3.1) NS}
\end{equation}

\item \emph{Bounded convolution constant.} 
\begin{equation}
\mathbf{D}\doteq \underset{x,y\in \mathfrak{L}}{\sup }\sum_{z\in \mathfrak{L}%
}\frac{\mathbf{F}\left( x,z\right) \mathbf{F}\left( z,y\right) }{\mathbf{F}%
\left( x,y\right) }<\infty \ .  \label{(3.2) NS}
\end{equation}
\end{itemize}

\noindent Examples of functions $\mathbf{F}:\mathfrak{L}^{2}\rightarrow
(0,1] $ satisfying (\ref{(3.1) NS})-(\ref{(3.2) NS}) for any lattice $%
\mathfrak{L}\subseteq \mathbb{Z}^{d}$ ($d\in \mathbb{N}$) are given by 
\begin{equation}
\mathbf{F}\left( x,y\right) =\left( 1+\left\vert x-y\right\vert \right)
^{-(d+\epsilon )}\qquad \text{or}\qquad \mathbf{F}\left( x,y\right) =\mathrm{%
e}^{-\varsigma \left\vert x-y\right\vert }(1+\left\vert x-y\right\vert
)^{-(d+\epsilon )}  \label{examples}
\end{equation}%
for every $\varsigma ,\epsilon \in \mathbb{R}^{+}$. In all the paper, (\ref%
{(3.1) NS})-(\ref{(3.2) NS}) are assumed to be satisfied.

Then, a norm for interactions $\Phi $ is defined by%
\begin{equation}
\left\Vert \Phi \right\Vert _{\mathcal{W}}\doteq \underset{x,y\in \mathfrak{L%
}}{\sup }\sum\limits_{\Lambda \in \mathcal{P}_{f},\;\Lambda \supseteq
\{x,y\}}\frac{\Vert \Phi _{\Lambda }\Vert _{\mathcal{U}}}{\mathbf{F}\left(
x,y\right) }  \label{iteration0}
\end{equation}%
and 
\begin{equation*}
\mathcal{W\equiv W}^{(\mathfrak{L})}\equiv (\mathcal{W},\left\Vert \cdot
\right\Vert _{\mathcal{W}})
\end{equation*}%
denotes the separable Banach space of interactions $\Phi $ satisfying $%
\left\Vert \Phi \right\Vert _{\mathcal{W}}<\infty $. Elements $\Phi \in 
\mathcal{W}$ are named \emph{short-range} interactions on $\mathfrak{L}%
\subseteq \mathbb{Z}^{d}$. The (real) Banach subspace of all self-adjoint
interactions is denoted by $\mathcal{W}^{\mathbb{R}}\varsubsetneq \mathcal{W}
$, similar to $\mathcal{U}^{\mathbb{R}}\varsubsetneq \mathcal{U}$.

By definition, an interaction $\Phi $ on $\mathfrak{L}=\mathbb{Z}^{d}$ is
translation-invariant\ if, for all $x\in \mathbb{Z}^{d}$ and $\Lambda \in 
\mathcal{P}_{f}$, $\Phi _{\Lambda +x}=\alpha _{x}(\Phi _{\Lambda })$, where%
\begin{equation}
\Lambda +x\doteq \left\{ y+x\in \mathbb{Z}^{d}:y\in \Lambda \right\} \ .
\label{translation box}
\end{equation}%
Recall that $\{\alpha _{x}\}_{x\in \mathbb{Z}^{d}}$ is the family of
(translation) $\ast $-automorphisms of $\mathcal{U}$ defined by (\ref{transl}%
). We denote by $\mathcal{W}_{1}\varsubsetneq \mathcal{W}$ the (separable)
Banach subspace of translation-invariant, short-range interactions on $%
\mathfrak{L}=\mathbb{Z}^{d}$.

For any $\Phi \in \mathcal{W}$ and $\vec{\ell}\in \mathbb{N}^{d}$, we define
the even observable%
\begin{equation}
\mathfrak{e}_{\Phi ,\vec{\ell}}\doteq \frac{1}{\ell _{1}\cdots \ell _{d}}%
\sum\limits_{x=(x_{1},\ldots ,x_{d}),\;x_{i}\in \{0,\ldots ,\ell
_{i}-1\}}\sum\limits_{\mathcal{Z}\in \mathcal{P}_{f},\;\mathcal{Z}\ni x}%
\frac{\Phi _{\mathcal{Z}}}{\left\vert \mathcal{Z}\right\vert }\ .
\label{eq:enpersite}
\end{equation}%
From (\ref{(3.1) NS}) and (\ref{iteration0}), note that%
\begin{equation}
\Vert \mathfrak{e}_{\Phi ,\vec{\ell}}\Vert _{\mathcal{U}}\leq \left\Vert 
\mathbf{F}\right\Vert _{1,\mathfrak{L}}\left\Vert \Phi \right\Vert _{%
\mathcal{W}}\ ,\qquad \Phi \in \mathcal{W},\ \vec{\ell}\in \mathbb{N}^{d}.
\label{e phi}
\end{equation}%
For any self-adjoint, translation-invariant interaction $\Phi $, they refer
to the energy density observables of \cite[Eq. (1.16)]{BruPedra2}. See also
Proposition \ref{density of periodic states copy(2)} below. Additionally,
for any $\Lambda \in \mathcal{P}_{f}$, we define the closed subspaces%
\footnote{%
This follows from the continuity and linearity of the mappings $\Phi \mapsto
\Phi _{\mathcal{Z}}$ for all $\mathcal{Z}\in \mathcal{P}_{f}$.}%
\begin{equation}
\mathcal{W}_{\Lambda }\doteq \left\{ \Phi \in \mathcal{W}_{1}:\Phi _{%
\mathcal{Z}}=0\text{ whenever }\mathcal{Z}\nsubseteq \Lambda \text{, }%
\mathcal{Z}\ni 0\right\}  \label{eq:enpersitebis}
\end{equation}%
of interactions that are both finite-range and translation-invariant. Note
that, for any $\Lambda \in \mathcal{P}_{f}$ and $\vec{\ell}\in \mathbb{N}%
^{d} $, 
\begin{equation}
\mathcal{W}_{\Lambda }\subseteq \left\{ \Phi \in \mathcal{W}_{1}:\mathfrak{e}%
_{\Phi ,\vec{\ell}}\in \mathcal{U}_{\Lambda ^{(\vec{\ell})}}\right\}
\subseteq \mathcal{W}_{1}\ ,  \label{eq:enpersitebisbis}
\end{equation}%
where 
\begin{equation}
\Lambda ^{(\vec{\ell})}\doteq \cup \left\{ \Lambda +x:x=(x_{1},\ldots
,x_{d}),\;x_{i}\in \{0,\ldots ,\ell _{i}-1\}\right\} \in \mathcal{P}_{f}\ .
\label{eq:enpersitebisbisbis}
\end{equation}%
Recall that, for $\Lambda \in \mathcal{P}_{f}$, $\mathcal{U}_{\Lambda
}\subseteq \mathcal{U}$ is the finite-dimensional unital $C^{\ast }$-algebra
generated by elements $\{a_{x,\mathrm{s}}\}_{x\in \Lambda ,\mathrm{s}\in 
\mathrm{S}}$ satisfying the CAR (\ref{CARbis}). Similar to (\ref{simple}),%
\begin{equation}
\mathcal{W}_{0}\doteq \bigcup_{L\in \mathbb{N}}\mathcal{W}_{\Lambda
_{L}}\subseteq \mathcal{W}_{1}  \label{W0}
\end{equation}%
is a dense subspace of $\mathcal{W}_{1}$. Elements of $\mathcal{W}_{0}$ are 
\emph{finite-range}, translation-invariant interactions.

\subsection{Local Energy Elements}

We define a sequence of local elements associated with any complex
interaction $\Phi \in \mathcal{W}$ as follows:

\begin{definition}[Local energy]
\label{definition fininte vol dynam0}\mbox{ }\newline
The local energy elements associated with a complex interaction $\Phi \in 
\mathcal{W}$ are given by%
\begin{equation*}
U_{L}^{\Phi }\doteq \sum\limits_{\Lambda \subseteq \Lambda _{L}}\Phi
_{\Lambda }\in \mathcal{U}_{\Lambda _{L}}\cap \mathcal{U}^{+}\ ,\qquad L\in 
\mathbb{N}\ .
\end{equation*}
\end{definition}

\noindent If $\Phi \in \mathcal{W}^{\mathbb{R}}$, then $(U_{L}^{\Phi
})_{L\in \mathbb{N}}\in \mathcal{U}^{\mathbb{R}}$ and so, Definition \ref%
{definition fininte vol dynam0} yields a sequence of local Hamiltonians,
which are used to generate finite-volume dynamics.

By straightforward estimates using (\ref{(3.1) NS}) and (\ref{iteration0}),
as in (\ref{e phi}), note that, for any complex interactions $\Phi ,\Psi \in 
\mathcal{W}$, 
\begin{equation}
\left\Vert U_{L}^{\Phi }-U_{L}^{\Psi }\right\Vert _{\mathcal{U}}=\left\Vert
U_{L}^{\Phi -\Psi }\right\Vert _{\mathcal{U}}\leq \left\vert \Lambda
_{L}\right\vert \left\Vert \mathbf{F}\right\Vert _{1,\mathfrak{L}}\left\Vert
\Phi -\Psi \right\Vert _{\mathcal{W}}\ ,\qquad L\in \mathbb{N}\ .
\label{norm Uphi}
\end{equation}%
In addition, similar to \cite[Lemma 1.32]{BruPedra2}, local energy elements
yield energy densities for translation-invariant short-range interactions:

\begin{proposition}[Energy density of periodic states]
\label{density of periodic states copy(2)}\mbox{
}\newline
For any $\vec{\ell}\in \mathbb{N}^{d}$, $\vec{\ell}$-periodic state $\rho
\in E_{\vec{\ell}}$ (\ref{periodic invariant states}) and
translation-invariant complex interaction$\ \Phi \in \mathcal{W}_{1}$, 
\begin{equation*}
\lim\limits_{L\rightarrow \infty }\frac{\rho \left( U_{L}^{\Phi }\right) }{%
\left\vert \Lambda _{L}\right\vert }=\rho (\mathfrak{e}_{\Phi ,\vec{\ell}})
\end{equation*}%
with $\mathfrak{e}_{\Phi ,\vec{\ell}}$ being the even observable defined by (%
\ref{eq:enpersite}).
\end{proposition}

\begin{proof}
All arguments of the proof can be found in \cite[Lemma 4.17]{BruPedra2}.
\end{proof}

By Proposition \ref{density of periodic states copy(2)}, the energy density%
\begin{equation*}
e_{\Phi }\left( \rho \right) \doteq \lim\limits_{L\rightarrow \infty }\frac{%
\rho \left( U_{L}^{\Phi }\right) }{\left\vert \Lambda _{L}\right\vert }\
,\qquad \Phi \in \mathcal{W}_{1}\ ,
\end{equation*}%
exists for all periodic states, which form a weak$^{\ast }$-dense subset $E_{%
\mathrm{p}}$ of the set $E^{+}$ of even states, by Proposition \ref{density
of periodic states}. Note that there are infinitely many (uncountable) other
states with this property. Examples of non-periodic states for which $%
e_{\Phi }(\rho )$ exists can be constructed by using KMS\ states associated
with \emph{random} interactions together with the Akcoglu-Krengel ergodic
theorem \cite{birkoff}. See, e.g., \cite{OhmVI}. For all $\rho \in E$ such
that $e_{\Phi }\left( \rho \right) $ exists, observe from Inequality (\ref%
{norm Uphi}) that 
\begin{equation}
\left\vert e_{\Phi }\left( \rho \right) -e_{\Psi }\left( \rho \right)
\right\vert \leq \left\Vert \mathbf{F}\right\Vert _{1,\mathfrak{L}%
}\left\Vert \Phi -\Psi \right\Vert _{\mathcal{W}}\ ,\qquad \Phi ,\Psi \in 
\mathcal{W}_{1}\ .  \label{inequality a la con}
\end{equation}

\subsection{Derivations on the CAR Algebra\label{sect Lieb--Robinson copy(1)}%
}

For any short-range interaction $\Phi \in \mathcal{W}$, the elements of
Definition \ref{definition fininte vol dynam0} determine a sequence of
bounded operators on $\mathcal{U}$:

\begin{definition}[Derivations on the CAR\ algebra for short-range
interactions]
\label{dynamic series}\mbox{ }\newline
The derivations $\{\delta _{L}^{\Phi }\}_{L\in \mathbb{N}}\subseteq \mathcal{%
B}(\mathcal{U})$ associated with any interaction $\Phi \in \mathcal{W}$ are
defined by%
\begin{equation*}
\delta _{L}^{\Phi }(A)\doteq i\left[ U_{L}^{\Phi },A\right] \doteq i\left(
U_{L}^{\Phi }A-AU_{L}^{\Phi }\right) \ ,\qquad A\in \mathcal{U},\ L\in 
\mathbb{N}\ .
\end{equation*}
\end{definition}

\noindent They are (bounded) derivations on $\mathcal{U}$ since, for any $%
\Phi \in \mathcal{W}$ and $L\in \mathbb{N}$, 
\begin{equation}
\delta _{L}^{\Phi }\left( AB\right) =\delta _{L}^{\Phi }\left( A\right)
B+A\delta _{L}^{\Phi }\left( B\right) \ ,\qquad A,B\in \mathcal{U}\ .
\label{derivation}
\end{equation}%
They are symmetric (or $\ast $-derivations) when $\Phi \in \mathcal{W}^{%
\mathbb{R}}$: 
\begin{equation}
\delta _{L}^{\Phi }\left( A\right) ^{\ast }=\delta _{L}^{\Phi }(A^{\ast })\
,\qquad A\in \mathcal{U},\ L\in \mathbb{N}\ .  \label{derivationbis}
\end{equation}

In the thermodynamic limit $L\rightarrow \infty $, this sequence of
derivations leads to a limit derivation defined on the (dense) subset $%
\mathcal{U}_{0}$ (\ref{simple}) of local elements:

\begin{proposition}[Strong convergence of finite-volume derivations]
\label{Lemma cigare0}\mbox{
}\newline
For any $\Phi \in \mathcal{W}$, $\Lambda \in \mathcal{P}_{f}$, $A\in 
\mathcal{U}_{\Lambda }$ and $L_{2}\geq L_{1}\geq 0$, 
\begin{equation*}
\left\Vert \delta _{L_{2}}^{\Phi }\left( A\right) -\delta _{L_{1}}^{\Phi
}\left( A\right) \right\Vert _{\mathcal{U}}\leq 2\left\vert \Lambda
\right\vert \left\Vert A\right\Vert _{\mathcal{U}}\left\Vert \Phi
\right\Vert _{\mathcal{W}}\sup_{y\in \Lambda }\sum\limits_{x\in \Lambda
_{L_{1}}^{c}}\mathbf{F}\left( x,y\right)
\end{equation*}%
and%
\begin{equation*}
\sup_{L\in \mathbb{N}}\left\Vert \delta _{L}^{\Phi }\left( A\right)
\right\Vert _{\mathcal{U}}\leq 2\left\vert \Lambda \right\vert \left\Vert
A\right\Vert _{\mathcal{U}}\left\Vert \Phi \right\Vert _{\mathcal{W}%
}\left\Vert \mathbf{F}\right\Vert _{1,\mathfrak{L}}\ .
\end{equation*}
\end{proposition}

\begin{proof}
Fixing all parameters of the proposition, we straightforwardly get the
estimate%
\begin{eqnarray*}
\left\Vert \delta _{L_{2}}^{\Phi }\left( A\right) -\delta _{L_{1}}^{\Phi
}\left( A\right) \right\Vert _{\mathcal{U}} &\leq &\sum_{y\in \Lambda
}\sum_{x\in \Lambda _{L_{1}}^{c}}\sum\limits_{\mathcal{Z}\in \mathcal{P}%
_{f},\ \mathcal{Z}\supseteq \{x,y\}}\left\Vert \left[ \Phi _{\mathcal{Z}},A%
\right] \right\Vert _{\mathcal{U}} \\
&\leq &2\left\Vert A\right\Vert _{\mathcal{U}}\left\Vert \Phi \right\Vert _{%
\mathcal{W}}\sum_{y\in \Lambda }\sum_{x\in \Lambda _{L_{1}}^{c}}\mathbf{F}%
\left( x,y\right) \ ,
\end{eqnarray*}%
which implies the first assertion. The second assertion is even simpler to
prove. We omit the details.
\end{proof}

\begin{corollary}[Generators of infinite-volume short-range dynamics]
\label{Lemma cigare1}\mbox{
}\newline
For any $\Phi \in \mathcal{W}$ and $A\in \mathcal{U}_{0}$, the limit%
\begin{equation*}
\delta ^{\Phi }\left( A\right) =\lim_{L\rightarrow \infty }\delta _{L}^{\Phi
}\left( A\right)
\end{equation*}%
exists and defines a (densely defined) derivation $\delta ^{\Phi }$ from $%
\mathcal{U}_{0}\subseteq \mathcal{U}$ to $\mathcal{U}$ satisfying the
following bound:%
\begin{equation*}
\left\Vert \delta ^{\Phi }\left( A\right) \right\Vert _{\mathcal{U}}\leq
2\left\vert \Lambda \right\vert \left\Vert A\right\Vert _{\mathcal{U}%
}\left\Vert \Phi \right\Vert _{\mathcal{W}}\left\Vert \mathbf{F}\right\Vert
_{1,\mathfrak{L}}\ ,\qquad A\in \mathcal{U}_{\Lambda },\ \Lambda \in 
\mathcal{P}_{f}\ .
\end{equation*}%
Additionally, $\delta ^{\Phi }$ is symmetric (or a $\ast $-derivation)
whenever $\Phi \in \mathcal{W}^{\mathbb{R}}$.
\end{corollary}

\begin{proof}
Combine Proposition \ref{Lemma cigare0} with Equations (\ref{derivation})-(%
\ref{derivationbis}) and the completeness of $\mathcal{U}$.
\end{proof}

\begin{remark}[Closure of limit derivations]
\label{Remark Closure of limit derivations}\mbox{
}\newline
If $\Phi \in \mathcal{W}^{\mathbb{R}}$ then the symmetric derivation $\delta
^{\Phi }$ is (norm-) closable \cite[Lemma 4.6]{brupedraLR}. It is proven
from its dissipativity \cite[Definition 1.4.6, Proposition 1.4.7]%
{Bratelli-derivation}, which is, in turn, deduced from \cite[Theorem 1.4.9]%
{Bratelli-derivation} because $A\in \mathcal{U}_{0}$ and $A\geq 0$ implies $%
A^{1/2}\in \mathcal{U}_{0}$. Moreover, its closure generates a strongly
continuous group of $\ast $-auto%
\-%
morphisms of $\mathcal{U}$, by \cite[Theorem 4.8]{brupedraLR}. See also
Proposition \ref{Theorem Lieb-Robinson} below.
\end{remark}

\subsection{Dynamics Generated by Short-Range Interactions \label{sect
Lieb--Robinson}}

We now consider time-dependent interactions. Let $\Psi \in C(\mathbb{R};%
\mathcal{W})$ be a continuous function from $\mathbb{R}$ to the Banach space 
$\mathcal{W}$ of interactions on $\mathfrak{L}\subseteq \mathbb{Z}^{d}$.
Then, for any $L\in \mathbb{N}$, there is a unique (fundamental) solution $%
(\tau _{t,s}^{(L,\Psi )})_{_{s,t\in \mathbb{R}}}$ in $\mathcal{B}(\mathcal{U}%
)$ to the (finite-volume) non-auto%
\-%
nomous evolution equations%
\begin{equation}
\forall s,t\in {\mathbb{R}}:\qquad \partial _{s}\tau _{t,s}^{(L,\Psi
)}=-\delta _{L}^{\Psi \left( s\right) }\circ \tau _{t,s}^{(L,\Psi )}\
,\qquad \tau _{t,t}^{(L,\Psi )}=\mathbf{1}_{\mathcal{U}}\ ,  \label{cauchy1}
\end{equation}%
and%
\begin{equation}
\forall s,t\in {\mathbb{R}}:\qquad \partial _{t}\tau _{t,s}^{(L,\Psi )}=\tau
_{t,s}^{(L,\Psi )}\circ \delta _{L}^{\Psi \left( t\right) }\ ,\qquad \tau
_{s,s}^{(L,\Psi )}=\mathbf{1}_{\mathcal{U}}\ .  \label{cauchy2}
\end{equation}%
In these two equations, $\mathbf{1}_{\mathcal{U}}$ refers to the identity
mapping of $\mathcal{U}$. Note also that, for any $L\in \mathbb{N}$ and $%
\Psi \in C\left( \mathbb{R};\mathcal{W}\right) $, $(\tau _{t,s}^{(L,\Psi
)})_{_{s,t\in \mathbb{R}}}$ is a continuous two-para%
\-%
meter family of bounded operators that satisfies the (reverse) cocycle
property%
\begin{equation*}
\forall s,r,t\in \mathbb{R}:\qquad \tau _{t,s}^{(L,\Psi )}=\tau
_{r,s}^{(L,\Psi )}\tau _{t,r}^{(L,\Psi )}\ .
\end{equation*}

If $\Psi \in C(\mathbb{R};\mathcal{W}^{\mathbb{R}})$ then $\delta _{L}^{\Psi
\left( t\right) }$ is always a symmetric derivation (or $\ast $-derivation)
and thus, in this case, $\tau _{t,s}^{(L,\Psi )}$ is a $\ast $-auto%
\-%
morphism of $\mathcal{U}$ for all lengths $L\in \mathbb{N}$ and times $%
s,t\in \mathbb{R}$. Moreover, in the thermodynamic limit $L\rightarrow
\infty $, the family $(\tau _{t,s}^{(L,\Psi )})_{_{s,t\in \mathbb{R}}}$
strongly converges to a strongly continuous two-parameter family of $\ast $%
-auto%
\-%
morphisms of $\mathcal{U}$, associated with the family $\{\delta ^{\Psi
\left( t\right) }\}_{t\in {\mathbb{R}}}$ of limit symmetric derivations of
Corollary \ref{Lemma cigare1}:

\begin{proposition}[Infinite-volume short-range dynamics]
\label{Theorem Lieb-Robinson}\mbox{
}\newline
For any $\Psi \in C(\mathbb{R};\mathcal{W}^{\mathbb{R}})$, as $L\rightarrow
\infty $, $(\tau _{t,s}^{(L,\Psi )})_{s,t\in {\mathbb{R}}}$ converges
strongly, uniformly for $s,t$ on compacta, to a strongly continuous two-para%
\-%
meter family $(\tau _{t,s}^{\Psi })_{s,t\in {\mathbb{R}}}$ of $\ast $-auto%
\-%
morphisms of $\mathcal{U}$, which is the unique solution in $\mathcal{B}(%
\mathcal{U})$ to the non-auto%
\-%
nomous evolutions equation%
\begin{equation}
\forall s,t\in {\mathbb{R}}:\qquad \partial _{t}\tau _{t,s}^{\Psi }=\tau
_{t,s}^{\Psi }\circ \delta ^{\Psi \left( t\right) }\ ,\qquad \tau
_{s,s}^{\Psi }=\mathbf{1}_{\mathcal{U}}\ ,  \label{cauchy trivial1}
\end{equation}%
in the strong sense on the dense subspace $\mathcal{U}_{0}\subseteq \mathcal{%
U}$. In particular, it satisfies the reverse cocycle property:%
\begin{equation}
\forall s,r,t\in \mathbb{R}:\qquad \tau _{t,s}^{\Psi }=\tau _{r,s}^{\Psi
}\tau _{t,r}^{\Psi }\ .  \label{reverse cocycle}
\end{equation}
\end{proposition}

\begin{proof}
See \cite[Corollary 5.2]{brupedraLR}.
\end{proof}

It is convenient to introduce at this point the notation%
\begin{equation*}
\partial _{\Psi }\Lambda \doteq \left\{ x\in \Lambda :\exists \mathcal{Z}\in 
\mathcal{P}_{f}\text{ with }x\in \mathcal{Z}\text{, }\Psi _{\mathcal{Z}}\neq
0\text{ and }\mathcal{Z}\cap \Lambda \neq 0\text{, }\mathcal{Z}\cap \Lambda
^{c}\neq 0\right\}
\end{equation*}%
for any interaction $\Psi $ and any finite subset $\Lambda \in \mathcal{P}%
_{f}$ with complement $\Lambda ^{c}\doteq \mathfrak{L}\backslash \Lambda $.
For any $s,t\in \mathbb{R}$, $t\wedge s$ and $t\vee s$ stand, respectively,
for the minimum and maximum of the set $\{s,t\}$. We are now in a position
to give additional estimates on the limit dynamics, like the celebrated
Lieb-Robinson bounds:

\begin{proposition}[Estimates on short-range dynamics]
\label{Theorem Lieb-Robinson copy(3)}\mbox{
}\newline
For any $\Psi \in C(\mathbb{R};\mathcal{W}^{\mathbb{R}})$, $(\tau
_{t,s}^{\Psi })_{s,t\in {\mathbb{R}}}$ satisfies the following
bounds:\medskip \newline
\emph{(i)} Lieb-Robinson bounds (cf. (\ref{commutatorbis})). For any $s,t\in 
\mathbb{R}$, sets $\Lambda ^{(1)},\Lambda ^{(2)}\in \mathcal{P}_{f}$ with $%
\Lambda ^{(1)}\cap \Lambda ^{(2)}=\emptyset $, every even element $A_{1}\in 
\mathcal{U}^{+}\cap \mathcal{U}_{\Lambda ^{(1)}}$ and all $A_{2}\in \mathcal{%
U}_{\Lambda ^{(2)}}$,%
\begin{equation*}
\left\Vert \lbrack \tau _{t,s}^{\Psi }\left( A_{1}\right) ,A_{2}]\right\Vert
_{\mathcal{U}}\leq 2\mathbf{D}^{-1}\left\Vert A_{1}\right\Vert _{\mathcal{U}%
}\left\Vert A_{2}\right\Vert _{\mathcal{U}}\left( \mathrm{e}^{2\mathbf{D}%
\int_{t\wedge s}^{t\vee s}\left\Vert \Psi \left( \alpha \right) \right\Vert
_{\mathcal{W}}\mathrm{d}\alpha }-1\right) \sum_{x\in \partial _{\Psi
}\Lambda ^{(1)}}\sum_{y\in \Lambda ^{(2)}}\mathbf{F}\left( x,y\right) \ .
\end{equation*}%
\emph{(ii)} Rate of convergence. For any $s,t\in \mathbb{R}$, $\Lambda \in 
\mathcal{P}_{f}$, $A\in \mathcal{U}_{\Lambda }$ and $L\in \mathbb{N}$ such
that $\Lambda \subseteq \Lambda _{L}$,%
\begin{eqnarray*}
&&\left\Vert \tau _{t,s}^{\Psi }\left( A\right) -\tau _{t,s}^{(L,\Psi
)}\left( A\right) \right\Vert _{\mathcal{U}} \\
&\leq &2\left\Vert A\right\Vert _{\mathcal{U}}\int_{t\wedge s}^{t\vee
s}\left( \left\Vert \Psi \left( \alpha _{1}\right) \right\Vert _{\mathcal{W}}%
\mathrm{e}^{2\mathbf{D}\int_{\alpha _{1}\wedge s}^{\alpha _{1}\vee
s}\left\Vert \Psi \left( \alpha _{2}\right) \right\Vert _{\mathcal{W}}%
\mathrm{d}\alpha _{2}}\right) \mathrm{d}\alpha _{1}\sum\limits_{y\in 
\mathfrak{L}\backslash \Lambda _{L}}\sum_{x\in \Lambda }\mathbf{F}\left(
x,y\right) \ .
\end{eqnarray*}%
\emph{(iii)} Lipschitz continuity with respect to $\Psi $. For any $s,t\in 
\mathbb{R}$, $\Psi ,\tilde{\Psi}\in C\left( \mathbb{R};\mathcal{W}^{\mathbb{R%
}}\right) $, $\Lambda \in \mathcal{P}_{f}$ and $A\in \mathcal{U}_{\Lambda }$%
, 
\begin{equation*}
\left\Vert \tau _{t,s}^{\Psi }(A)-\tau _{t,s}^{\tilde{\Psi}}(A)\right\Vert _{%
\mathcal{U}}\leq 2\left\vert \Lambda \right\vert \left\Vert A\right\Vert _{%
\mathcal{U}}\left\Vert \mathbf{F}\right\Vert _{1,\mathfrak{L}}\int_{t\wedge
s}^{t\vee s}\mathrm{e}^{2\mathbf{D}\int_{t\wedge \alpha }^{t\vee \alpha
}\left\Vert \Psi \left( \alpha _{1}\right) \right\Vert _{\mathcal{W}}\mathrm{%
d}\alpha _{1}}\left\Vert \tilde{\Psi}\left( \alpha \right) -\Psi \left(
\alpha \right) \right\Vert _{\mathcal{W}}\mathrm{d}\alpha \ .
\end{equation*}%
\emph{(iv)} Uniform continuity with respect to times. For any $%
s_{1},s_{2},t_{1},t_{2}\in \mathbb{R}$, $\Psi \in C(\mathbb{R};\mathcal{W}^{%
\mathbb{R}})$, $\Lambda \in \mathcal{P}_{f}$ and $A\in \mathcal{U}_{\Lambda
} $, 
\begin{eqnarray*}
&&\left\Vert \tau _{t_{1},s_{1}}^{\Psi }(A)-\tau _{t_{2},s_{2}}^{\Psi
}(A)\right\Vert _{\mathcal{U}} \\
&\leq &2\left\vert \Lambda \right\vert \left\Vert A\right\Vert _{\mathcal{U}%
}\left\Vert \mathbf{F}\right\Vert _{1,\mathfrak{L}}\left( \int_{t_{1}\wedge
t_{2}}^{t_{1}\vee t_{2}}\left\Vert \Psi \left( \alpha \right) \right\Vert _{%
\mathcal{W}}\mathrm{d}\alpha +\int_{s_{1}\wedge s_{2}}^{s_{1}\vee s_{2}}%
\mathrm{e}^{2\mathbf{D}\int_{t_{2}\wedge \alpha }^{t_{2}\vee \alpha
}\left\Vert \Psi \left( \alpha _{1}\right) \right\Vert _{\mathcal{W}}\mathrm{%
d}\alpha _{1}}\left\Vert \Psi \left( \alpha \right) \right\Vert _{\mathcal{W}%
}\mathrm{d}\alpha \right) \ .
\end{eqnarray*}
\end{proposition}

\begin{proof}
The proof of Assertion (i) is almost done in \cite[Theorem 5.1, Corollary
5.2 (ii)]{brupedraLR}. However, the bound there refers to the supremum with
respect to $\alpha $ of the norm $\Vert \Psi \left( \alpha \right) \Vert _{%
\mathcal{W}}$. Here, we need a slightly more accurate estimate (a point-wise
estimate). In fact, by \cite[equation after Eq. (5.4)]{brupedraLR} and
similar arguments as in \cite[Eqs. (4.16)-(4.18)]{brupedraLR}, we get
Assertion (i). Then, Assertion (ii) is proven exactly like in the proof of 
\cite[Theorem 5.1 (ii)]{brupedraLR}, by replacing \cite[Theorem 5.1 (i)]%
{brupedraLR} with Assertion (i). It remains to prove Assertions (iii) and
(iv). We start with (iii):

For any $s,t\in \mathbb{R}$, $\Psi ,\tilde{\Psi}\in C(\mathbb{R};\mathcal{W}%
^{\mathbb{R}})$, $\Lambda \in \mathcal{P}_{f}$, $A\in \mathcal{U}_{\Lambda }$
and any sufficiently large $L\in \mathbb{N}$ such that $\mathcal{U}_{\Lambda
}\subseteq \mathcal{U}_{\Lambda _{L}}$, by (\ref{cauchy1}) and Proposition %
\ref{Theorem Lieb-Robinson}, observe that 
\begin{eqnarray}
\tau _{t,s}^{\tilde{\Psi}}(A)-\tau _{t,s}^{(L,\Psi )}(A) &=&\int_{s}^{t}\tau
_{\alpha ,s}^{\tilde{\Psi}}\circ \left( \delta ^{\tilde{\Psi}\left( \alpha
\right) }-\delta ^{\Psi \left( \alpha \right) }\right) \circ \tau _{t,\alpha
}^{(L,\Psi )}\left( A\right) \mathrm{d}\alpha
\label{assertion bisbisbisbis0100} \\
&&+\int_{s}^{t}\tau _{\alpha ,s}^{\tilde{\Psi}}\circ \left( \delta ^{\Psi
\left( \alpha \right) }-\delta ^{(L,\Psi \left( \alpha \right) )}\right)
\circ \tau _{t,\alpha }^{(L,\Psi )}\left( A\right) \mathrm{d}\alpha \ . 
\notag
\end{eqnarray}%
By Definitions \ref{definition fininte vol dynam0}, \ref{dynamic series} and
Corollary \ref{Lemma cigare1}, it follows that 
\begin{eqnarray}
\left\Vert \tau _{t,s}^{\tilde{\Psi}}(A)-\tau _{t,s}^{(L,\Psi
)}(A)\right\Vert _{\mathcal{U}} &\leq &\int_{t\wedge s}^{t\vee
s}\sum\limits_{\mathcal{Z}\in \mathcal{P}_{f}}\left\Vert \left[ \left( 
\tilde{\Psi}\left( \alpha \right) -\Psi \left( \alpha \right) \right) _{%
\mathcal{Z}},\tau _{t,\alpha }^{(L,\Psi )}\left( A\right) \right]
\right\Vert _{\mathcal{U}}\mathrm{d}\alpha  \label{assertion bisbisbisbis01}
\\
&&+\int_{t\wedge s}^{t\vee s}\sum\limits_{\mathcal{Z}\in \mathcal{P}_{f},\ 
\mathcal{Z}\cap \Lambda _{L}^{c}\neq \emptyset }\left\Vert \left[ \Psi
\left( \alpha \right) _{\mathcal{Z}},\tau _{t,\alpha }^{(L,\Psi )}\left(
A\right) \right] \right\Vert _{\mathcal{U}}\mathrm{d}\alpha \ ,  \notag
\end{eqnarray}%
where $\Lambda _{L}^{c}\doteq \mathfrak{L}\backslash \Lambda _{L}$ is the
complement of the cubic box $\Lambda _{L}$ (\ref{eq:def lambda n}). Now, by
using Assertion (i) for $\Psi _{L}\in C(\mathbb{R};\mathcal{W}^{\mathbb{R}})$
defined\footnote{$\mathbf{1}\left[ p\right] =1$ when the proposition $p$ is
true and $0$, else.}, for $L\in \mathbb{N}$, by 
\begin{equation*}
\Psi _{L}\left( t\right) _{\mathcal{Z}}=\Psi \left( t\right) _{\mathcal{Z}}%
\mathbf{1}\left[ \mathcal{Z}\subseteq \Lambda _{L}\right] \ ,\qquad \mathcal{%
Z}\in \mathcal{P}_{f},\ t\in \mathbb{R}\ ,
\end{equation*}%
together with $\Vert \Psi _{L}\left( t\right) \Vert _{\mathcal{W}}\leq \Vert
\Psi \left( t\right) \Vert _{\mathcal{W}}$, (\ref{(3.1) NS})-(\ref{(3.2) NS}%
) and Equation (\ref{iteration0}), we get that 
\begin{eqnarray}
&&\sum\limits_{\mathcal{Z}\in \mathcal{P}_{f},\ \mathcal{Z}\cap \Lambda
_{L}^{c}\neq \emptyset }\left\Vert \left[ \Psi \left( \alpha \right) _{%
\mathcal{Z}},\tau _{t,\alpha }^{(L,\Psi )}\left( A\right) \right]
\right\Vert _{\mathcal{U}}  \label{assertion bisbisbisbis02} \\
&\leq &2\left\Vert A\right\Vert _{\mathcal{U}}\left( \mathrm{e}^{2\mathbf{D}%
\int_{t\wedge \alpha }^{t\vee \alpha }\left\Vert \Psi \left( \alpha
_{1}\right) \right\Vert _{\mathcal{W}}\mathrm{d}\alpha _{1}}\right)
\left\Vert \Psi \left( \alpha \right) \right\Vert _{\mathcal{W}%
}\sum\limits_{x\in \Lambda }\sum\limits_{y\in \Lambda _{L}^{c}}\mathbf{F}%
\left( x,y\right)  \notag
\end{eqnarray}%
and%
\begin{eqnarray}
&&\sum\limits_{\mathcal{Z}\in \mathcal{P}_{f}}\left\Vert \left[ \left( 
\tilde{\Psi}\left( \alpha \right) -\Psi \left( \alpha \right) \right) _{%
\mathcal{Z}},\tau _{t,\alpha }^{(L,\Psi )}\left( A\right) \right]
\right\Vert _{\mathcal{U}}  \label{assertion bisbisbisbis03} \\
&\leq &2\left\vert \Lambda \right\vert \left\Vert A\right\Vert _{\mathcal{U}%
}\left\Vert \mathbf{F}\right\Vert _{1,\mathfrak{L}}\mathrm{e}^{2\mathbf{D}%
\int_{t\wedge \alpha }^{t\vee \alpha }\left\Vert \Psi \left( \alpha
_{1}\right) \right\Vert _{\mathcal{W}}\mathrm{d}\alpha _{1}}\left\Vert 
\tilde{\Psi}\left( \alpha \right) -\Psi \left( \alpha \right) \right\Vert _{%
\mathcal{W}}\ .  \notag
\end{eqnarray}%
To prove these two inequalities, see \cite[Eqs. (4.25)-(4.25)]{brupedraLR}.
Since 
\begin{equation}
\underset{L\rightarrow \infty }{\lim }\sum_{x\in \Lambda }\sum\limits_{y\in
\Lambda _{L}^{c}}\mathbf{F}\left( x,y\right) =0\ ,
\label{assertion bisbisbisbis}
\end{equation}%
because of (\ref{(3.1) NS}), Assertion (iii) follows by combining (\ref%
{assertion bisbisbisbis01})-(\ref{assertion bisbisbisbis}) with Assertion
(ii).

Finally, to get Assertion (iv), note first that Corollary \ref{Lemma cigare1}
directly implies that 
\begin{equation}
\left\Vert \tau _{t_{1},s}^{\Psi }(A)-\tau _{t_{2},s}^{\Psi }(A)\right\Vert
_{\mathcal{U}}\leq 2\left\vert \Lambda \right\vert \left\Vert A\right\Vert _{%
\mathcal{U}}\left\Vert \mathbf{F}\right\Vert _{1,\mathfrak{L}%
}\int_{t_{1}\wedge t_{2}}^{t_{1}\vee t_{2}}\left\Vert \Psi \left( \alpha
\right) \right\Vert _{\mathcal{W}}\mathrm{d}\alpha
\label{asserion (iv) con1}
\end{equation}%
for any $s,t_{1},t_{2}\in \mathbb{R}$, $\Psi \in C(\mathbb{R};\mathcal{W}^{%
\mathbb{R}})$, $\Lambda \in \mathcal{P}_{f}$ and $A\in \mathcal{U}_{\Lambda
} $. Meanwhile, fix $s_{1},s_{2},t\in \mathbb{R}$, $\Psi \in C(\mathbb{R};%
\mathcal{W}^{\mathbb{R}})$, $\Lambda \in \mathcal{P}_{f}$ and $A\in \mathcal{%
U}_{\Lambda }$. By Assertion (ii), for any $\varepsilon \in \mathbb{R}^{+}$
there is $L\in \mathbb{R}$ such that%
\begin{equation*}
\left\Vert \tau _{t,s_{1}}^{\Psi }(A)-\tau _{t,s_{2}}^{\Psi }(A)\right\Vert
_{\mathcal{U}}\leq \left\Vert \tau _{t,s_{1}}^{(L,\Psi )}(A)-\tau
_{t,s_{2}}^{(L,\Psi )}(A)\right\Vert _{\mathcal{U}}+\varepsilon \ ,
\end{equation*}%
which, by Equation (\ref{cauchy1}), implies that%
\begin{equation}
\left\Vert \tau _{t,s_{1}}^{\Psi }(A)-\tau _{t,s_{2}}^{\Psi }(A)\right\Vert
_{\mathcal{U}}\leq \int_{s_{1}\wedge s_{2}}^{s_{1}\vee s_{2}}\sum\limits_{%
\mathcal{Z}\in \mathcal{P}_{f}}\left\Vert \left[ \Psi \left( \alpha \right)
_{\mathcal{Z}},\tau _{t,\alpha }^{(L,\Psi )}\left( A\right) \right]
\right\Vert _{\mathcal{U}}\mathrm{d}\alpha +\varepsilon \ .
\label{estimation1}
\end{equation}%
Similar to (\ref{assertion bisbisbisbis03}), it follows that%
\begin{equation}
\left\Vert \tau _{t,s_{1}}^{\Psi }(A)-\tau _{t,s_{2}}^{\Psi }(A)\right\Vert
_{\mathcal{U}}\leq 2\left\vert \Lambda \right\vert \left\Vert A\right\Vert _{%
\mathcal{U}}\left\Vert \mathbf{F}\right\Vert _{1,\mathfrak{L}%
}\int_{s_{1}\wedge s_{2}}^{s_{1}\vee s_{2}}\mathrm{e}^{2\mathbf{D}%
\int_{t\wedge \alpha }^{t\vee \alpha }\left\Vert \Psi \left( \alpha
_{1}\right) \right\Vert _{\mathcal{W}}\mathrm{d}\alpha _{1}}\left\Vert \Psi
\left( \alpha \right) \right\Vert _{\mathcal{W}}\mathrm{d}\alpha \ .
\label{asserion (iv) con2}
\end{equation}%
Assertion (iv) is a combination of (\ref{asserion (iv) con1}) and (\ref%
{asserion (iv) con2}).
\end{proof}

\section{Lattice Fermions with Long-Range Interactions\label{Long-Range
systems}}

\subsection{Banach Space of Long-Range Models\label{Long-rande gef}}

Fix now $\mathfrak{L}=\mathbb{Z}^{d}$, $d\in \mathbb{N}$. Let $\mathbb{S}$
be the unit sphere of the Banach space $\mathcal{W}_{1}$\ of
translation-invariant (complex) interactions. Observe that any finite signed
Borel measure $\mathfrak{a}$ on $\mathbb{S}$ defines an interaction%
\begin{equation}
\int_{\mathbb{S}}\Psi \ \mathfrak{a}\left( \mathrm{d}\Psi \right) \in 
\mathcal{W}_{1}  \label{definition integral interaction0}
\end{equation}%
by 
\begin{equation}
\left( \int_{\mathbb{S}}\Psi \ \mathfrak{a}\left( \mathrm{d}\Psi \right)
\right) _{\Lambda }\doteq \int_{\mathbb{S}}\Psi _{\Lambda }\mathfrak{a}%
\left( \mathrm{d}\Psi \right) \ ,\qquad \Lambda \in \mathcal{P}_{f}\ .
\label{definition integral interaction}
\end{equation}%
This last integral is well-defined because, for each $\Lambda \in \mathcal{P}%
_{f}$, the integrand is an absolutely integrable function taking values in a
finite-dimensional normed space, which is $\mathcal{U}_{\Lambda }$. Below,
we extend this observation to define long-range, or mean-field, models. Note
that (\ref{definition integral interaction0}) can also be seen as a Bochner
integral because the measure $\mathfrak{a}$ is finite and $\mathcal{W}$ is a
separable Banach space. See, e.g., \cite[Theorems 1.1 and 1.2]{pettis}.

For any $n\in \mathbb{N}$ and any finite signed Borel measure $\mathfrak{a}$
on the Cartesian product $\mathbb{S}^{n}$ (endowed with its product
topology), we define the finite signed Borel measure $\mathfrak{a}^{\ast }$
to be the pushforward of $\mathfrak{a}$ through the automorphism 
\begin{equation}
\left( \Psi ^{(1)},\ldots ,\Psi ^{(n)}\right) \mapsto ((\Psi ^{(n)})^{\ast
},\ldots ,(\Psi ^{(1)})^{\ast })\in \mathbb{S}^{n}
\label{push forward self-adjoint}
\end{equation}%
of $\mathbb{S}^{n}$ as a topological space. A finite signed Borel measure $%
\mathfrak{a}$ on $\mathbb{S}^{n}$ is, by definition, \emph{self-adjoint}
whenever $\mathfrak{a^{\ast }}=\mathfrak{a}$.

For any $n\in \mathbb{N}$, we denote the space of \emph{self-adjoint},
finite, signed Borel measures on $\mathbb{S}^{n}$ by $\mathcal{S}(\mathbb{S}%
^{n}\mathbb{)}$, which is a real Banach space with the norm of the total
variation%
\begin{equation}
\Vert \mathfrak{a}\Vert _{\mathcal{S}(\mathbb{S}^{n}\mathbb{)}}\doteq |%
\mathfrak{a}|(\mathbb{S}^{n})\ ,\qquad n\in \mathbb{N}\ .
\label{definition 0}
\end{equation}%
The set of all sequences $\mathfrak{a}\equiv (\mathfrak{a}_{n})_{n\in 
\mathbb{N}}$ of self-adjoint, finite, signed Borel measures $\mathfrak{a}%
_{n}\in \mathcal{S}(\mathbb{S}^{n}\mathbb{)}$ is a real vector space, where 
\begin{equation*}
(\mathfrak{a}+\mathfrak{\tilde{a}})_{n}\doteq \mathfrak{a}_{n}+\mathfrak{%
\tilde{a}}_{n}\text{\qquad and\qquad }(\lambda \mathfrak{a})_{n}\doteq
\lambda \mathfrak{a}_{n}\ ,\qquad n\in \mathbb{N}\ ,
\end{equation*}%
for any sequence $\mathfrak{a}\equiv (\mathfrak{a}_{n})_{n\in \mathbb{N}},%
\mathfrak{\tilde{a}}\equiv (\mathfrak{\tilde{a}}_{n})_{n\in \mathbb{N}}$\
and all $\lambda \in \mathbb{R}$. We define the (real) space $\mathcal{S}$
to be the set of all sequences $\mathfrak{a}\equiv (\mathfrak{a}_{n})_{n\in 
\mathbb{N}}$ of self-adjoint, finite signed Borel measures $\mathfrak{a}%
_{n}\in \mathcal{S}(\mathbb{S}^{n}\mathbb{)}$ with 
\begin{equation}
\left\Vert \mathfrak{a}\right\Vert _{\mathcal{S}}\doteq \sum_{n\in \mathbb{N}%
}n^{2}\left\Vert \mathbf{F}\right\Vert _{1,\mathbb{Z}^{d}}^{n-1}\left\Vert 
\mathfrak{a}_{n}\right\Vert _{\mathcal{S}(\mathbb{S}^{n}\mathbb{)}}<\infty \
,  \label{definition 0bis}
\end{equation}%
where we recall that $\mathbf{F}:\mathbb{Z}^{d}\times \mathbb{Z}%
^{d}\rightarrow (0,1]$ is the decay function with maximum value $\mathbf{F}%
\left( x,x\right) =1$ for $x\in \mathbb{Z}^{d}$ and satisfying Conditions (%
\ref{(3.1) NS})-(\ref{(3.2) NS}). See also (\ref{iteration0}). Observe that $%
(\mathcal{S},\Vert \cdot \Vert _{\mathcal{S}})$ is a real Banach space. We
are now in a position to define \emph{long-range} models:

\begin{definition}[Long-range models]
\label{def long range}\mbox{
}\newline
The (real) Banach space of long-range models is the space $\mathcal{M}\doteq 
\mathcal{W}^{\mathbb{R}}\times \mathcal{S}$ along with the norm 
\begin{equation*}
\left\Vert \mathfrak{m}\right\Vert _{\mathcal{M}}\doteq \left\Vert \Phi
\right\Vert _{\mathcal{W}}+\left\Vert \mathfrak{a}\right\Vert _{\mathcal{S}%
}\,,\qquad \mathfrak{m}\doteq \left( \Phi ,\mathfrak{a}\right) \in \mathcal{M%
}\ .
\end{equation*}
\end{definition}

\noindent Note that $\mathcal{W}^{\mathbb{R}}$ and $\mathcal{S}$ are
canonically seen as subspaces of $\mathcal{M}$, i.e., 
\begin{equation*}
\mathcal{W}^{\mathbb{R}}\subseteq \mathcal{M}\qquad \text{and}\qquad 
\mathcal{S\subseteq M}\ .
\end{equation*}%
We emphasize that long-range models are \emph{not necessarily}
translation-invariant since, obviously, 
\begin{equation}
\mathcal{M}_{1}\doteq \left( \mathcal{W}_{1}\cap \mathcal{W}^{\mathbb{R}%
}\right) \times \mathcal{S\varsubsetneq M}\ .
\label{translatino invariatn long range models}
\end{equation}

Similar to (\ref{eq:enpersite})-(\ref{eq:enpersitebis}), we define the
subsets 
\begin{equation}
\mathcal{M}_{\Lambda }\doteq \mathcal{W}^{\mathbb{R}}\times \mathcal{S}%
_{\Lambda }\subseteq \mathcal{M}\ ,\qquad \Lambda \in \mathcal{P}_{f}\ ,
\label{S00bis}
\end{equation}%
where, for any $\Lambda \in \mathcal{P}_{f}$ ,%
\begin{equation}
\mathcal{S}_{\Lambda }\doteq \left\{ (\mathfrak{a}_{n})_{n\in \mathbb{N}}\in 
\mathcal{S}:\forall n\in \mathbb{N},\ |\mathfrak{a}_{n}|(\mathbb{S}^{n})=|%
\mathfrak{a}_{n}|((\mathbb{S}\cap \mathcal{W}_{\Lambda })^{n})\right\} \ .
\label{S0bis}
\end{equation}%
Note that the short-range part of models in $\mathcal{M}_{\Lambda }$ is 
\emph{not necessarily} finite-range, but their long-range interactions are
built from finite-range interactions. Similar to (\ref{W0}) we can define a
dense subspace 
\begin{equation}
\mathcal{M}_{0}\doteq \bigcup_{L\in \mathbb{N}}\mathcal{M}_{\Lambda
_{L}}\subseteq \mathcal{M}\ .  \label{S00bisbis}
\end{equation}

\subsection{Local Hamiltonians and Derivations on the CAR Algebra\label{sect
Lieb--Robinson copy(2)}}

Similar to Definition \ref{definition fininte vol dynam0}, we define a
sequence of local Hamiltonians for any model $\mathfrak{m}\in \mathcal{M}$:
At any fixed $n\in \mathbb{N}$ and $L\in \mathbb{N}$, the mapping%
\begin{equation*}
\left( \Psi ^{(1)},\ldots ,\Psi ^{(n)}\right) \mapsto U_{L}^{\Psi
^{(1)}}\cdots U_{L}^{\Psi ^{(n)}}
\end{equation*}%
from $\mathbb{S}^{n}$ to $\mathcal{U}$\ is continuous (see (\ref{norm Uphi}%
)), and so, for any long-range model $\mathfrak{m}\in \mathcal{M}$, we can
define the following self-adjoint element of $\mathcal{U}$:

\begin{definition}[Hamiltonians]
\label{definition long range energy}\mbox{ }\newline
The local Hamiltonians of any model $\mathfrak{m}\in \mathcal{M}$ are%
\begin{equation*}
U_{L}^{\mathfrak{m}}\doteq U_{L}^{\Phi }+\sum_{n\in \mathbb{N}}\frac{1}{%
\left\vert \Lambda _{L}\right\vert ^{n-1}}\int_{\mathbb{S}^{n}}U_{L}^{\Psi
^{(1)}}\cdots U_{L}^{\Psi ^{(n)}}\mathfrak{a}_{n}\left( \mathrm{d}\Psi
^{(1)},\ldots ,\mathrm{d}\Psi ^{(n)}\right) \,,\qquad L\in \mathbb{N}\ .
\end{equation*}
\end{definition}

\noindent Note that 
\begin{equation*}
U_{L}^{\mathfrak{m}}\in \mathcal{U}_{\Lambda _{L}}\cap \mathcal{U}^{\mathbb{R%
}}\cap \mathcal{U}^{+}\,,\qquad L\in \mathbb{N}\ ,
\end{equation*}%
and straightforward estimates using Equations (\ref{norm Uphi}), (\ref%
{definition 0})-(\ref{definition 0bis}) and Definition \ref{def long range}
yield the bound 
\begin{equation}
\left\Vert U_{L}^{\mathfrak{m}}\right\Vert _{\mathcal{U}}\leq \left\vert
\Lambda _{L}\right\vert \left\Vert \mathbf{F}\right\Vert _{1,\mathfrak{L}%
}\left\Vert \mathfrak{m}\right\Vert _{\mathcal{M}}\ ,\qquad L\in \mathbb{N}\
.  \label{energy bound long range}
\end{equation}%
(This upper bound is relatively coarse, in general.)

All well-established Hamiltonians for lattice fermions in condensed matter
physics can be written as $U_{L}^{\mathfrak{m}}$ via some model $\mathfrak{m}%
\in \mathcal{M}$. In Section \ref{Section applications}, we shortly explain
an important example related to the BCS theory. Other examples can also be
found in \cite[Section 2.2]{BruPedra2} (for instance, in relation with the
forward scattering approximation) as well as in \cite{Bru-pedra-MF-IV} which
refers to the dynamics of the strong-coupling BCS-Hubbard model. The latter
is an interesting model because it predicts the existence of a
superconductor-Mott insulator phase transition, like in cuprates which must
be doped to become superconductors. See \cite{BruPedra1} for more details.

Note that the long-range character of $\mathfrak{m}\in \mathcal{M}$ with
local internal energy $U_{L}^{\mathfrak{m}}$ -- as compared to the usual
models defined from short-range interactions $\Phi \in \mathcal{W}$ only --
can be understood as follows: Take $\mathfrak{m}=\left( \Phi ,\mathfrak{a}%
\right) \in \mathcal{M}$. For each fixed $\epsilon \in (0,1)$, we define the
long-range truncation of $U_{L}^{\Phi }$, the short-range component of the
total internal energy (Definition \ref{definition fininte vol dynam0}), by%
\begin{equation*}
U_{L,\epsilon }^{\Phi }\doteq \sum\limits_{\Lambda \subseteq \Lambda _{L},\;%
\text{{\o }}(\Lambda )>\epsilon L}\Phi _{\Lambda },
\end{equation*}%
where the function {\o }$(\Lambda )$ is the diameter of $\Lambda \in 
\mathcal{P}_{f}$. Analogously, the long-range truncation of the internal
energy $(U_{L}^{\mathfrak{m}}-U_{L}^{\Phi })$ associated with the long-range
part of $\mathfrak{m}$ is by definition equal to%
\begin{equation*}
U_{L,\epsilon }^{\mathfrak{m}}\doteq \sum_{n\in \mathbb{N}}\frac{1}{%
\left\vert \Lambda _{L}\right\vert ^{n-1}}\int_{\mathbb{S}%
^{n}}\sum\limits_{\Lambda \subseteq \Lambda _{L},\;\text{{\o }}(\Lambda
^{(1)}\cup \cdots \cup \Lambda ^{(n)})>\epsilon L}\Psi _{\Lambda
^{(1)}}^{(1)}\cdots \Psi _{\Lambda ^{(n)}}^{(n)}\mathfrak{a}_{n}\left( 
\mathrm{d}\Psi ^{(1)},\ldots ,\mathrm{d}\Psi ^{(n)}\right) .
\end{equation*}%
If $\mathfrak{m}=\left( \Phi ,\mathfrak{a}\right) \in \mathcal{M}\backslash 
\mathcal{W}$, for any $\epsilon \in (0,1)$ one generally has that%
\begin{equation*}
\lim\limits_{L\rightarrow \infty }\Vert U_{L,\epsilon }^{\Phi }\Vert
=0\qquad \text{and}\qquad \lim\limits_{L\rightarrow \infty }\Vert
U_{L,\epsilon }^{\mathfrak{m}}\Vert >0.
\end{equation*}%
In particular, the long-range part $(U_{L}^{\mathfrak{m}}-U_{L}^{\Phi })$ of
the internal energy $U_{L}^{\mathfrak{m}}$ generally dominates the
interaction at large distances, in the infinite volume limit.

For any translation-invariant model $\mathfrak{m}\doteq \left( \Phi ,%
\mathfrak{a}\right) \in \mathcal{M}_{1}$ (cf. (\ref{translatino invariatn
long range models})), observe that 
\begin{equation*}
U_{L}^{\mathfrak{m}}\doteq U_{L}^{\tilde{\Phi}}+\sum_{n=2}^{\infty }\frac{1}{%
\left\vert \Lambda _{L}\right\vert ^{n-1}}\int_{\mathbb{S}^{n}}U_{L}^{\Psi
^{(1)}}\cdots U_{L}^{\Psi ^{(n)}}\mathfrak{a}_{n}\left( \mathrm{d}\Psi
^{(1)},\ldots ,\mathrm{d}\Psi ^{(n)}\right) \,,\qquad L\in \mathbb{N}\ ,
\end{equation*}%
where%
\begin{equation*}
\tilde{\Phi}\doteq \Phi +\int_{\mathbb{S}}\Psi \ \mathfrak{a}_{1}\left( 
\mathrm{d}\Psi \right) \in \mathcal{W}\ ,
\end{equation*}%
the last integral being defined by (\ref{definition integral interaction}).
If the model is short-range and translation-invariant, i.e., $\Phi ,\tilde{%
\Phi}\in \mathcal{W}_{1}$, then the interaction $\tilde{\Phi}$ can be
encoded in some self-adjoint, finite, signed Borel measure $\mathfrak{\tilde{%
a}}_{1}$, leading to the definition of a new model $\mathfrak{\tilde{m}}%
\doteq \left( 0,\mathfrak{\tilde{a}}\right) \in \mathcal{M}_{1}$ such that $%
U_{L}^{\mathfrak{m}}=U_{L}^{\mathfrak{\tilde{m}}}$. In other words, if one
is only interested in models that are short-range and translation-invariant,
then one can only consider the Banach space $\mathcal{S}$. Finally, for any
translation-invariant model $\mathfrak{m}=\left( \Phi ,(0,\mathfrak{a}%
_{2},0,\ldots )\right) \in \mathcal{M}_{1}$, remark that 
\begin{equation}
U_{L}^{\mathfrak{m}}\doteq U_{L}^{\Phi }+\frac{1}{\left\vert \Lambda
_{L}\right\vert }\int_{\mathbb{S}^{2}}U_{L}^{\Psi ^{(1)}}U_{L}^{\Psi ^{(2)}}%
\mathfrak{a}_{2}\left( \mathrm{d}\Psi ^{(1)},\mathrm{d}\Psi ^{(2)}\right) \ ,
\label{definition quadratic}
\end{equation}%
which can be seen as a local Hamiltonian of a long-range model in the sense
of \cite{BruPedra2}, as explained in\ Section \ref{Long-range models}.
Explicit examples of such models are given in \cite[Section 2.2]{BruPedra2}.

Like in Definition \ref{dynamic series}, any model $\mathfrak{m}\in \mathcal{%
M}$ yields a sequence of bounded derivations:

\begin{definition}[Derivations on the CAR\ algebra for long-range
interactions]
\label{derivation long range}\mbox{ }\newline
The (symmetric) derivations $\{\delta _{L}^{\mathfrak{m}}\}_{L\in \mathbb{N}%
}\subseteq \mathcal{B}(\mathcal{U})$ associated with any model $\mathfrak{m}%
\in \mathcal{M}$ are defined by%
\begin{equation*}
\delta _{L}^{\mathfrak{m}}(A)\doteq i\left[ U_{L}^{\mathfrak{m}},A\right]
\doteq i\left( U_{L}^{\mathfrak{m}}A-AU_{L}^{\mathfrak{m}}\right) \ ,\qquad
A\in \mathcal{U},\ L\in \mathbb{N}\ .
\end{equation*}
\end{definition}

\subsection{Dynamical Problem Associated with Long-Range Interactions\label%
{section pb}}

For any long-range model $\mathfrak{m}\in \mathcal{M}$, the finite-volume
dynamics are always well-defined: For all $L\in \mathbb{N}$ there is a
strongly continuous one-parameter group $(\tau _{t}^{(L,\mathfrak{m}%
)})_{_{t\in \mathbb{R}}}$ of $\ast $-auto%
\-%
morphisms of $\mathcal{U}$ generated by $\delta _{L}^{\mathfrak{m}}\in 
\mathcal{B}(\mathcal{U})$: 
\begin{equation}
\tau _{t}^{(L,\mathfrak{m})}\left( A\right) \doteq \mathrm{e}^{itU_{L}^{%
\mathfrak{m}}}A\mathrm{e}^{-itU_{L}^{\mathfrak{m}}}\ ,\qquad A\in \mathcal{U}%
\ .  \label{definition fininte vol dynam}
\end{equation}%
Compare with (\ref{cauchy1})-(\ref{cauchy2}) for short-range interactions in
the autonomous situation.

Nevertheless, in contrast with short-range interactions (cf. Corollary \ref%
{Lemma cigare1} and Proposition \ref{Theorem Lieb-Robinson}), in the
thermodynamic limit $L\rightarrow \infty $, the finite-volume dynamics does
not generally converge within the $C^{\ast }$-algebra $\mathcal{U}$. To see
this, consider the following elementary example: Choose a model $\mathfrak{m}%
=\left( 0,(0,\mathfrak{a}_{2},0,\ldots )\right) \in \mathcal{M}_{1}$ such
that%
\begin{equation*}
U_{L}^{\mathfrak{m}}=\frac{1}{2\left\vert \Lambda _{L}\right\vert }%
N_{L}^{2}\qquad \text{with}\qquad N_{L}\doteq \sum_{x\in \Lambda _{L},%
\mathrm{s}\in \mathrm{S}}a_{x,\mathrm{s}}^{\ast }a_{x,\mathrm{s}}\ .
\end{equation*}%
Take $A=a_{0,\mathrm{s}}\in \mathcal{U}_{0}$ for some fixed spin $\mathrm{s}%
\in \mathrm{S}$. Observe that 
\begin{equation*}
\tau _{t}^{(L,\mathfrak{m})}\left( a_{0,\mathrm{s}}\right) =\mathrm{e}%
^{it\left( 2\left\vert \Lambda _{L}\right\vert \right) ^{-1}}\mathrm{e}%
^{it\left\vert \Lambda _{L}\right\vert ^{-1}a_{0,\mathrm{s}}^{\ast }a_{0,%
\mathrm{s}}N_{L}}a_{0,\mathrm{s}}\mathrm{e}^{-it\left\vert \Lambda
_{L}\right\vert ^{-1}a_{0,\mathrm{s}}^{\ast }a_{0,\mathrm{s}}N_{L}}\ .
\end{equation*}%
Therefore, for any $t\in \mathbb{R}$, 
\begin{equation}
\mathrm{e}^{-it\left( 2\left\vert \Lambda _{L}\right\vert \right) ^{-1}}\tau
_{t}^{(L,\mathfrak{m})}\left( a_{0,\mathrm{s}}\right) =a_{0,\mathrm{s}%
}+it\left\vert \Lambda _{L}\right\vert ^{-1}\left[ a_{0,\mathrm{s}}^{\ast
}a_{0,\mathrm{s}}N_{L},a_{0,\mathrm{s}}\right] +R_{L}\left( t\right)
\label{dfkjsdfkljdf}
\end{equation}%
with $\left\Vert R_{L}\left( t\right) \right\Vert _{\mathcal{U}}\leq 2t^{2}$%
. Note that 
\begin{equation*}
\left\vert \Lambda _{L}\right\vert ^{-1}\left[ a_{0,\mathrm{s}}^{\ast }a_{0,%
\mathrm{s}}N_{L},a_{0,\mathrm{s}}\right] =-\left\vert \Lambda
_{L}\right\vert ^{-1}a_{0,\mathrm{s}}N_{L}
\end{equation*}%
and it is straightforward to check that this last element does not converge
in $\mathcal{U}$, as $L\rightarrow \infty $. By Equation (\ref{dfkjsdfkljdf}%
), at least at small times $\left\vert t\right\vert >0$, $(\tau _{t}^{(L,%
\mathfrak{m})}(A))_{L\in \mathbb{N}}\subseteq \mathcal{U}$ does \emph{not}
converge, as $L\rightarrow \infty $.

The non-convergence property is generic: For any integer $n\geq 2$, $\Psi
^{(1)},\ldots ,\Psi ^{(n)}\in \mathcal{W}$, $A\in \mathcal{U}$ and all $L\in 
\mathbb{N}$, 
\begin{eqnarray}
\frac{1}{\left\vert \Lambda _{L}\right\vert ^{n-1}}\left[ U_{L}^{\Psi
^{(1)}}\cdots U_{L}^{\Psi ^{(n)}},A\right] &=&\left[ U_{L}^{\Psi ^{(1)}},A%
\right] \frac{U_{L}^{\Psi ^{(2)}}}{\left\vert \Lambda _{L}\right\vert }%
\cdots \frac{U_{L}^{\Psi ^{(n)}}}{\left\vert \Lambda _{L}\right\vert }+\frac{%
U_{L}^{\Psi ^{(1)}}}{\left\vert \Lambda _{L}\right\vert }\cdots \frac{%
U_{L}^{\Psi ^{(n-1)}}}{\left\vert \Lambda _{L}\right\vert }\left[
U_{L}^{\Psi ^{(n)}},A\right]  \notag \\
&&+\sum_{m=2}^{n-1}\frac{U_{L}^{\Psi ^{(1)}}}{\left\vert \Lambda
_{L}\right\vert }\cdots \frac{U_{L}^{\Psi ^{(m-1)}}}{\left\vert \Lambda
_{L}\right\vert }\left[ U_{L}^{\Psi ^{(m)}},A\right] \frac{U_{L}^{\Psi
^{(m+1)}}}{\left\vert \Lambda _{L}\right\vert }\cdots \frac{U_{L}^{\Psi
^{(n)}}}{\left\vert \Lambda _{L}\right\vert }\ .
\label{equation commutators}
\end{eqnarray}%
(Compare with Definitions \ref{definition long range energy} and \ref%
{derivation long range}.) Note that the element (\ref{equation commutators})
of $\mathcal{U}$ is uniformly bounded with respect to $L\in \mathbb{N}$,
since, by Corollary \ref{Lemma cigare1}, the commutators in the right-hand
side of this last equation have a limit in $\mathcal{U}$ for any $A\in 
\mathcal{U}_{0}$, as $L\rightarrow \infty $. However, $|\Lambda
_{L}|^{-1}U_{L}^{\Psi }$ does not generally converge in the norm sense of $%
\mathcal{U}$: Since 
\begin{equation*}
\lim_{L\rightarrow \infty }\frac{1}{\left\vert \Lambda _{L}\right\vert }%
\left[ U_{L}^{\Psi },A\right] =0\ ,\qquad A\in \mathcal{U}\ ,
\end{equation*}%
if the sequence $(|\Lambda _{L}|^{-1}U_{L}^{\Psi })_{L\in \mathbb{N}}$ would
converge in $\mathcal{U}$, as $L\rightarrow \infty $, then its limit would
be an element of the center of $\mathcal{U}$. For $\mathcal{U}$ is a simple
algebra (cf. \cite[Corollary 2.6.19]{BrattelliRobinsonI}), its center is
trivial. Therefore, $(|\Lambda _{L}|^{-1}U_{L}^{\Psi })_{L\in \mathbb{N}}$
would converge to $c\mathfrak{1}$ for some $c\in \mathbb{C}$. In particular,
by Proposition \ref{density of periodic states copy(2)}, for $\Psi \in 
\mathcal{W}_{1}$,%
\begin{equation*}
\left( \mathfrak{e}_{\Psi ,\vec{\ell}}-c\mathfrak{1}\right) \in
\bigcap_{\rho \in E_{\vec{\ell}}}\mathrm{ker}\rho \ ,
\end{equation*}%
which is clearly wrong, in general. This observation is well-known. See,
e.g., \cite[p. 2225]{Bona88}.

By contrast, taking $\Psi \in \mathcal{W}_{1}$, $\vec{\ell}\in \mathbb{N}%
^{d} $ and any cyclic representation $\left( \mathcal{H}_{\rho },\pi _{\rho
},\Omega _{\rho }\right) $ of an \emph{extreme} (or ergodic) $\vec{\ell}$%
-periodic state $\rho \in E_{\vec{\ell}}$ (\ref{periodic invariant states}),
one has that the sequence $(|\Lambda _{L}|^{-1}\pi _{\rho }(U_{L}^{\Psi
}))_{L\in \mathbb{N}}$ does strongly converge to $\pi _{\rho }(c)$ for some $%
c=c_{\rho }\in \mathbb{C}$. This is basically Haag's argument \cite{haag62}
proposed in 1962 in order to give a mathematical meaning to the dynamics of
the BCS model. In the more general case of a not necessarily extreme state $%
\rho \in E_{\vec{\ell}}$, the strong operator limit of the sequence is an
element of the (possibly non-trivial) center of the von Neumann algebra $\pi
_{\rho }(\mathcal{U})^{\prime \prime }$. These facts compel us to consider a
more general setting, in particular the notion of \emph{state-dependent}
observables and interactions.

\section{$C^{\ast }$-Algebra of Continuous Functions on States}

In this section, we define the extended quantum framework introduced in \cite%
{Bru-pedra-MF-I}, starting with the classical $C^{\ast }$-algebra. Note that
we consider here the $C^{\ast }$-algebra $\mathcal{U}$, with state space $E$%
, as the primordial $C^{\ast }$-algebra $\mathcal{X}$ of \cite%
{Bru-pedra-MF-I}. Nevertheless, from the point of view of physics, only the
subalgebra $\mathcal{U}^{+}\subseteq \mathcal{U}$ of even elements is
relevant, as already explained in Section \ref{even}. In this case, the
physical state space is the set $E^{+}\subseteq E$ of all even states.

\subsection{The Classical $C^{\ast }$-Algebra of Continuous Functions on
States}

Recall that $E$ stands for the metrizable, weak$^{\ast }$-compact and convex
set of states on $\mathcal{U}$, as defined by (\ref{states CAR}). It is the
state space of the classical dynamics defined on the space $C(E;\mathbb{C})$
of complex-valued weak$^{\ast }$-continuous functions on $E$:\bigskip

\noindent \underline{Classical\ algebra:} Endowed with the point-wise
operations and complex conjugation, $C(E;\mathbb{C})$ becomes a unital
commutative $C^{\ast }$-algebra denoted by%
\begin{equation}
\mathfrak{C}\doteq \left( C\left( E;\mathbb{C}\right) ,+,\cdot _{{\mathbb{C}}%
},\times ,\overline{\left( \cdot \right) },\left\Vert \cdot \right\Vert _{%
\mathfrak{C}}\right) \ ,  \label{metaciagre set 2}
\end{equation}%
where the corresponding $C^{\ast }$-norm is 
\begin{equation}
\left\Vert f\right\Vert _{\mathfrak{C}}\doteq \max_{\rho \in E}\left\vert
f\left( \rho \right) \right\vert \ ,\qquad f\in \mathfrak{C}\ .
\label{metaciagre set 2bis}
\end{equation}%
Note that the \textquotedblleft $\max $\textquotedblright\ in the definition
of the norm\ is well-defined because of the continuity of $f$ together with
the compactness of $E$. $\mathfrak{C}$ is the classical $C^{\ast }$-algebra
of weak$^{\ast }$-continuous complex-valued functions on states. The (real)
Banach subspace of all real-valued functions is denoted by $\mathfrak{C}^{%
\mathbb{R}}\varsubsetneq \mathfrak{C}$. The $C^{\ast }$-algebra $\mathfrak{C}
$ is separable, $E$ being metrizable and compact.\bigskip

\noindent \underline{Gelfand transform:} Elements of the (separable and
unital)$\ C^{\ast }$-algebra $\mathcal{U}$ naturally define continuous and
affine functions $\hat{A}\in \mathfrak{C}$ by 
\begin{equation}
\hat{A}\left( \rho \right) \doteq \rho \left( A\right) \ ,\qquad \rho \in
E,\ A\in \mathcal{U}\ .  \label{fA}
\end{equation}%
This is the well-known Gelfand transform. Note that $A\neq B$ implies $\hat{A%
}\neq \hat{B}$, as states separates elements of $\mathcal{U}$. Since 
\begin{equation}
\left\Vert A\right\Vert _{\mathcal{U}}=\max_{\rho \in E}\left\vert \rho
\left( A\right) \right\vert \ ,\qquad A\in \mathcal{U}^{\mathbb{R}}\ ,
\label{norm properties}
\end{equation}%
the mapping $A\mapsto \hat{A}$ defines a linear isometry from the Banach
space $\mathcal{U}^{\mathbb{R}}$ of all self-adjoint elements to the space $%
\mathfrak{C}^{\mathbb{R}}$ of all real-valued functions on $E$. \bigskip

\noindent \underline{Dense classical subalgebra:} Recall that $\mathcal{U}%
_{0}$ is the normed $\ast $-algebra of local elements of $\mathcal{U}$
defined by (\ref{simple}). We denote respectively by%
\begin{equation}
\mathfrak{C}_{\mathcal{U}_{0}}\doteq \mathbb{C}[\{\hat{A}:A\in \mathcal{U}%
_{0}\}]\qquad \text{and}\qquad \mathfrak{C}_{\mathcal{U}}\doteq \mathbb{C}[\{%
\hat{A}:A\in \mathcal{U}\}]  \label{CU0}
\end{equation}%
the subalgebras of polynomials in the elements of $\{\hat{A}:A\in \mathcal{U}%
_{0}\}$ and $\{\hat{A}:A\in \mathcal{U}\}$, with complex coefficients. The
unit $\mathfrak{\hat{1}}\in \mathfrak{C}$ mapping any state to $1$ belongs
to $\mathfrak{C}_{\mathcal{U}_{0}}$, by definition. Since $\mathcal{U}_{0}$
is, by construction, dense in $\mathcal{U}$, the subalgebra $\mathfrak{C}_{%
\mathcal{U}_{0}}$ separates states. Therefore, by the Stone-Weierstrass
theorem, $\mathfrak{C}_{\mathcal{U}_{0}}\subseteq \mathfrak{C}_{\mathcal{U}}$
is dense in $\mathfrak{C}$, i.e., $\mathfrak{C}=\overline{\mathfrak{C}_{%
\mathcal{U}_{0}}}$.

\subsection{Poisson Structure Associated with the State Space\label{Poisson
Structure}}

We define $\mathcal{A}\left( E;\mathbb{C}\right) \subseteq \mathfrak{C}$ to
be the closed subspace of all \emph{affine}, weak$^{\ast }$-continuous
complex-valued functions over $E$. By \cite[Definition 3.7]{Bru-pedra-MF-I},
the convex G\^{a}teaux derivative of $f\in \mathfrak{C}$ at a fixed state $%
\rho \in E$ is an affine weak$^{\ast }$-continuous complex-valued function
over $E$, defined as follows:

\begin{definition}[Convex weak$^{\ast }$-continuous G\^{a}teaux derivative]
\label{convex Frechet derivative}\mbox{ }\newline
For any $f\in \mathfrak{C}$ and $\rho \in E$, we say that $\mathrm{d}f\left(
\rho \right) :E\rightarrow \mathbb{C}$ is the (unique) convex weak$^{\ast }$%
-continuous G\^{a}teaux derivative of $f$ at $\rho \in E$ if $\mathrm{d}%
f\left( \rho \right) \in \mathcal{A}(E;\mathbb{C})$ and%
\begin{equation*}
\lim_{\lambda \rightarrow 0^{+}}\lambda ^{-1}\left( f\left( \left( 1-\lambda
\right) \rho +\lambda \upsilon \right) -f\left( \rho \right) \right) =\left[ 
\mathrm{d}f\left( \rho \right) \right] \left( \upsilon \right) \ ,\qquad
\rho ,\upsilon \in E\ .
\end{equation*}
\end{definition}

A function $f\in \mathfrak{C}$ such that $\mathrm{d}f\left( \rho \right) $
exists for all $\rho \in E$ is called (convex-)\emph{differentiable} and we
use the notation 
\begin{equation*}
\mathrm{d}f\equiv (\mathrm{d}f\left( \rho \right) )_{\rho \in
E}:E\rightarrow \mathcal{A}(E;\mathbb{C})\ .
\end{equation*}%
If $f\in \mathcal{A}(E;\mathbb{C})$ then 
\begin{equation*}
\mathrm{d}f\left( \rho \right) =f-f\left( \rho \right) \ ,\qquad \rho \in E\
,
\end{equation*}%
which means that affine functions of $\mathfrak{C}$ are continuously
(convex-)differentiable, as expected.

We define the (non-empty) subspace of continuously differentiable
complex-valued functions over the convex and weak$^{\ast }$-compact set $E$
by 
\begin{equation}
\mathfrak{Y}\equiv \mathfrak{Y}\left( E\right) \doteq \left\{ f\in \mathfrak{%
C}:\mathrm{d}f\in C\left( E;\mathfrak{C}\right) \right\} \ .  \label{C1}
\end{equation}%
We endow this vector space with the norm%
\begin{equation}
\left\Vert f\right\Vert _{\mathfrak{Y}}\doteq \left\Vert f\right\Vert _{%
\mathfrak{C}}+\max_{\rho \in E}\left\Vert \mathrm{d}f\left( \rho \right)
\right\Vert _{\mathfrak{C}}<\infty \ ,\qquad f\in \mathfrak{Y}\ ,
\label{C1bis}
\end{equation}%
in order to obtain a Banach space, again denoted by $\mathfrak{Y}$. The
\textquotedblleft $\max $\textquotedblright\ in the definition of the norm
is well-defined because of the continuity of $f$ and $\mathrm{d}f$ together
with the compactness of $E$. The normed vector space $\mathfrak{Y}$ is
complete, see \cite[Section 3.4]{Bru-pedra-MF-I}.

The (real) Banach subspace of all continuously differentiable real-valued
functions is denoted by $\mathfrak{Y}^{\mathbb{R}}\varsubsetneq \mathfrak{Y}$%
. By \cite[Proposition 3.8]{Bru-pedra-MF-I} together with Equations (\ref%
{metaciagre set 2bis})-(\ref{norm properties}) and (\ref{C1}), for any $f\in 
\mathfrak{Y}^{\mathbb{R}}\varsubsetneq \mathfrak{C}^{\mathbb{R}}$, there is
a unique $\mathrm{D}f\in C(E;\mathcal{U}^{\mathbb{R}})$ such that%
\begin{equation}
\mathrm{d}f\left( \rho \right) =\widehat{\mathrm{D}f\left( \rho \right) }\
,\qquad \rho \in E\ .  \label{clear2}
\end{equation}%
By (\ref{metaciagre set 2bis}) and (\ref{norm properties}), note that%
\begin{equation}
\left\Vert \mathrm{D}f\left( \rho \right) \right\Vert _{\mathcal{U}%
}=\left\Vert \mathrm{d}f\left( \rho \right) \right\Vert _{\mathfrak{C}}\
,\qquad \rho \in E\ .  \label{norm affine2}
\end{equation}%
Since the convex weak$^{\ast }$-continuous G\^{a}teaux derivative is linear
and because any function $f\in \mathfrak{Y}$ can be decomposed into real and
imaginary parts, it follows that, for any $f\in \mathfrak{Y}\varsubsetneq 
\mathfrak{C}$, there is a unique $\mathrm{D}f\in C(E;\mathcal{U})$
satisfying (\ref{clear2}). For instance,%
\begin{equation*}
\mathrm{D}\hat{A}\left( \rho \right) =A-\rho \left( A\right) \mathfrak{1}\
,\qquad A\in \mathcal{U}\ .
\end{equation*}%
Therefore, using \cite[Definition 3.9]{Bru-pedra-MF-I} and the standard
notation $[A,B]\doteq AB-BA$, $A,B\in \mathcal{U}$, for the commutator, we
can define a Poisson bracket for continuously differentiable complex-valued
functions on the state space:

\begin{definition}[Poisson bracket]
\label{convex Frechet derivative copy(1)}\mbox{ }\newline
We define the mapping $\{\cdot ,\cdot \}:\mathfrak{Y}\times \mathfrak{Y}%
\rightarrow \mathfrak{C}$ by%
\begin{equation*}
\left\{ f,g\right\} \left( \rho \right) \doteq \rho \left( i\left[ \mathrm{D}%
f\left( \rho \right) ,\mathrm{D}g\left( \rho \right) \right] \right) \
,\qquad f,g\in \mathfrak{Y}\ .
\end{equation*}
\end{definition}

\noindent By \cite[Proposition 3.10 and discussions afterwards]%
{Bru-pedra-MF-I}, this mapping is a Poisson bracket on $\mathfrak{C}_{%
\mathcal{U}}$, that is, for complex-valued polynomials. In other words, it
is a skew-symmetric biderivation on $\mathfrak{C}_{\mathcal{U}}$ satisfying
the Jacobi identity.

In fact, by generalizing the well-known construction of a Poisson bracket
for the polynomial functions on the dual space of finite dimensional Lie
groups \cite[Section 7.1]{Poission}, we define in \cite[Section 3.2]%
{Bru-pedra-MF-I} a Poisson bracket for the polynomial functions on the
hermitian continuous functional (like the states) on any $C^{\ast }$%
-algebra. Then, in \cite[Section 3.3]{Bru-pedra-MF-I}, the Poisson bracket
is localized on the state and phase\footnote{%
The phase space in \cite[Definition 2.2]{Bru-pedra-MF-I} is the weak$^{\ast
} $ closure of the subset $\mathcal{E}(E)$ of extreme points of the state
space $E$. For general $C^{\ast }$-algebras, $E$ and the weak$^{\ast }$
closure of $\mathcal{E}(E)$ are not necessarly the same. But here, if $%
\mathfrak{L}$ is infinite then $E=\overline{\mathcal{E}(E)}$, i.e., phase
and state spaces coincide.} spaces associated with this algebra, by taking
quotients with respect to conveniently chosen Poisson ideals. In particular,
this leads, in an elegant way, to a Poisson bracket for polynomial functions
of the classical $C^{\ast }$-algebra $\mathfrak{C}$. Definitions \ref{convex
Frechet derivative}-\ref{convex Frechet derivative copy(1)} just yield a
very convenient explicit expression of this Poisson bracket for functions on
the state space.

\subsection{The Quantum $C^{\ast }$-Algebra of Continuous Functions on
States \label{State-Dependent Short-Range Interactions copy(1)}}

The long-range dynamics takes place in the space $C(E;\mathcal{U})$ of weak$%
^{\ast }$-continuous $\mathcal{U}$-valued functions on the metrizable
compact space $E$, which can be endowed with a $C^{\ast }$-algebra
structure:\bigskip

\noindent \underline{Quantum\ algebra:} Endowed with the point-wise $\ast $%
-algebra operations inherited from $\mathcal{U}$, $C(E;\mathcal{U})$ is a
unital \emph{non-commutative} $C^{\ast }$-algebra denoted by%
\begin{equation}
\mathfrak{U}\equiv \mathfrak{U}_{\mathfrak{L}}\doteq \left( C\left( E;%
\mathcal{U}\right) ,+,\cdot _{{\mathbb{C}}},\times ,^{\ast },\left\Vert
\cdot \right\Vert _{\mathfrak{U}}\right) \ .  \label{metaciagre set}
\end{equation}%
The unique $C^{\ast }$-norm $\left\Vert \cdot \right\Vert _{\mathfrak{U}}$
is the supremum norm for functions on $E$ taking values in the normed space $%
\mathcal{U}$ , i.e., 
\begin{equation}
\left\Vert f\right\Vert _{\mathfrak{U}}\doteq \max_{\rho \in E}\left\Vert
f\left( \rho \right) \right\Vert _{\mathcal{U}}\ ,\qquad f\in \mathfrak{U}\ .
\label{metaciagre set bis}
\end{equation}%
Recall that $\mathcal{U}^{\mathbb{R}}\varsubsetneq \mathcal{U}$ is the
(real) Banach subspace of all self-adjoint elements of $\mathcal{U}$. The
(real) Banach subspace of all $\mathcal{U}^{\mathbb{R}}$-valued functions of 
$\mathfrak{U}$ is similarly denoted by $\mathfrak{U}^{\mathbb{R}%
}\varsubsetneq \mathfrak{U}$.

We identify the primordial $C^{\ast }$-algebra $\mathcal{U}$, on which the
quantum dynamics is usually defined, with the subalgebra of constant
functions of $\mathfrak{U}$. Meanwhile, the classical dynamics appears in
the algebra $\mathfrak{C}$ of complex-valued weak$^{\ast }$-continuous
functions on $E$. This unital commutative $C^{\ast }$-algebra is identified
with the subalgebra of functions of $\mathfrak{U}$ whose values are
multiples of the unit $\mathfrak{1}\in \mathcal{U}$. In other words, we have
the canonical inclusions%
\begin{equation}
\mathcal{U}\subseteq \mathfrak{U}\qquad \text{and}\qquad \mathfrak{%
C\subseteq U}\ .  \label{subset}
\end{equation}%
See \cite[Eq. (67)]{Bru-pedra-MF-I}.\bigskip

\noindent \underline{Dense subalgebras:} Similar to (\ref{simple}), we
define the $\ast $-subalgebras%
\begin{equation}
\mathfrak{U}_{\Lambda }\doteq \left\{ f\in \mathfrak{U}:f\left( E\right)
\subseteq \mathcal{U}_{\Lambda }\right\} \ ,\qquad \Lambda \in \mathcal{P}%
_{f}\ ,  \label{local elements}
\end{equation}%
and 
\begin{equation}
\mathfrak{U}_{0}\doteq \bigcup\limits_{L\in \mathbb{N}}\mathfrak{U}_{\Lambda
_{L}}\subseteq \mathfrak{U}\ .  \label{U frac 0}
\end{equation}%
The union $\mathfrak{C}_{\mathcal{U}_{0}}\cup \mathcal{U}_{0}$ generates the 
$C^{\ast }$-algebra $\mathfrak{U}$ and 
\begin{equation}
\mathfrak{U}_{0}=\mathrm{span}\left\{ \mathfrak{C}\mathcal{U}_{0}\right\}
\supseteq \mathrm{span}\left\{ \mathfrak{C}_{\mathcal{U}_{0}}\mathcal{U}%
_{0}\right\}  \label{spain}
\end{equation}%
are dense $\ast $-subalgebras of $\mathfrak{U}$. To see this, use the
density of $\mathcal{U}_{0}\subseteq \mathcal{U}$ and $\mathfrak{C}_{%
\mathcal{U}_{0}}$, as well as the compactness of $E$ together with the
existence of partitions of unity subordinated to any open cover of the
metrizable (weak$^{\ast }$-compact) space $E$ (paracompactness and
metrizability of $E$).\bigskip

\noindent \underline{Positivity of elements of $\mathfrak{U}$:} The
positivity of elements of $\mathfrak{U}$ is equivalent to their point-wise
positivity. This is a direct consequence of the following lemma:

\begin{lemma}[Spectrum]
\label{Spectrum}\mbox{ }\newline
For any $f\in \mathfrak{U}$, its spectrum equals 
\begin{equation*}
\mathrm{spec}\left( f\right) =\bigcup\limits_{\rho \in E}\mathrm{spec}\left(
f\left( \rho \right) \right) \ .
\end{equation*}
\end{lemma}

\begin{proof}
Fix $f\in \mathfrak{U}$ and take any $z\in \mathbb{C}$ in the resolvent set
of $f$. Then, clearly, $z$ also belongs to the resolvent set of $f(\rho )\in 
\mathcal{U}$ for all $\rho \in E$. It follows that 
\begin{equation}
\bigcup\limits_{\rho \in E}\mathrm{spec}\left( f\left( \rho \right) \right)
\subseteq \mathrm{spec}\left( f\right) \ .  \label{sssss}
\end{equation}%
Take now $z\in \mathbb{C}$ in the resolvent set of $f(\rho )\in \mathcal{U}$
for all $\rho \in E$. Using the Neumann series 
\begin{equation*}
\left( A-B-z\mathbf{1}_{\mathcal{U}}\right) ^{-1}=\left( z\mathbf{1}_{%
\mathcal{U}}-A\right) ^{-1}\sum_{n=1}^{\infty }\left( B\left( z\mathbf{1}_{%
\mathcal{U}}-A\right) ^{-1}\right) ^{n}
\end{equation*}%
for all $A,B\in \mathcal{U}$ with $\Vert B\left( z\mathbf{1}_{\mathcal{U}%
}-A\right) ^{-1}\Vert _{\mathcal{U}}<1$, one sees that the mapping 
\begin{equation*}
\rho \mapsto \left( z\mathbf{1}_{\mathcal{U}}-f\left( \rho \right) \right)
^{-1}
\end{equation*}%
from $E$ to $\mathcal{U}$ is weak$^{\ast }$-continuous. In particular, this
mapping is an element of $\mathfrak{U}$ which, by construction, is the
resolvent of $f$ at $z\in \mathbb{C}$. It follows that 
\begin{equation}
\mathrm{spec}\left( f\right) \subseteq \bigcup\limits_{\rho \in E}\mathrm{%
spec}\left( f\left( \rho \right) \right) \ .  \label{sssssssssssssss}
\end{equation}%
By (\ref{sssss})-(\ref{sssssssssssssss}), the assertion follows.
\end{proof}

\begin{corollary}[Positivity]
\label{Spectrum copy(1)}\mbox{ }\newline
Any element $f\in \mathfrak{U}^{\mathbb{R}}$ is positive iff $f(\rho )\in 
\mathcal{U}^{\mathbb{R}}$ is positive for all $\rho \in E$.
\end{corollary}

\noindent \underline{Functional calculus for elements of $\mathfrak{U}$:} It
turns out that the continuous functional calculus in $\mathfrak{U}$
coincides with the point-wise continuous functional calculus in $\mathcal{U}$%
:

\begin{lemma}[Continuous functional calculus]
\label{Spectrum copy(2)}\mbox{ }\newline
For any self-adjoint $f\in \mathfrak{U}^{\mathbb{R}}$ with spectrum $\mathrm{%
spec}(f)$ and any continuous function $\varphi \in C(\mathrm{spec}(f);%
\mathbb{C})$, $\varphi (f)=\varphi \circ f$, where%
\begin{equation*}
\varphi \circ f\left( \rho \right) \doteq \varphi \left( f\left( \rho
\right) \right) \ ,\qquad \rho \in E\ .
\end{equation*}%
Note that $\varphi \left( f\left( \rho \right) \right) $ is well-defined
because $\mathrm{spec}\left( f\left( \rho \right) \right) \subseteq \mathrm{%
spec}(f)$ for all $\rho \in E$, by Lemma \ref{Spectrum}.
\end{lemma}

\begin{proof}
Note first that, for all self-adjoint $f\in \mathfrak{U}^{\mathbb{R}}$ and
any polynomial function $\varphi \in C(\mathrm{spec}(f);\mathbb{C})$, $%
\varphi \left( f\right) =\varphi \circ f\in \mathfrak{U}$. Observe next
that, for any $\varphi _{1},\varphi _{2}\in C(\mathrm{spec}(f);\mathbb{C})$
and $\rho \in E$, 
\begin{equation}
\left\Vert \varphi _{1}\circ f\left( \rho \right) -\varphi _{2}\circ f\left(
\rho \right) \right\Vert _{\mathcal{U}}\leq \sup_{\rho \in \mathrm{spec}%
\left( f\left( \rho \right) \right) }\left\vert \varphi _{1}\left( \rho
\right) -\varphi _{2}\left( \rho \right) \right\vert \leq \left\Vert \varphi
_{1}-\varphi _{2}\right\Vert _{\infty }\ ,  \label{using be}
\end{equation}%
by Lemma \ref{Spectrum}. If $\varphi \in C(\mathrm{spec}(f);\mathbb{C})$ is
a general continuous function then take any sequence of polynomial functions 
$\varphi _{n}\in C(\mathrm{spec}(f);\mathbb{C})$ converging uniformly to $%
\varphi $. Such a sequence always exists, by the Stone-Weierstrass theorem
and the compactness of $\mathrm{spec}(f)$. Therefore, we infer from (\ref%
{using be}) that $\varphi \circ f$ is the uniform limit of the sequence of
weak$^{\ast }$-continuous function $(\varphi _{n}\circ f)_{n\in \mathbb{N}%
}\subseteq \mathfrak{U}$. In particular, $\varphi \circ f$ is weak$^{\ast }$%
-continuous, i.e., $\varphi \circ f\in \mathfrak{U}$. It is easy to check
that the mapping $\varphi \mapsto \varphi \circ f$ is a $\ast $-homomorphism
from $C(\mathrm{spec}(f);\mathbb{C})$ to $\mathfrak{U}$ with $1\circ f=%
\mathfrak{1}$ being the unit of $\mathfrak{U}$ and \textrm{id}$_{\mathrm{spec%
}(f)}\circ f=f$. By the uniqueness of the continuous functional calculus, $%
\varphi \circ f=\varphi (f)$ for all $f\in \mathfrak{U}^{\mathbb{R}}$ and $%
\varphi \in C(\mathrm{spec}(f);\mathbb{C})$.
\end{proof}

\subsection{Important $\ast $-Automorphisms of the Quantum $C^{\ast }$%
-Algebra}

\noindent \underline{Parity:} The $\ast $-automorphism $\sigma $ of $%
\mathcal{U}$ uniquely defined by the condition (\ref{automorphism gauge
invariance}) naturally induces a $\ast $-automorphism $\Xi $ of $\mathfrak{U}
$ defined by 
\begin{equation}
\left[ \Xi \left( f\right) \right] \left( \rho \right) \doteq \sigma \left(
f\left( \rho \right) \right) \ ,\qquad \rho \in E,\ f\in \mathfrak{U}\ .
\label{gauge invariant}
\end{equation}

\noindent Elements $f_{1},f_{2}\in \mathcal{U}$ satisfying $\Xi
(f_{1})=f_{1} $ and $\Xi (f_{2})=-f_{2}$ are respectively called \emph{even}
and \emph{odd}. The set 
\begin{equation}
\mathfrak{U}^{+}\equiv \mathfrak{U}_{\mathfrak{L}}^{+}\doteq \{f\in 
\mathfrak{U}:f=\Xi (f)\}=\left\{ f\in \mathfrak{U}:f\left( E\right)
\subseteq \mathcal{U}^{+}\right\} \subseteq \mathfrak{U}
\label{even state dependent}
\end{equation}%
of all even weak$^{\ast }$-continuous $\mathcal{U}$-valued functions on
states is a $C^{\ast }$-subalgebra. Compare with (\ref{definition of even
operators}). \bigskip

\noindent \underline{Translations:} Let $\mathfrak{L}=\mathbb{Z}^{d}$. The $%
\ast $-automorphisms $\alpha _{x}$, $x\in \mathbb{Z}^{d}$, of $\mathcal{U}$
uniquely defined by the condition (\ref{transl}) naturally induce a group
homomorphism $x\mapsto \mathrm{A}_{x}$ from $\mathbb{Z}^{d}$ to the group of 
$\ast $-automorphisms of $\mathfrak{U}$, defined by 
\begin{equation}
\left[ \mathrm{A}_{x}\left( f\right) \right] \left( \rho \right) \doteq
\alpha _{x}\left( f\left( \rho \right) \right) \ ,\qquad \rho \in E,\ f\in 
\mathfrak{U},\ x\in \mathbb{Z}^{d}\ .  \label{translatbis}
\end{equation}%
These $\ast $-automorphisms represent the translation group in $\mathfrak{U}$%
. \bigskip

\noindent \underline{Permutations:} In the same way, we can define a group
homomorphism $\pi \mapsto \mathfrak{P}_{\pi }$ from the set $\Pi $ of all
bijective mappings from $\mathfrak{L}$ into itself, which leave all but
finitely many elements invariant, to the group of $\ast $-automorphisms of $%
\mathfrak{U}$ by using (\ref{definition perm automorphism}) and the
definition 
\begin{equation*}
\left[ \mathfrak{P}_{\pi }\left( f\right) \right] \left( \rho \right) \doteq 
\mathfrak{p}_{\pi }\left( f\left( \rho \right) \right) \ ,\qquad \rho \in
E,\ f\in \mathfrak{U},\ \pi \in \Pi \ .
\end{equation*}%
The $\ast $-automorphisms representing the permutation group in $\mathfrak{U}
$ are not used here and are only given for completeness.

\section{Limit Long-Range Dynamics -- Classical Part\label{sect
Lieb--Robinson copy(3)}}

\subsection{Banach Spaces of State-Dependent Short-Range Interactions\label%
{State-Dependent Short-Range Interactions}}

Similar to what is done in Section \ref{Section Banach space interaction}, a
state-dependent (complex) interaction is defined to be a mapping $\mathbf{%
\Phi }:\mathcal{P}_{f}\rightarrow \mathfrak{U}^{+}$ such that $\Phi
_{\Lambda }\in \mathfrak{U}_{\Lambda }$ for any $\Lambda \in \mathcal{P}_{f}$%
. See Equations (\ref{local elements}) and (\ref{even state dependent}).
Similar to (\ref{vector space interaction}) and (\ref{involution}), the set
of all state-dependent interactions is naturally endowed with the structure
of a complex vector space on which the natural involution 
\begin{equation}
\mathbf{\Phi }\mapsto \mathbf{\Phi }^{\ast }\doteq (\mathbf{\Phi }_{\Lambda
}^{\ast })_{\Lambda \in \mathcal{P}_{f}}  \label{involutioninvolution}
\end{equation}%
is defined. Self-adjoint state-dependent interactions $\mathbf{\Phi }$ are,
by definition, those satisfying $\mathbf{\Phi }=\mathbf{\Phi }^{\ast }$. A
state-dependent interaction $\mathbf{\Phi }$ can be identified with a
mapping $\rho \mapsto \mathbf{\Phi }\left( \rho \right) $ from $E$ to the
vector space of usual interactions (of Section \ref{Section Banach space
interaction}) via the definition 
\begin{equation*}
\mathbf{\Phi }\left( \rho \right) _{\Lambda }\doteq \mathbf{\Phi }_{\Lambda
}\left( \rho \right) \ ,\qquad \Lambda \in \mathcal{P}_{f}\ .
\end{equation*}

In this paper, we only consider a particular space of state-dependent
short-range interactions: Using the Banach space $\mathcal{W}$ of (usual)
short-range interactions, we define%
\begin{equation*}
\mathfrak{W}\doteq \left( C\left( E;\mathcal{W}\right) ,+,\cdot _{{\mathbb{C}%
}},^{\ast },\left\Vert \cdot \right\Vert _{\mathfrak{W}}\right)
\end{equation*}%
to be the Banach space of weak$^{\ast }$-continuous, state-dependent
short-range interactions, along with the supremum norm%
\begin{equation}
\left\Vert \mathbf{\Phi }\right\Vert _{\mathfrak{W}}\doteq \max_{\rho \in
E}\left\Vert \mathbf{\Phi }\left( \rho \right) \right\Vert _{\mathcal{W}}\
,\qquad \mathbf{\Phi }\in \mathfrak{W}\ ,  \label{iteration0bis}
\end{equation}%
where $\Vert \cdot \Vert _{\mathcal{W}}$ is defined by (\ref{iteration0}).
Note that (\ref{involutioninvolution}) is an isometric antilinear involution
on $\mathfrak{W}$. Recall that $\mathcal{W}^{\mathbb{R}}\varsubsetneq 
\mathcal{W}$ is the (real) Banach subspace of all self-adjoint interactions
and the (real) Banach subspace of all self-adjoint state-dependent
interactions is similarly denoted by 
\begin{equation*}
\mathfrak{W}^{\mathbb{R}}\equiv \left( C\left( E;\mathcal{W}^{\mathbb{R}%
}\right) ,+,\cdot _{{\mathbb{R}}},\left\Vert \cdot \right\Vert _{\mathfrak{W}%
}\right) \varsubsetneq \mathfrak{W}\ .
\end{equation*}%
To simplify notation, for any $\mathbf{\Psi }\in C(\mathbb{R};\mathfrak{W})$
and $\rho \in E$, $\mathbf{\Psi }\left( \rho \right) \in C(\mathbb{R};%
\mathcal{W})$ stands for the time-dependent interaction defined by 
\begin{equation}
\mathbf{\Psi }\left( \rho \right) \left( t\right) \doteq \mathbf{\Psi }%
\left( t;\rho \right) \ ,\qquad \rho \in E,\ t\in \mathbb{R}\ .
\label{notation state interactionbis}
\end{equation}

When $\mathfrak{L}=\mathbb{Z}^{d}$, $\mathbf{\Phi }\in \mathfrak{W}$ is, by
definition, translation-invariant\ if, for all $x\in \mathbb{Z}^{d}$ and $%
\Lambda \in \mathcal{P}_{f}$, $\mathbf{\Phi }_{\Lambda +x}=\mathrm{A}_{x}(%
\mathbf{\Phi }_{\Lambda })$, see (\ref{translation box}). Recall that $\{%
\mathrm{A}_{x}\}_{x\in \mathbb{Z}^{d}}$ is the family of (translation) $\ast 
$-automorphisms on $\mathfrak{U}$ defined by (\ref{translatbis}). Similar to
the (separable) Banach subspace $\mathcal{W}_{1}\varsubsetneq \mathcal{W}$
of translation-invariant, short-range interactions on $\mathfrak{L}=\mathbb{Z%
}^{d}$, we denote by $\mathfrak{W}_{1}\varsubsetneq \mathfrak{W}$ the Banach
subspace of translation-invariant, state-dependent, short-range interactions
on $\mathfrak{L}=\mathbb{Z}^{d}$.

\subsection{Derivations on the Quantum $C^{\ast }$-Algebra of Functions and
Dynamics}

For any state-dependent (short-range) interaction $\mathbf{\Phi }\in 
\mathfrak{W}$, we can naturally define limit derivations in the quantum $%
C^{\ast }$-algebra $\mathfrak{U}$:

\begin{definition}[Derivations for state-dependant interactions]
\label{symmetric derivation}\mbox{ }\newline
The symmetric derivations $\mathbf{\delta }^{\mathbf{\Phi }}$ associated
with any $\mathbf{\Phi }\in \mathfrak{W}$ is defined on the dense subset $%
\mathfrak{U}_{0}=\mathrm{span}\left\{ \mathfrak{C}\mathcal{U}_{0}\right\} $
(see (\ref{spain})) by%
\begin{equation*}
\left[ \mathbf{\delta }^{\mathbf{\Phi }}(fA)\right] \left( \rho \right)
\doteq f\left( \rho \right) \delta ^{\mathbf{\Phi }\left( \rho \right) }(A)\
,\qquad \rho \in E,\ f\in \mathfrak{C},\ A\in \mathcal{U}_{0}\ .
\end{equation*}
\end{definition}

\noindent The right-hand side of the last equation defines an element of $%
\mathfrak{U}$, by Corollary \ref{Lemma cigare1}. Similar to Remark \ref%
{Remark Closure of limit derivations}, observe also that, if $\Phi \in 
\mathfrak{W}^{\mathbb{R}}$, then the symmetric derivation $\mathbf{\delta }^{%
\mathbf{\Phi }}$ is (norm-) closable: By Corollary \ref{Spectrum copy(1)}
and Lemma \ref{Spectrum copy(2)}, for all $f\in \mathfrak{U}_{0}$, $f\geq 0$
implies $f^{1/2}\in \mathfrak{U}_{0}$. It follows from \cite[Theorem 1.4.9]%
{Bratelli-derivation} that $\mathbf{\delta }^{\mathbf{\Phi }}$ is
dissipative \cite[Definition 1.4.6]{Bratelli-derivation}, and, by \cite[%
Proposition 1.4.7]{Bratelli-derivation}, it thus norm-closable and its
closure is also dissipative.

Any $\mathbf{\Psi }\in C(\mathbb{R};\mathfrak{W}^{\mathbb{R}})$ determines a
two-parameter family $\mathfrak{T}^{\mathbf{\Psi }}\equiv (\mathfrak{T}%
_{t,s}^{\mathbf{\Psi }})_{_{s,t\in \mathbb{R}}}$ of $\ast $-auto%
\-%
morphisms of the quantum $C^{\ast }$-algebra $\mathfrak{U}$ defined by (\ref%
{metaciagre set}):%
\begin{equation}
\left[ \mathfrak{T}_{t,s}^{\mathbf{\Psi }}\left( f\right) \right] \left(
\rho \right) \doteq \tau _{t,s}^{\mathbf{\Psi }\left( \rho \right) }\left(
f\left( \rho \right) \right) \ ,\qquad \rho \in E,\ f\in \mathfrak{U},\
s,t\in \mathbb{R}\ .  \label{definiotion tho frac}
\end{equation}%
For all $f\in \mathfrak{U}_{0}$, the right -hand side of (\ref{definiotion
tho frac}) defines an element of $\mathfrak{U}$, by Proposition \ref{Theorem
Lieb-Robinson copy(3)} (iii) and Lebesgue's dominated convergence theorem.
By density of $\mathfrak{U}_{0}\subseteq \mathfrak{U}$ and the fact that $%
\tau _{t,s}^{\mathbf{\Psi }\left( \rho \right) }$ is a contraction, it
follows that the right-hand side of (\ref{definiotion tho frac}) defines an
element of $\mathfrak{U}$ for all $f\in \mathfrak{U\varsupsetneq U}_{0}$.
Similar to Proposition \ref{Theorem Lieb-Robinson}, this family satisfies a
non-autonomous evolution equation with infinitesimal generator $\mathbf{%
\delta }^{\mathbf{\Psi }(t)}$ for $t\in {\mathbb{R}}$:

\begin{proposition}[Infinite-volume state-dependent short-range dynamics]
\label{Lemma cigare1 copy(5)}\mbox{
}\newline
For any $\mathbf{\Psi }\in C(\mathbb{R};\mathfrak{W}^{\mathbb{R}})$, $%
\mathfrak{T}^{\mathbf{\Psi }}\equiv (\mathfrak{T}_{t,s}^{\mathbf{\Psi }%
})_{_{s,t\in \mathbb{R}}}$ is a strongly continuous two-para%
\-%
meter family of $\ast $-auto%
\-%
morphisms of $\mathfrak{U}$, which is the unique solution in $\mathcal{B}(%
\mathfrak{U})$ to the non-auto%
\-%
nomous evolution equation%
\begin{equation}
\forall s,t\in {\mathbb{R}}:\qquad \partial _{t}\mathfrak{T}_{t,s}^{\mathbf{%
\Psi }}=\mathfrak{T}_{t,s}^{\mathbf{\Psi }}\circ \mathbf{\delta }^{\mathbf{%
\Psi }\left( t\right) }\ ,\qquad \mathfrak{T}_{s,s}^{\mathbf{\Psi }}=\mathbf{%
1}_{\mathfrak{U}}\ ,  \label{klsdf}
\end{equation}%
in the strong sense on the dense subspace $\mathfrak{U}_{0}\subseteq 
\mathfrak{U}$, $\mathbf{1}_{\mathfrak{U}}$ being the identity mapping of $%
\mathfrak{U}$. In particular, it satisfies the reverse cocycle property:%
\begin{equation}
\forall s,r,t\in \mathbb{R}:\qquad \mathfrak{T}_{t,s}^{\mathbf{\Psi }}=%
\mathfrak{T}_{r,s}^{\mathbf{\Psi }}\mathfrak{T}_{t,r}^{\mathbf{\Psi }}\ .
\label{reverse cocycle0}
\end{equation}
\end{proposition}

\begin{proof}
Fix $\mathbf{\Psi }\in C(\mathbb{R};\mathfrak{W}^{\mathbb{R}})$. The fact
that $(\mathfrak{T}_{t,s}^{\mathbf{\Psi }})_{_{s,t\in \mathbb{R}}}$ is a
family of $\ast $-auto%
\-%
morphisms of $\mathfrak{U}$ satisfying (\ref{reverse cocycle0}) is a direct
consequence of (\ref{definiotion tho frac}), since $(\tau _{t,s}^{\mathbf{%
\Psi }\left( \rho \right) })_{_{s,t\in \mathbb{R}}}$ is a family of $\ast $%
-auto%
\-%
morphisms of $\mathcal{U}$ satisfying (\ref{reverse cocycle}) for any $\rho
\in E$.

We now prove that the family $(\mathfrak{T}_{t,s}^{\mathbf{\Psi }%
})_{_{s,t\in \mathbb{R}}}$ is strongly continuous: By (\ref{definiotion tho
frac}) and \cite[Lemma 5.1 (iii)]{Bru-pedra-MF-I}, it suffice to prove that $%
(\tau _{t,s}^{\mathbf{\Psi }\left( \rho \right) })_{\left( \rho ,s,t\right)
\in E\times \mathbb{R}^{2}}$ is a strongly continuous family. To this end,
take three sequences $(s_{n})_{n\in \mathbb{N}},(t_{n})_{n\in \mathbb{N}%
}\subseteq \mathbb{R}$ and $(\rho _{n})_{n\in \mathbb{N}}\subseteq E$
converging respectively to $s,t\in \mathbb{R}$ and $\rho \in E$. For any $%
A\in \mathcal{U}$,%
\begin{equation}
\left\Vert \tau _{t_{n},s_{n}}^{\mathbf{\Psi }\left( \rho _{n}\right)
}\left( A\right) -\tau _{t,s}^{\mathbf{\Psi }\left( \rho \right) }\left(
A\right) \right\Vert _{\mathcal{U}}\leq \left\Vert \tau _{t_{n},s_{n}}^{%
\mathbf{\Psi }\left( \rho _{n}\right) }\left( A\right) -\tau _{t_{n},s_{n}}^{%
\mathbf{\Psi }\left( \rho \right) }\left( A\right) \right\Vert _{\mathcal{U}%
}+\left\Vert \tau _{t_{n},s_{n}}^{\mathbf{\Psi }\left( \rho \right) }\left(
A\right) -\tau _{t,s}^{\mathbf{\Psi }\left( \rho \right) }\left( A\right)
\right\Vert _{\mathcal{U}}\ .  \label{cigare state1}
\end{equation}%
By Proposition \ref{Theorem Lieb-Robinson},%
\begin{equation}
\lim_{n\rightarrow \infty }\left\Vert \tau _{t_{n},s_{n}}^{\mathbf{\Psi }%
\left( \rho \right) }\left( A\right) -\tau _{t,s}^{\mathbf{\Psi }\left( \rho
\right) }\left( A\right) \right\Vert _{\mathcal{U}}=0\ ,\qquad A\in \mathcal{%
U}\ .  \label{cigare state2}
\end{equation}%
Because $\mathbf{\Psi }\in C\left( \mathbb{R};\mathfrak{W}^{\mathbb{R}%
}\right) $, observe from (\ref{iteration0bis}) that, for $T\in \mathbb{R}%
^{+} $, 
\begin{equation}
\sup_{\alpha \in \left[ -T,T\right] }\sup_{\rho \in E}\left\Vert \mathbf{%
\Psi }\left( \alpha ;\rho \right) \right\Vert _{\mathcal{W}}=\sup_{\alpha
\in \left[ -T,T\right] }\left\Vert \mathbf{\Psi }\left( \alpha \right)
\right\Vert _{\mathfrak{W}}<\infty \ .  \label{cigare state3}
\end{equation}%
Therefore, we infer from Proposition \ref{Theorem Lieb-Robinson copy(3)}
(iii) and Lebesgue's dominated convergence theorem that%
\begin{equation*}
\lim_{n\rightarrow \infty }\left\Vert \tau _{t_{n},s_{n}}^{\mathbf{\Psi }%
\left( \rho _{n}\right) }\left( A\right) -\tau _{t_{n},s_{n}}^{\mathbf{\Psi }%
\left( \rho \right) }\left( A\right) \right\Vert _{\mathcal{U}}=0\ ,\qquad
A\in \mathcal{U}_{0}\ .
\end{equation*}%
By Combining (\ref{cigare state1})-(\ref{cigare state2}) with (\ref%
{definiotion tho frac}), \cite[Lemma 5.1 (iii)]{Bru-pedra-MF-I}, the density
of $\mathcal{U}_{0}$ in $\mathcal{U}$ and the fact that $\tau _{t,s}^{%
\mathbf{\Psi }\left( \rho \right) }$ is a contraction for any $s,t\in 
\mathbb{R}$, we deduce that $(\mathfrak{T}_{t,s}^{\mathbf{\Psi }})_{_{s,t\in 
\mathbb{R}}}$ is a strongly continuous two-para%
\-%
meter family.

We next prove the non-autonomous evolution equation for $(\mathfrak{T}%
_{t,s}^{\mathbf{\Psi }})_{_{s,t\in \mathbb{R}}}$: By Equation (\ref%
{definiotion tho frac}), 
\begin{equation}
\forall s,t\in {\mathbb{R}},\ f\in \mathfrak{C}\subseteq \mathfrak{U}:\qquad
\partial _{t}\mathfrak{T}_{t,s}^{\mathbf{\Psi }}\left( f\right) =0\ ,\qquad 
\mathfrak{T}_{s,s}^{\mathbf{\Psi }}=\mathbf{1}_{\mathfrak{U}}\ .
\label{fdqsdqsd}
\end{equation}%
For any $s,t\in {\mathbb{R}}$, $A\in \mathcal{U}_{0}\subseteq \mathfrak{U}%
_{0}$ and $h\in {\mathbb{R}}\backslash \{0\}$, observe additionally that 
\begin{eqnarray}
&&\left\Vert h^{-1}\left( \mathfrak{T}_{t+h,s}^{\mathbf{\Psi }}\left(
A\right) -\mathfrak{T}_{t,s}^{\mathbf{\Psi }}\left( A\right) \right) -%
\mathfrak{T}_{t,s}^{\mathbf{\Psi }}\circ \mathbf{\delta }^{\mathbf{\Psi }%
(t)}\left( A\right) \right\Vert _{\mathfrak{U}}  \label{cigare meta1} \\
&=&\sup_{\rho \in E}\left\Vert h^{-1}\left( \tau _{t+h,t}^{\mathbf{\Psi }%
\left( \rho \right) }\left( A\right) -A\right) -\delta ^{\mathbf{\Psi }%
\left( t;\rho \right) }\left( A\right) \right\Vert _{\mathcal{U}}\ ,  \notag
\end{eqnarray}%
using Proposition \ref{Theorem Lieb-Robinson}. Now, by contradiction, assume
the existence of a zero sequence $(h_{n})_{n\in \mathbb{N}}$, a sequence $%
(\rho _{n})_{n\in \mathbb{N}}\subseteq E$ and a positive constant $D>0$ such
that 
\begin{equation*}
\inf_{n\in \mathbb{N}}\left\Vert h_{n}^{-1}\left( \tau _{t+h_{n},t}^{\mathbf{%
\Psi }\left( \rho _{n}\right) }\left( A\right) -A\right) -\delta ^{\mathbf{%
\Psi }\left( t;\rho _{n}\right) }\left( A\right) \right\Vert _{\mathcal{U}%
}\geq D>0.
\end{equation*}%
By weak$^{\ast }$-compactness of $E$, we can assume without loss of
generality that $(\rho _{n})_{n\in \mathbb{N}}$ converges in the weak$^{\ast
}$ topology to some $\rho \in E$, as $n\rightarrow \infty $. From Corollary %
\ref{Lemma cigare1}, it follows that 
\begin{equation}
\liminf_{n\rightarrow \infty }\left\Vert h_{n}^{-1}\left( \tau _{t+h_{n},t}^{%
\mathbf{\Psi }\left( \rho _{n}\right) }\left( A\right) -A\right) -\delta ^{%
\mathbf{\Psi }\left( t;\rho \right) }\left( A\right) \right\Vert _{\mathcal{U%
}}\geq D>0\ .  \label{toto chiant1}
\end{equation}%
Meanwhile, using Proposition \ref{Theorem Lieb-Robinson copy(3)} (iii), for
any $\varepsilon \in \mathbb{R}^{+}$, there is $n_{0}\in \mathbb{N}$ such
that, for all $n\geq n_{0}$,%
\begin{eqnarray}
h_{n}^{-1}\left( \tau _{t+h_{n},t}^{\mathbf{\Psi }\left( \rho _{n}\right)
}-\tau _{t+h_{n},t}^{\mathbf{\Psi }\left( \rho \right) }\right) \left(
A\right) &\leq &2\left\vert \Lambda \right\vert \left\Vert A\right\Vert _{%
\mathcal{U}}\left\Vert \mathbf{F}\right\Vert _{1,\mathfrak{L}}\mathrm{e}^{2%
\mathbf{D}\int_{-\varepsilon }^{\varepsilon }\left\Vert \Psi \left( \rho
;t+\alpha _{1}\right) \right\Vert _{\mathcal{W}}\mathrm{d}\alpha _{1}}
\label{toto chiant2} \\
&&\times \max_{\alpha \in \left[ -\varepsilon ,\varepsilon \right]
}\left\Vert \mathbf{\Psi }\left( t+\alpha ;\rho _{n}\right) -\mathbf{\Psi }%
\left( t+\alpha ;\rho \right) \right\Vert _{\mathcal{W}}\ .  \notag
\end{eqnarray}%
Note that the mapping $(t,\rho )\mapsto \mathbf{\Psi }\left( t;\rho \right) $
is (jointly) continuous on $\mathbb{R}\times E$, by definition of $\mathfrak{%
W}$, and we infer from Inequalities (\ref{toto chiant1})-(\ref{toto chiant2}%
) that 
\begin{equation*}
\liminf_{n\rightarrow \infty }\left\Vert h_{n}^{-1}\left( \tau _{t+h_{n},t}^{%
\mathbf{\Psi }\left( \rho \right) }\left( A\right) -A\right) -\delta ^{%
\mathbf{\Psi }\left( t;\rho \right) }\left( A\right) \right\Vert _{\mathcal{U%
}}\geq D>0\ ,
\end{equation*}%
which contradicts (\ref{cauchy trivial1}) for $\Psi =\mathbf{\Psi }\left(
\rho \right) $. By Equation (\ref{cigare meta1}), it follows that 
\begin{equation}
\forall s,t\in {\mathbb{R}},\ A\in \mathcal{U}_{0}\subseteq \mathfrak{U}%
_{0}:\qquad \partial _{t}\mathfrak{T}_{t,s}^{\mathbf{\Psi }}\left( A\right) =%
\mathfrak{T}_{t,s}^{\mathbf{\Psi }}\circ \mathbf{\delta }^{\mathbf{\Psi }%
(t)}\left( A\right) \ .  \label{fdqsdqsdfdqsdqsd}
\end{equation}%
By using that $\mathbf{\delta }^{\mathbf{\Psi }(t)}$, $t\in {\mathbb{R}}$,
are derivations and $\mathfrak{T}_{t,s}^{\mathbf{\Psi }}$, $s,t\in \mathbb{R}
$, are $\ast $-auto%
\-%
morphisms of $\mathfrak{U}$, we deduce (\ref{klsdf}) on $\mathfrak{U}_{0}$,
from (\ref{fdqsdqsd}) and (\ref{fdqsdqsdfdqsdqsd}). Recall that $\mathfrak{U}%
_{0}=\mathrm{span}\left\{ \mathfrak{C}\mathcal{U}_{0}\right\} $, by (\ref%
{spain}).

Finally, in order to prove the uniqueness of the solution to (\ref{klsdf}),
assume that $(\mathfrak{\tilde{T}}_{t,s})_{_{s,t\in \mathbb{R}}}\subseteq 
\mathcal{B}(\mathfrak{U})$ is a two-parameter family satisfying (\ref{klsdf}%
) on $\mathfrak{U}_{0}$. Since $\mathbf{\delta }^{\mathbf{\Psi }(t)}\left( 
\mathfrak{C}\right) =\{0\}$, $\mathfrak{C}$ is a subspace of the fixed point
algebra of $(\mathfrak{\tilde{T}}_{t,s})_{_{s,t\in \mathbb{R}}}$. In
particular, by \cite[Lemma 5.2]{Bru-pedra-MF-I}, it comes from a strongly
continuous family $(\tilde{\tau}_{t,s}^{\rho })_{\left( \rho ,s,t\right) \in
E\times \mathbb{R}^{2}}$ defined by 
\begin{equation*}
\tilde{\tau}_{t,s}^{\rho }\left( A\right) \doteq \lbrack \mathfrak{\tilde{T}}%
_{t,s}\left( A\right) ]\left( \rho \right) \ ,\qquad \rho \in E,\ A\in 
\mathcal{U}\subseteq \mathfrak{U},\ s,t\in \mathbb{R}\ .
\end{equation*}%
Through (\ref{klsdf}) (cf. (\ref{cigare meta1})), for each $\rho \in E$, $%
(\tau _{t,s}^{\mathbf{\Psi }\left( \rho \right) })_{_{s,t\in \mathbb{R}}}$
and $(\tilde{\tau}_{t,s}^{\rho })_{_{s,t\in \mathbb{R}}}$ are both solution
in $\mathcal{B}(\mathcal{U})$ to the non-autonomous evolution equation (\ref%
{cauchy trivial1}) on $\mathcal{U}_{0}$ for $\Psi =\mathbf{\Psi }\left( \rho
\right) $. Therefore, by Proposition \ref{Theorem Lieb-Robinson}, the
solution to the non-autonomous evolution equation (\ref{klsdf}) on $%
\mathfrak{U}_{0}$ is also unique.
\end{proof}

\subsection{From Quantum Dynamics to Classical Flows}

Proposition \ref{Lemma cigare1 copy(5)} means that, for any $\mathbf{\Psi }%
\in C(\mathbb{R};\mathfrak{W}^{\mathbb{R}})$, $(\mathfrak{U},\mathfrak{T}^{%
\mathbf{\Psi }})$ is a state-dependent $C^{\ast }$-dynamical system, as
defined in \cite[Definition 5.3]{Bru-pedra-MF-I}. Therefore, as explained in 
\cite[Section 5.2]{Bru-pedra-MF-I}, for any $\mathbf{\Psi }\in C(\mathbb{R};%
\mathfrak{W}^{\mathbb{R}})$, $(\mathfrak{U},\mathfrak{T}^{\mathbf{\Psi }})$
induces a Feller dynamics within the classical $C^{\ast }$-algebra $%
\mathfrak{C}$ defined by (\ref{metaciagre set 2})-(\ref{metaciagre set 2bis}%
): \medskip

\noindent \underline{State-space trajectories:} Let $C\left( E;E\right) $ be
the set of weak$^{\ast }$-continuous functions from the state space $E$ to
itself endowed with the topology of uniform convergence. In other words, any
net $(f_{j})_{j\in J}\subseteq C\left( E;E\right) $ converges to $f\in
C\left( E;E\right) $ whenever%
\begin{equation}
\lim_{j\in J}\max_{\rho \in E}\left\vert f_{j}(\rho )(A)-f(\rho
)(A)\right\vert =0\ ,\qquad \text{for all }A\in \mathcal{U}\ .
\label{uniform convergence weak*}
\end{equation}%
We denote by $\mathrm{Aut}\left( E\right) \varsubsetneq C\left( E;E\right) $
the subspace of all automorphisms of $E$, i.e., element of $C\left(
E;E\right) $ with weak$^{\ast }$-continuous inverse. Equivalently, $\mathrm{%
Aut}\left( E\right) $ is the set of all bijective mappings in $C\left(
E;E\right) $, because $E$ is a compact Hausdorff space. From the family $%
\mathfrak{T}^{\mathbf{\Psi }}\equiv (\mathfrak{T}_{t,s}^{\mathbf{\Psi }%
})_{s,t\in \mathbb{R}}$, we define a continuous family $(\phi _{t,s}^{%
\mathbf{\Psi }})_{s,t\in \mathbb{R}}\subseteq \mathrm{Aut}\left( E\right) $
by 
\begin{equation}
\phi _{t,s}^{\mathbf{\Psi }}\left( \rho \right) \doteq \rho \circ \tau
_{t,s}^{\mathbf{\Psi }\left( \rho \right) }\ ,\qquad \rho \in E,\ s,t\in 
\mathbb{R}\ ,  \label{phitsbis}
\end{equation}%
where $(\tau _{t,s}^{\mathbf{\Psi }\left( \rho \right) })_{\left( \rho
,s,t\right) \in E\times \mathbb{R}^{2}}$ is the unique strongly continuous
family of $\ast $-automorphisms of $\mathcal{U}$ satisfying (\ref%
{definiotion tho frac}), see also \cite[Lemma 5.2]{Bru-pedra-MF-I}.\medskip

\noindent \underline{Classical flows as Feller evolution systems:} The
state-space trajectories, in turn, yield a strongly continuous two-parameter
family $(V_{t,s}^{\mathbf{\Psi }})_{s,t\in \mathbb{R}}$ of $\ast $-auto%
\-%
morphisms of the classical $C^{\ast }$-algebra $\mathfrak{C}$, defined by 
\begin{equation}
V_{t,s}^{\mathbf{\Psi }}f\doteq f\circ \phi _{t,s}^{\mathbf{\Psi }}\ ,\qquad
f\in \mathfrak{C},\ s,t\in \mathbb{R}\ .  \label{shorodinger dynamicsbis}
\end{equation}%
This classical dynamics is a \emph{Feller evolution system} in the following
sense: As a $\ast $-automorphism, $V_{t,s}^{\mathbf{\Psi }}$ is
self-adjointness- and positivity-preserving while $\Vert V_{t,s}^{\mathbf{%
\Psi }}\Vert _{\mathcal{B}\left( \mathfrak{C}^{\mathbb{R}}\right) }=1$; $%
(V_{t,s}^{\mathbf{\Psi }})_{s,t\in \mathbb{R}}$ is a strongly continuous
two-parameter family satisfying 
\begin{equation}
\forall s,r,t\in \mathbb{R}:\qquad V_{t,s}^{\mathbf{\Psi }}=V_{r,s}^{\mathbf{%
\Psi }}\circ V_{t,r}^{\mathbf{\Psi }}\ ,  \label{reverse}
\end{equation}%
by (\ref{reverse cocycle0}). Therefore, as explained in \cite[Section 4.4]%
{Bru-pedra-MF-I}, the classical dynamics defined as the restriction of $%
(V_{t,s}^{\mathbf{\Psi }})_{s,t\in \mathbb{R}}$ to the real space $\mathfrak{%
C}^{\mathbb{R}}$ can be associated in this case with Feller processes%
\footnote{%
The positivity and norm-preserving property are reminiscent of Markov
semigroups.} in probability theory: By the Riesz-Markov representation
theorem and the monotone convergence theorem, there is a unique
two-parameter group $(p_{t,s}^{\mathbf{\Psi }})_{s,t\in \mathbb{R}}$\ of
Markov transition kernels $p_{t,s}^{\mathbf{\Psi }}(\cdot ,\cdot )$ on $E$
such that%
\begin{equation*}
V_{t,s}^{\mathbf{\Psi }}f\left( \rho \right) =\int_{E}f\left( \hat{\rho}%
\right) p_{t,s}^{\mathbf{\Psi }}(\rho ,\mathrm{d}\hat{\rho})\ ,\qquad f\in 
\mathfrak{C}^{\mathbb{R}}\ .
\end{equation*}%
The right-hand side of the above identity makes sense for bounded measurable
functions from $E$ to $\mathbb{R}$. In fact, one can naturally extend $%
(V_{t,s}^{\mathbf{\Psi }})_{s,t\in \mathbb{R}}$ to this more general class
of functions on $E$.

The notion of Feller evolution system, which is only an extension of Feller
semigroups to non-auto%
\-%
nomous two-parameter families, has been introduced (at least) in 2014 \cite%
{Feller}. \medskip

\noindent \underline{Parity:} For any $\mathbf{\Psi }\in C(\mathbb{R};%
\mathfrak{W}^{\mathbb{R}})$, the family $(\mathfrak{T}_{t,s}^{\mathbf{\Psi }%
})_{_{s,t\in \mathbb{R}}}$ is parity-preserving, i.e.,%
\begin{equation*}
\Xi \circ \mathfrak{T}_{t,s}^{\mathbf{\Psi }}=\mathfrak{T}_{t,s}^{\mathbf{%
\Psi }}\circ \Xi \ ,\qquad s,t\in \mathbb{R}\ ,
\end{equation*}%
where $\Xi $ is the $\ast $-automorphism of $\mathfrak{U}$ defined by (\ref%
{gauge invariant}). It follows that 
\begin{equation}
\phi _{t,s}^{\mathbf{\Psi }}\left( E^{+}\right) \subseteq E^{+}\ ,\qquad
\phi _{t,s}^{\mathbf{\Psi }}\left( E\backslash E^{+}\right) \subseteq
E\backslash E^{+}\ ,\qquad s,t\in \mathbb{R}\ ,  \label{symmetry group eq1}
\end{equation}%
which in turn implies that $V_{t,s}^{\mathbf{\Psi }}$ can be seen as a
mapping on either $C(E^{+},\mathbb{C})$ or $C(E\backslash E^{+},\mathbb{C})$:%
\begin{equation}
V_{t,s}^{\mathbf{\Psi }}\left( f|_{E^{+}}\right) \doteq \left( V_{t,s}^{%
\mathbf{\Psi }}f\right) |_{E^{+}}\ ,\qquad V_{t,s}^{\mathbf{\Psi }}\left(
f|_{E\backslash E^{+}}\right) \doteq \left( V_{t,s}^{\mathbf{\Psi }}f\right)
|_{E\backslash E^{+}}\ ,\qquad f\in \mathfrak{C},\ s,t\in \mathbb{R}\ .
\label{symmetry group eq2}
\end{equation}%
Recall that $E^{+}$ is the weak$^{\ast }$-compact convex set of even states
defined by (\ref{gauge invariant states}), which is the physical state
space. Similar to \cite[Corollary 4.3]{Bru-pedra-MF-I}, the set $\mathcal{E}%
(E^{+})$ of extreme points of $E^{+}$\ is also conserved by the flow.
\medskip

\noindent \underline{Translations:} Let $\mathfrak{L}=\mathbb{Z}^{d}$. For
any continuous mapping $\mathbf{\Psi }\in C(\mathbb{R};\mathfrak{W}_{1}\cap 
\mathfrak{W}^{\mathbb{R}})$ from $\mathbb{R}$ to the space $\mathfrak{W}%
_{1}\cap \mathfrak{W}^{\mathbb{R}}$ of translation-invariant, self-adjoint
and state-dependent interactions, the mapping $x\mapsto \mathrm{A}_{x}$ from 
$\mathbb{Z}^{d}$ to the group of $\ast $-automorphisms of $\mathfrak{U}$,
defined by (\ref{translatbis}), is a symmetry group of the state-dependent $%
C^{\ast }$-dynamical system $(\mathfrak{U},\mathfrak{T}^{\mathbf{\Psi }})$,
i.e., 
\begin{equation*}
\mathrm{A}_{x}\circ \mathfrak{T}_{t,s}^{\mathbf{\Psi }}=\mathfrak{T}_{t,s}^{%
\mathbf{\Psi }}\circ \mathrm{A}_{x}\ ,\qquad s,t\in \mathbb{R},\ x\in 
\mathbb{Z}^{d}\ .
\end{equation*}%
As a consequence, for any $\vec{\ell}\in \mathbb{N}^{d}$, 
\begin{equation}
\phi _{t,s}^{\mathbf{\Psi }}\left( E_{\vec{\ell}}\right) \subseteq E_{\vec{%
\ell}}\ ,\qquad \phi _{t,s}^{\mathbf{\Psi }}\left( E\backslash E_{\vec{\ell}%
}\right) \subseteq E\backslash E_{\vec{\ell}}\ ,\qquad s,t\in \mathbb{R}\ ,
\label{G-invariance states3}
\end{equation}%
which in turn implies that $V_{t,s}^{\mathbf{\Psi }}$ can be seen as a
mapping on either $C(E_{\vec{\ell}},\mathbb{C})$ or $C(E\backslash E_{\vec{%
\ell}},\mathbb{C})$:%
\begin{equation}
V_{t,s}^{\mathbf{\Psi }}(f|_{E_{\vec{\ell}}})\doteq (V_{t,s}^{\mathbf{\Psi }%
}f)|_{E_{\vec{\ell}}}\ ,\qquad V_{t,s}^{\mathbf{\Psi }}(f|_{E\backslash E_{%
\vec{\ell}}})\doteq (V_{t,s}^{\mathbf{\Psi }}f)|_{E\backslash E_{\vec{\ell}%
}}\ ,\qquad f\in \mathfrak{C},\ s,t\in \mathbb{R}\ .
\label{G-invariance states4}
\end{equation}%
Recall that $E_{\vec{\ell}}\subseteq E^{+}$ is the weak$^{\ast }$-compact
convex set of $\vec{\ell}$-periodic states defined by (\ref{periodic
invariant states}) for any $\vec{\ell}\in \mathbb{N}^{d}$. Similar to \cite[%
Corollary 4.3]{Bru-pedra-MF-I}, for every $\vec{\ell}\in \mathbb{N}^{d}$,
the set $\mathcal{E}(E_{\vec{\ell}})$ of extreme points of $E_{\vec{\ell}}$\
is, in this case, also conserved by the flow.

\subsection{Self-Consistency Equations}

By using Equation (\ref{equation commutators}), Proposition \ref{density of
periodic states copy(2)} and Corollary \ref{Lemma cigare1} together with the
linearity of the mapping $\Phi \mapsto \delta ^{\Phi }$, for every integer $%
n\geq 2$, any translation-invariant interactions $\Psi ^{(1)},\ldots ,\Psi
^{(n)}\in \mathcal{W}_{1}$, each local element $A\in \mathcal{U}_{0}$, $\vec{%
\ell}\in \mathbb{N}^{d}$ and every extreme $\vec{\ell}$-periodic state $\rho
\in \mathcal{E}(E_{\vec{\ell}})\subseteq E_{\vec{\ell}}$, one can prove that 
\begin{equation}
\lim_{L\rightarrow \infty }\frac{1}{\left\vert \Lambda _{L}\right\vert ^{n-1}%
}\rho \left( i\left[ U_{L}^{\Psi ^{(1)}}\cdots U_{L}^{\Psi ^{(n)}},A\right]
\right) =\rho \circ \delta ^{\mathbf{\Psi }\left( \rho \right) }\left(
A\right) \ ,  \label{equation commutatorsbis}
\end{equation}%
where, in this case, the state-dependent interaction $\mathbf{\Psi }\in 
\mathfrak{W}$ equals%
\begin{equation}
\mathbf{\Psi }\left( \rho \right) =\left\lfloor \rho ;\Psi ^{(1)},\ldots
,\Psi ^{(n)}\right\rfloor _{\vec{\ell}}\doteq \sum_{m=1}^{n}\Psi
^{(m)}\prod\limits_{j\in \left\{ 1,\ldots ,n\right\} ,j\neq m}\rho (%
\mathfrak{e}_{\Psi ^{(j)},\vec{\ell}})\in \mathcal{W}_{1}\ ,\qquad \rho \in
E\ .  \label{def aussi utile}
\end{equation}%
Compare (\ref{equation commutatorsbis}) with Definitions \ref{definition
long range energy} and \ref{derivation long range}.

The proof of (\ref{equation commutatorsbis}) uses the \emph{ergodicity} of
extreme periodic states in a crucial way. It is non-trivial and will be
performed in detail in \cite{Bru-pedra-MF-III}. One central argument in this
proof, like in Haag's approach \cite{haag62} to mean-field theories, is
that, for any translation-invariant interaction $\Phi \in \mathcal{W}_{1}$
and any cyclic representation $\left( \mathcal{H}_{\rho },\pi _{\rho
},\Omega _{\rho }\right) $ of an \emph{extreme} state $\rho \in \mathcal{E}%
(E_{\vec{\ell}})$ at fixed $\vec{\ell}\in \mathbb{N}^{d}$, the uniformly
bounded family 
\begin{equation*}
\left\{ \frac{\pi _{\rho }\left( U_{L}^{\Phi }\right) }{\left\vert \Lambda
_{L}\right\vert }\right\} _{L\in \mathbb{N}}\subseteq \mathcal{B}\left( 
\mathcal{H}_{\rho }\right)
\end{equation*}%
converges to the operator $\rho (\mathfrak{e}_{\Phi ,\vec{\ell}})\mathbf{1}_{%
\mathcal{H}_{\rho }}$, in the sense of the strong operator topology. See 
\cite[Section 5.3]{Bru-pedra-MF-III} for more details. Compare with
Proposition \ref{density of periodic states copy(2)}.

Having in mind the discussions of Section \ref{section pb} and the linearity
of the mapping $\Phi \mapsto \delta ^{\Phi }$, we give here the limit (\ref%
{equation commutatorsbis}) in order to convince the reader that appropriate
approximating (state-dependent, short-range) interactions naturally appears
in the description of the infinite-volume dynamics of lattice-fermion (or
quantum-spin) systems with long-range interactions. To define them, recall
that the integral of interactions is defined by (\ref{definition integral
interaction0})-(\ref{definition integral interaction}).

\begin{definition}[Non-autonomous approximating interactions]
\label{definition BCS-type model approximated}\mbox{ }\newline
For $\vec{\ell}\in \mathbb{N}^{d}$ and any continuous functions $\mathfrak{m}%
=(\Phi (t),\mathfrak{\mathfrak{a}}(t))_{t\in \mathbb{R}}\in C\left( \mathbb{R%
};\mathcal{M}\right) $, $\xi \in C\left( \mathbb{R};E\right) $, we define
the mapping $\Phi ^{(\mathfrak{m},\xi )}$ from $\mathbb{R}$ to $\mathcal{W}^{%
\mathbb{R}}$ by%
\begin{equation*}
\Phi ^{(\mathfrak{m},\xi )}\left( t\right) \doteq \Phi \left( t\right)
+\sum_{n\in \mathbb{N}}\int_{\mathbb{S}^{n}}\ \left\lfloor \xi \left(
t\right) ;\Psi ^{(1)},\ldots ,\Psi ^{(n)}\right\rfloor _{\vec{\ell}}\ 
\mathfrak{a}\left( t\right) _{n}\left( \mathrm{d}\Psi ^{(1)},\ldots ,\mathrm{%
d}\Psi ^{(n)}\right) \,,\qquad t\in \mathbb{R}\ ,
\end{equation*}%
with $\left\lfloor \rho ;\Psi \right\rfloor _{\vec{\ell}}\doteq \Psi $. If $%
\mathbf{\xi }\in C\left( \mathbb{R};\mathrm{Aut}\left( E\right) \right) $,
then a mapping $\mathbf{\Phi }^{(\mathfrak{m},\mathbf{\xi })}$ from $\mathbb{%
R}$ to $\mathfrak{W}^{\mathbb{R}}$ is defined, for any $\rho \in E$ and $%
t\in \mathbb{R}$, by%
\begin{equation*}
\mathbf{\Phi }^{(\mathfrak{m},\mathbf{\xi })}\left( t;\rho \right) \doteq
\Phi \left( t\right) +\sum_{n\in \mathbb{N}}\int_{\mathbb{S}^{n}}\
\left\lfloor \mathbf{\xi }\left( t;\rho \right) ;\Psi ^{(1)},\ldots ,\Psi
^{(n)}\right\rfloor _{\vec{\ell}}\ \mathfrak{a}\left( t\right) _{n}\left( 
\mathrm{d}\Psi ^{(1)},\ldots ,\mathrm{d}\Psi ^{(n)}\right) \ .
\end{equation*}
\end{definition}

\noindent Such approximating interactions can be used to define, via
Proposition \ref{Theorem Lieb-Robinson} or Proposition \ref{Lemma cigare1
copy(5)}, $\ast $-automorphisms of $\mathcal{U}$ or $\mathfrak{U}$, because
they are always bounded continuous functions:

\begin{lemma}[Continuity of approximating interactions]
\label{definition BCS-type model approximated copy(2)}\mbox{ }\newline
Let $\vec{\ell}\in \mathbb{N}^{d}$ and $\mathfrak{m}\in C\left( \mathbb{R};%
\mathcal{M}\right) $. For any $\xi \in C\left( \mathbb{R};E\right) $, $\Phi
^{(\mathfrak{m},\xi )}\in C(\mathbb{R};\mathcal{W}^{\mathbb{R}})$ and, for
any $\mathbf{\xi }\in C\left( \mathbb{R};\mathrm{Aut}\left( E\right) \right) 
$, $\mathbf{\Phi }^{(\mathfrak{m},\mathbf{\xi })}\in C(\mathbb{R};\mathfrak{W%
}^{\mathbb{R}})$ with%
\begin{equation}
\left\Vert \Phi ^{(\mathfrak{m},\xi )}\left( t\right) \right\Vert _{\mathcal{%
W}}\leq \left\Vert \mathfrak{m}\left( t\right) \right\Vert _{\mathcal{M}}%
\text{\qquad and\qquad }\left\Vert \mathbf{\Phi }^{(\mathfrak{m},\mathbf{\xi 
})}\left( t\right) \right\Vert _{\mathfrak{W}}\leq \left\Vert \mathfrak{m}%
\left( t\right) \right\Vert _{\mathcal{M}}\,,\qquad t\in \mathbb{R}\ .
\label{inequality trivial}
\end{equation}
\end{lemma}

\begin{proof}
Inequalities (\ref{inequality trivial}) are direct consequences of Equations
(\ref{e phi}), (\ref{definition 0})-(\ref{definition 0bis}), Definition \ref%
{def long range} and the triangle inequality, recalling that $\mathbb{S}^{n}$
is the $n$-fold Cartesian product of the unit sphere $\mathbb{S}$ of the
Banach space $\mathcal{W}_{1}$. Using the same arguments, note also that,
for any $\mathfrak{m}\in C\left( \mathbb{R};\mathcal{M}\right) $, $\xi \in
C\left( \mathbb{R};E\right) $ and $t_{1},t_{2}\in \mathbb{R}$, 
\begin{align}
& \left\Vert \Phi ^{(\mathfrak{m},\xi )}\left( t_{1}\right) -\Phi ^{(%
\mathfrak{m},\xi )}\left( t_{2}\right) \right\Vert _{\mathcal{W}}
\label{totototot} \\
& \leq \left\Vert \Phi \left( t_{1}\right) -\Phi \left( t_{2}\right)
\right\Vert _{\mathcal{W}}+\sum_{n\in \mathbb{N}}n\left\Vert \mathbf{F}%
\right\Vert _{1,\mathfrak{L}}^{n-1}\left\Vert \mathfrak{a}\left(
t_{1}\right) _{n}-\mathfrak{a}\left( t_{2}\right) _{n}\right\Vert _{\mathcal{%
S}(\mathbb{S}^{n}\mathbb{)}}  \notag \\
& \qquad +\sum_{n\in \mathbb{N}}\left\Vert \mathbf{F}\right\Vert _{1,%
\mathfrak{L}}^{n-1}\int_{\mathbb{S}^{n}}\sum_{j=1}^{n}\left\vert \left( \xi
\left( t_{1}\right) -\xi \left( t_{2}\right) \right) (\mathfrak{e}_{\Psi
^{(j)},\vec{\ell}})\right\vert \ \left\vert \mathfrak{a}\left( t_{1}\right)
_{n}\right\vert \left( \mathrm{d}\Psi ^{(1)},\ldots ,\mathrm{d}\Psi
^{(n)}\right) \ .  \notag
\end{align}%
Since $\mathfrak{m}\in C\left( \mathbb{R};\mathcal{M}\right) $ and $\xi \in
C\left( \mathbb{R};E\right) $, we invoke Lebesgue's dominated convergence
theorem to deduce from the last inequality that $\Phi ^{(\mathfrak{m},\xi
)}\in C(\mathbb{R};\mathcal{W}^{\mathbb{R}})$. The difference of $\mathbf{%
\Phi }^{(\mathfrak{m},\mathbf{\xi })}$ for two different times at a fixed
state $\rho $ satisfies the same inequality as (\ref{totototot}), $\xi
\left( t_{1}\right) ,\xi \left( t_{2}\right) $ being replaced with $\mathbf{%
\xi }(t_{1};\rho ),\mathbf{\xi }(t_{2};\rho )$, respectively. Since $\mathrm{%
Aut}\left( E\right) \varsubsetneq C\left( E;E\right) $ is the subspace of
all automorphisms of $E$ endowed with the topology of uniform convergence,
as stated in Equation (\ref{uniform convergence weak*}), one also infers
from Lebesgue's dominated convergence theorem that $\mathbf{\Phi }^{(%
\mathfrak{m},\mathbf{\xi })}\in C(\mathbb{R};\mathfrak{W}^{\mathbb{R}})$,
provided $\mathfrak{m}\in C\left( \mathbb{R};\mathcal{M}\right) $ and $%
\mathbf{\xi }\in C\left( \mathbb{R};\mathrm{Aut}\left( E\right) \right) $.
\end{proof}

As is usual, we do the following identifications for the subspaces of
constant functions:%
\begin{equation}
E\subseteq C\left( \mathbb{R};E\right) \text{\qquad and\qquad }\mathrm{Aut}%
\left( E\right) \subseteq C\left( \mathbb{R};\mathrm{Aut}\left( E\right)
\right) \ .  \label{identify}
\end{equation}%
When $\xi \in E$ and $\mathfrak{m}=\left( \Phi ,(0,\mathfrak{a}_{2},0,\ldots
)\right) \in \mathcal{M}_{1}$, the first part of Definition \ref{definition
BCS-type model approximated} corresponds to the (autonomous) approximating
interactions first introduced in \cite[Definition 2.31]{BruPedra2}, see
also\ Section \ref{Long-range models}. They are used there to characterize
(generalized equilibrium) states minimizing the free-energy density through 
\emph{self-consistency equations} (gap equations), whose solutions are
related to non-cooperative equilibria of a two-person zero-sum game
(thermodynamic game).

More generally, by Proposition \ref{Lemma cigare1 copy(5)} and Lemma \ref%
{definition BCS-type model approximated copy(2)}, for any $\mathbf{\xi }\in
C\left( \mathbb{R};\mathrm{Aut}\left( E\right) \right) $, there is a
strongly continuous two-parameter family $(\mathfrak{T}_{t,s}^{\mathbf{\Phi }%
^{(\mathfrak{m},\mathbf{\xi })}})_{s,t\in \mathbb{R}}$ of $\ast $-auto%
\-%
morphisms of $\mathfrak{U}$ satisfying the reverse cocycle property. This
family satisfies a non-autonomous evolution equation, similar to Equation (%
\ref{cauchy trivial1}). It is used to construct the infinite-volume limit of
the non-autonomous dynamics of time-dependent long-range models of $\mathcal{%
M}$ within a cyclic representation associated with an arbitrary periodic
state. We show that, generically, a long-range (or mean-field) dynamics is
equivalent to an intricate combination of a classical \emph{and} short-range
quantum dynamics. Both dynamics will be (non-trivial) consequences of the
well-posedness of \emph{self-consistency equations}, which are reminiscent
of \cite[Theorem 4.1]{Bru-pedra-MF-I}.

To present these equations, recall that, for any $\Lambda \in \mathcal{P}%
_{f} $, $\mathcal{M}_{\Lambda }$ belongs to the dense subset $\mathcal{M}%
_{0}\subseteq \mathcal{M}$ of models with an arbitrary short-range part,
while the long-range interactions are polynomials of interactions that are
finite-range and translation-invariant. See (\ref{S00bis})-(\ref{S00bisbis})
and Definition \ref{definition long range energy}.

\begin{theorem}[Self-consistency equations]
\label{theorem sdfkjsdklfjsdklfj}\mbox{ }\newline
Fix $\Lambda \in \mathcal{P}_{f}$ and $\mathfrak{m}\in C_{b}(\mathbb{R};%
\mathcal{M}_{\Lambda })$. There is a unique $\mathbf{\varpi }^{\mathfrak{m}%
}\in C\left( \mathbb{R}^{2};\mathrm{Aut}\left( E\right) \right) $ such that 
\begin{equation*}
\mathbf{\varpi }^{\mathfrak{m}}\left( s,t\right) =\phi _{t,s}^{\mathbf{\Phi }%
^{(\mathfrak{m},\mathbf{\varpi }^{\mathfrak{m}}\left( \alpha ,\cdot \right)
)}}|_{\alpha =s}\ ,\qquad s,t\in {\mathbb{R}}\ ,
\end{equation*}%
where $(\phi _{t,s}^{\mathbf{\Phi }^{(\mathfrak{m},\mathbf{\xi })}})_{s,t\in 
\mathbb{R}}$ is the continuous family of automorphisms of $E$ defined by (%
\ref{phitsbis}) with the state-dependent interaction $\mathbf{\Psi }=\mathbf{%
\Phi }^{(\mathfrak{m},\mathbf{\xi })}$ of Definition \ref{definition
BCS-type model approximated} for $\mathbf{\xi }\in C\left( \mathbb{R};%
\mathrm{Aut}\left( E\right) \right) $.
\end{theorem}

\begin{proof}
The theorem is a consequence of Lemmata \ref{Solution self} and \ref{lemma
well copy(1)}.
\end{proof}

\begin{remark}
\mbox{ }\newline
Section \ref{Well-posedness sect} proves stronger results than Theorem \ref%
{theorem sdfkjsdklfjsdklfj}. See, in particular, Lemma \ref%
{Differentiability}.
\end{remark}

\begin{remark}
\mbox{ }\newline
Self-consistency equations have been, very recently, also highlighted in 
\cite{4} by studying particles governed by the Vlasov equation and
interacting with an oscillatory environment. In this case, self-consistency
equations turn out to be essential to describe the behavior of this system
at large times, giving an innovative use of such equations in EDP.
\end{remark}

At fixed $\Lambda \in \mathcal{P}_{f}$ and long-range model $\mathfrak{m}\in
C_{b}(\mathbb{R};\mathcal{M}_{\Lambda })$, Theorem \ref{theorem
sdfkjsdklfjsdklfj} means that, for any $s,t\in {\mathbb{R}}$, $\rho \in E$
and $A\in \mathcal{U}$, 
\begin{equation}
\rho _{s,t}\left( A\right) =\rho \circ \tau _{t,s}^{\mathbf{\Phi }^{(%
\mathfrak{m},\mathbf{\varpi }^{\mathfrak{m}}\left( \alpha ,\cdot \right)
)}\left( \rho \right) }\left( A\right) |_{\alpha =s}\qquad \text{with}\qquad
\rho _{s,t}\doteq \mathbf{\varpi }^{\mathfrak{m}}\left( s,t;\rho \right)
\doteq \mathbf{\varpi }^{\mathfrak{m}}\left( s,t\right) \left( \rho \right)
\in E\ .  \label{equations nonlinear}
\end{equation}%
See Equations (\ref{notation state interactionbis}), (\ref{definiotion tho
frac}) and (\ref{phitsbis}). Let $\left( \mathcal{C}_{n}\right) _{n\in 
\mathbb{N}}$ be an arbitrary family of closed sets 
\begin{equation*}
\mathcal{C}_{n}\subseteq (\mathbb{S}\cap \mathcal{W}_{\Lambda })^{n}\
,\qquad n\in \mathbb{N}\ ,
\end{equation*}%
such that, for $t\in \mathbb{R}$, 
\begin{equation*}
|\mathfrak{\mathfrak{a}}(t)_{n}|(\mathcal{C}_{n})=|\mathfrak{\mathfrak{a}}%
(t)_{n}|(\mathbb{S}^{n})\quad \text{with}\quad \mathfrak{m}=(\Phi (t),%
\mathfrak{\mathfrak{a}}(t))_{t\in \mathbb{R}}\in C_{b}(\mathbb{R};\mathcal{M}%
_{\Lambda })
\end{equation*}%
(cf. (\ref{eq:enpersitebis}) and (\ref{S00bis})-(\ref{S0bis})). Then, (\ref%
{equations nonlinear}) for the time-dependent expectation 
\begin{equation*}
\rho _{s,t}\left( A\right) \in {\mathbb{C}}\ ,\qquad s,t\in {\mathbb{R}}\ ,
\end{equation*}%
of elements 
\begin{equation*}
A\in \mathcal{V}_{\mathfrak{m}}\doteq \left\{ \mathfrak{e}_{\Psi ,\vec{\ell}%
}:\Psi \in \mathcal{C}_{n},\ n\in \mathbb{N}\right\}
\end{equation*}%
leads to a systems of\ non-autonomous, coupled and non-linear equations, in
general. These self-consistency equations are strongly related to the
self-consistency equations (gap equations) explained in \cite[Section 2.8]%
{BruPedra2} for the special case of (generalized) equilibrium states. By
contrast, Equations (\ref{equations nonlinear}) for elements $A\in \mathcal{U%
}\backslash \mathcal{V}_{\mathfrak{m}}$ are not coupled to each other.

Last but not least, for any $\mathfrak{m}\in C_{b}(\mathbb{R};\mathcal{M}%
_{\Lambda })$, observe from (\ref{symmetry group eq1}) that 
\begin{equation*}
\mathbf{\varpi }^{\mathfrak{m}}\left( s,t;E^{+}\right) \subseteq E^{+}\
,\qquad \mathbf{\varpi }^{\mathfrak{m}}\left( s,t;E\backslash E^{+}\right)
\subseteq E\backslash E^{+}\ ,\qquad s,t\in \mathbb{R}\ .
\end{equation*}%
Recall that $E^{+}$ is the weak$^{\ast }$-compact convex set of even states
defined by (\ref{gauge invariant states}). When long-range models are
translation-invariant, that is, if $\mathfrak{m}\in C_{b}(\mathbb{R};%
\mathcal{M}_{1}\cap \mathcal{M}_{\Lambda })$ (see (\ref{translatino
invariatn long range models})), the non-autonomous approximating
interactions of Definition \ref{definition BCS-type model approximated} are
also translation-invariant. By (\ref{G-invariance states3}), it follows, in
this case, that $\mathbf{\varpi }^{\mathfrak{m}}\left( s,t\right) $ maps
periodic states to periodic states: for any $\vec{\ell}\in \mathbb{N}^{d}$, 
\begin{equation}
\mathbf{\varpi }^{\mathfrak{m}}\left( s,t;E_{\vec{\ell}}\right) \subseteq E_{%
\vec{\ell}}\ ,\qquad \mathbf{\varpi }^{\mathfrak{m}}\left( s,t;E\backslash
E_{\vec{\ell}}\right) \subseteq E\backslash E_{\vec{\ell}}\ ,\qquad s,t\in 
\mathbb{R}\ .  \label{invariance translation}
\end{equation}%
Recall that $E_{\vec{\ell}}\subseteq E$ is the weak$^{\ast }$-compact convex
subset of $\vec{\ell}$-periodic states defined by (\ref{periodic invariant
states}). Like $E$, it has a dense set of extreme points ($\vec{\ell}$%
-ergodic states). Additionally, in this case, similar to \cite[Corollary 4.3]%
{Bru-pedra-MF-I}, for every $\vec{\ell}\in \mathbb{N}^{d}$, the set $%
\mathcal{E}(E_{\vec{\ell}})$ of extreme points of $E_{\vec{\ell}}$\ is
conserved by the flow.

\subsection{Classical Part of Long-Range Dynamics}

Similar to (\ref{shorodinger dynamicsbis}), the continuous family $\mathbf{%
\varpi }^{\mathfrak{m}}$ of Theorem \ref{theorem sdfkjsdklfjsdklfj} yields a
family $(V_{t,s}^{\mathfrak{m}})_{s,t\in \mathbb{R}}$ of $\ast $-auto%
\-%
morphisms of $\mathfrak{C}$ defined by 
\begin{equation}
V_{t,s}^{\mathfrak{m}}\left( f\right) \doteq f\circ \mathbf{\varpi }^{%
\mathfrak{m}}\left( s,t\right) \ ,\qquad f\in \mathfrak{C},\ s,t\in \mathbb{R%
}\ .  \label{classical evolution family}
\end{equation}%
It is a strongly continuous two-parameter family defining a classical
dynamics on the commutative $C^{\ast }$-algebra $\mathfrak{C}$ of continuous
complex-valued functions on states defined by (\ref{metaciagre set 2})-(\ref%
{metaciagre set 2bis}):

\begin{proposition}[Classical dynamics as Feller evolution system]
\label{lemma poisson copy(2)}\mbox{
}\newline
Fix $\Lambda \in \mathcal{P}_{f}$ and $\mathfrak{m}\in C_{b}(\mathbb{R};%
\mathcal{M}_{\Lambda })$. Then, $(V_{t,s}^{\mathfrak{m}})_{s,t\in \mathbb{R}%
} $ is a strongly continuous two-parameter family of $\ast $-automorphisms
of $\mathfrak{C}$ satisfying the reverse cocycle property:%
\begin{equation*}
\forall s,r,t\in \mathbb{R}:\qquad V_{t,s}^{\mathfrak{m}}=V_{r,s}^{\mathfrak{%
m}}\circ V_{t,r}^{\mathfrak{m}}\ .
\end{equation*}%
If $\mathfrak{m}\in \mathcal{M}_{\Lambda }\varsubsetneq C_{b}(\mathbb{R};%
\mathcal{M}_{\Lambda })$, i.e., $\mathfrak{m}$ is constant in time, then $%
V_{t,s}^{\mathfrak{m}}=V_{t-s,0}^{\mathfrak{m}}$ for any $s,t\in \mathbb{R}$%
\ and $(V_{t,0}^{\mathfrak{m}})_{t\in \mathbb{R}}$ is a $C_{0}$-group of $%
\ast $-automorphisms of $\mathfrak{C}$.
\end{proposition}

\begin{proof}
In order to prove these assertions, one simply adapts the argument used to
prove \cite[Proposition 3.4]{Bru-pedra-MF-I}, having in mind results of
Section \ref{Well-posedness sect}. We omit the details.
\end{proof}

Like (\ref{shorodinger dynamicsbis}), $(V_{t,s}^{\mathfrak{m}})_{s,t\in 
\mathbb{R}}$ can be associated with a Feller process in probability theory.
If $\mathfrak{m}\in C_{b}(\mathbb{R};\mathcal{M}_{\Lambda }\cap \mathcal{M}%
_{1})$, note that the classical flow conserves the Poulsen simplex $E_{\vec{%
\ell}}$, $\vec{\ell}\in \mathbb{N}^{d}$, and $V_{t,s}^{\mathfrak{m}}$ can be
seen as either a mapping from $C(E_{\vec{\ell}};\mathbb{C})$ to itself or
from $C(E\backslash E_{\vec{\ell}};\mathbb{C})$ to itself: 
\begin{equation}
V_{t,s}^{\mathfrak{m}}(f|_{E_{\vec{\ell}}})\doteq (V_{t,s}^{\mathfrak{m}%
}f)|_{E_{\vec{\ell}}}\ ,\qquad V_{t,s}^{\mathfrak{m}}(f|_{E\backslash E_{%
\vec{\ell}}})\doteq (V_{t,s}^{\mathfrak{m}}f)|_{E\backslash E_{\vec{\ell}}}\
,\qquad f\in \mathfrak{C},\ s,t\in \mathbb{R}\ ,
\label{translation invaration encore}
\end{equation}%
using (\ref{invariance translation}). Compare with Equation (\ref%
{G-invariance states4}). For all $\mathfrak{m}\in C_{b}(\mathbb{R};\mathcal{M%
}_{\Lambda })$, the same holds true for the weak$^{\ast }$-compact convex
set $E^{+}$ of even states.

For any constant function $\mathfrak{m}\in \mathcal{M}_{\Lambda
}\varsubsetneq C_{b}(\mathbb{R};\mathcal{M}_{\Lambda })$, $(V_{t,0}^{%
\mathfrak{m}})_{t\in \mathbb{R}}$ is a $C_{0}$-group of $\ast $%
-automorphisms of $\mathfrak{C}$ and we denote by $\daleth ^{\mathfrak{m}}$
its (well-defined) generator. By \cite[Chap. II, Sect. 3.11]{EngelNagel}, it
is a closed (linear) operator densely defined in $\mathfrak{C}$. Since $%
V_{t,0}^{\mathfrak{m}}$, $t\in \mathbb{R}$, are $\ast $-automorphisms, we
infer from the Nelson theorem \cite[Theorem 1.5.4]{Bratelli-derivation}, or
the Lumer-Phillips theorem \cite[Theorem 3.1.16]{BrattelliRobinsonI}, that $%
\pm \daleth ^{\mathfrak{m}}$ are dissipative operators, i.e., $\daleth ^{%
\mathfrak{m}}$ is conservative. The $\ast $-homomorphism property of $%
V_{t,0}^{\mathfrak{m}}$, $t\in \mathbb{R}$, is reflected by the fact that $%
\daleth ^{\mathfrak{m}}$ has to be a symmetric derivation of $\mathfrak{C}$.
In fact, similar to \cite[Theorem 4.5]{Bru-pedra-MF-I}, $\daleth ^{\mathfrak{%
m}}$ is directly related to a Poissonian symmetric derivation.

In order to understand this fact, we need to figure out the appropriate 
\emph{classical energy functions}, which correspond in \cite[Theorem 4.5]%
{Bru-pedra-MF-I} to the function $h$. Recall that $\hat{A}\in \mathfrak{C}$
is the continuous and affine function defined by (\ref{fA}) for any $A\in 
\mathcal{U}$, while $\{\cdot ,\cdot \}$ is the Poisson bracket of Definition %
\ref{convex Frechet derivative copy(1)}. Now, having in mind Definition \ref%
{definition long range energy}, it is natural to define\ the following
family of classical energy functions of $\mathfrak{C}$:

\begin{definition}[Classical energy functions of long-range dynamics]
\label{definition BCS-type model approximated copy(1)}\mbox{ }\newline
For any $\mathfrak{m}\in \mathcal{M}$, we define the functions%
\begin{equation*}
\mathrm{h}_{L}^{\mathfrak{m}}\doteq \widehat{U_{L}^{\Phi }}+\sum_{n\in 
\mathbb{N}}\frac{1}{\left\vert \Lambda _{L}\right\vert ^{n-1}}\int_{\mathbb{S%
}^{n}}\ \widehat{U_{L}^{\Psi ^{(1)}}}\cdots \widehat{U_{L}^{\Psi ^{(n)}}}\ 
\mathfrak{a}\left( t\right) _{n}\left( \mathrm{d}\Psi ^{(1)},\ldots ,\mathrm{%
d}\Psi ^{(n)}\right) \in \mathfrak{C}^{\mathbb{R}}\,,\qquad L\in \mathbb{N}\
,
\end{equation*}%
which we name the local classical energy functions associated with $%
\mathfrak{m}$.
\end{definition}

\noindent The integral in Definition \ref{definition BCS-type model
approximated copy(1)} is well-defined by the same reasons than in Definition %
\ref{definition long range energy}. It is important to stress that, although
these two definitions look similar, $\mathrm{h}_{L}^{\mathfrak{m}}\neq 
\widehat{U_{L}^{\mathfrak{m}}}$, in general. Local classical energy
functions are continuously differentiable real-valued functions over the
convex and weak$^{\ast }$-compact set $E$, i.e., $\{\mathrm{h}_{L}^{%
\mathfrak{m}}\}_{L\in \mathbb{N}}\subseteq \mathfrak{Y}^{\mathbb{R}}$, the
subspace $\mathfrak{Y}^{\mathbb{R}}\subseteq \mathfrak{C}^{\mathbb{R}}$
being defined by (\ref{C1}): For any $\mathfrak{m}\in \mathcal{M}$, by
straightforward estimates using Equations (\ref{norm Uphi}), (\ref%
{definition 0})-(\ref{definition 0bis}) and Definition \ref{def long range}
together with Lebesgue's dominated convergence theorem, one checks that, for
any $L\in \mathbb{N}$, $\mathrm{h}_{L}^{\mathfrak{m}}$ is continuously
differentiable and%
\begin{equation}
\mathrm{Dh}_{L}^{\mathfrak{m}}=U_{L}^{\Phi }-\widehat{U_{L}^{\Phi }}%
+\sum_{n\in \mathbb{N}}\sum_{m=1}^{n}\int_{\mathbb{S}^{n}}\left( U_{L}^{\Psi
^{(m)}}-\widehat{U_{L}^{\Psi ^{(m)}}}\right) \prod\limits_{j\in \left\{
1,\ldots ,n\right\} ,j\neq m}\frac{\widehat{U_{L}^{\Psi ^{(j)}}}}{\left\vert
\Lambda _{L}\right\vert }\mathfrak{a}\left( t\right) _{n}\left( \mathrm{d}%
\Psi ^{(1)},\ldots ,\mathrm{d}\Psi ^{(n)}\right) \ ,
\label{derivative classical}
\end{equation}%
where $\mathrm{Dh}_{L}^{\mathfrak{m}}\in C(E;\mathcal{U}^{\mathbb{R}})=%
\mathfrak{U}^{\mathbb{R}}$ is the function defined by (\ref{clear2}) for $f=%
\mathrm{h}_{L}^{\mathfrak{m}}$. Moreover, by (\ref{norm Uphi}), (\ref{norm
properties}), (\ref{C1bis}) and (\ref{norm affine2}), for any $\mathfrak{m}%
\in \mathcal{M}$,%
\begin{equation*}
\left\Vert \mathrm{h}_{L}^{\mathfrak{m}}\right\Vert _{\mathfrak{Y}%
}=\left\Vert \mathrm{h}_{L}^{\mathfrak{m}}\left( \rho \right) \right\Vert _{%
\mathfrak{C}}+\max_{\rho \in E}\left\Vert \mathrm{Dh}_{L}^{\mathfrak{m}%
}\left( \rho \right) \right\Vert _{\mathcal{U}}\leq 3\left\vert \Lambda
_{L}\right\vert \left\Vert \mathbf{F}\right\Vert _{1,\mathfrak{L}}\left\Vert 
\mathfrak{m}\right\Vert _{\mathcal{M}}\ ,\qquad L\in \mathbb{N}\ ,
\end{equation*}%
which is similar to Inequality (\ref{energy bound long range}) that bounds
the norm of the Hamiltonians $U_{L}^{\mathfrak{m}}$, $L\in \mathbb{N}$.

Last but not least, it is very instructive to compare (\ref{derivative
classical}) with the local Hamiltonians associated with approximating
interactions of Definition \ref{definition BCS-type model approximated}, in
the light of Proposition \ref{density of periodic states copy(2)}: All terms
of the form 
\begin{equation*}
\frac{\widehat{U_{L}^{\Psi }}\left( \rho \right) }{\left\vert \Lambda
_{L}\right\vert }=\frac{\rho \left( U_{L}^{\Psi }\right) }{\left\vert
\Lambda _{L}\right\vert }\ ,\qquad \Psi \in \mathbb{S},\ \rho \in E,
\end{equation*}%
in (\ref{derivative classical}) should converge, as $L\rightarrow \infty $.
By Proposition \ref{density of periodic states copy(2)}, we know this holds
true whenever $\rho $ is a periodic state, the limit being $\rho (\mathfrak{e%
}_{\Phi ,\vec{\ell}})$ for some $\vec{\ell}\in \mathbb{N}^{d}$. The
periodicity of states is therefore a very useful\ property in this context.
Recall that periodic states form a weak$^{\ast }$-dense set $E_{\mathrm{p}}$
(\ref{set of periodic states}) of the physically relevant set $E^{+}$\ (\ref%
{gauge invariant states}) of even states, by Proposition \ref{density of
periodic states}.

So, restricting our study to periodic states, we are now in a position to
link the generator $\daleth ^{\mathfrak{m}}$ to a Poissonian symmetric
derivation:

\begin{corollary}[Classical evolutions via Poisson brackets]
\label{Closed Poissonian symmetric derivations copy(1)}\mbox{
}\newline
For any $\Lambda \in \mathcal{P}_{f}$\ and $\mathfrak{m}\in \mathcal{M}%
_{\Lambda }$, the dense $\ast $-subalgebra $\mathfrak{C}_{\mathcal{U}_{0}}$ (%
\ref{CU0}) belongs to the domain of $\daleth ^{\mathfrak{m}}$ and%
\begin{equation*}
\daleth ^{\mathfrak{m}}\left( f\right) |_{E_{\mathrm{p}}}=\lim_{L\rightarrow
\infty }\left\{ \mathrm{h}_{L}^{\mathfrak{m}},f\right\} |_{E_{\mathrm{p}%
}},\qquad f\in \mathfrak{C}_{\mathcal{U}_{0}}\ ,
\end{equation*}%
where the limit has to be understood point-wise on the weak$^{\ast }$-dense
subspace $E_{\mathrm{p}}\subseteq E^{+}$ of all periodic states.
\end{corollary}

\begin{proof}
Fix $\Lambda \in \mathcal{P}_{f}$\ and $\mathfrak{m}\in \mathcal{M}_{\Lambda
}$. Comparing Definitions \ref{dynamic series} and \ref{definition BCS-type
model approximated} with Definition \ref{convex Frechet derivative copy(1)}
and (\ref{derivative classical}) and using Proposition \ref{density of
periodic states copy(2)} and Corollary \ref{Lemma cigare1} together with
Lebesgue's dominated convergence theorem, one computes that%
\begin{equation}
\lim_{L\rightarrow \infty }\left\{ \mathrm{h}_{L}^{\mathfrak{m}},f\right\}
\left( \rho \right) =\rho \circ \delta ^{\Phi ^{(\mathfrak{m},\rho )}}\left( 
\mathrm{D}f\left( \rho \right) \right) \ ,\qquad f\in \mathfrak{C}_{\mathcal{%
U}_{0}},\ \rho \in E_{\mathrm{p}}\ .  \label{limit1}
\end{equation}%
Note that one can interchange the local quantum derivation (which is a
commutator) in the left-hand side of (\ref{limit1}) with every integral over 
$(\mathbb{S}\cap \mathcal{W}_{\Lambda })^{n}$, $n\in \mathbb{N}$, by finite
dimensionality of the space $\mathcal{U}_{\Lambda _{L}}$ for all $L\in 
\mathbb{N}$. Here, we use (\ref{identify}) to define the approximating
interaction $\Phi ^{(\mathfrak{m},\rho )}$, as well as the fact that $%
\mathrm{D}f\left( \rho \right) \in \mathcal{U}_{0}$ whenever $f\in \mathfrak{%
C}_{\mathcal{U}_{0}}$. Now, the rest of the proof is very similar to the one
of \cite[Theorem 4.5]{Bru-pedra-MF-I}: By Lemma \ref{Differentiability}, one
verifies that, for any $A\in \mathcal{U}_{0}$ and $t\in \mathbb{R}$, 
\begin{equation}
\partial _{t}V_{t,0}^{\mathfrak{m}}(\hat{A})\left( \rho \right) =V_{t,0}^{%
\mathfrak{m}}\circ \daleth ^{\mathfrak{m}}(\hat{A})\left( \rho \right) =%
\mathbf{\varpi }^{\mathfrak{m}}\left( 0,t;\rho \right) \circ \delta ^{\Phi
^{(\mathfrak{m},\mathbf{\varpi }^{\mathfrak{m}}\left( 0,t;\rho \right)
)}}\left( A\right) \ .  \label{limit1limit1}
\end{equation}%
Since $\daleth ^{\mathfrak{m}}$ and $\mathfrak{d}^{\mathrm{h}_{L}^{\mathfrak{%
m}}}$, $L\in \mathbb{N}$, are symmetric derivations, the assertion follows
by combining (\ref{limit1}) and (\ref{limit1limit1}).
\end{proof}

Equation (\ref{limit1limit1}) also holds true for any time-dependent model $%
\mathfrak{m}\in C_{b}(\mathbb{R};\mathcal{M}_{\Lambda })$ and $\Lambda \in 
\mathcal{P}_{f}$, with $\daleth ^{\mathfrak{m}}$ being replaced with $%
\daleth ^{\mathfrak{m}(t)}$ in (\ref{limit1limit1}). Therefore, in the
non-autonomous situation, for any $s,t\in \mathbb{R}$, $\rho \in E$ and
polynomial function $f\in \mathfrak{C}_{\mathcal{U}_{0}}$, 
\begin{equation}
\partial _{t}V_{t,s}^{\mathfrak{m}}\left( f\right) \left( \rho \right)
=V_{t,s}^{\mathfrak{m}}\circ \daleth ^{\mathfrak{m}(t)}\left( f\right)
\left( \rho \right) =\mathbf{\varpi }^{\mathfrak{m}}\left( s,t;\rho \right)
\circ \delta ^{\Phi ^{(\mathfrak{m},\mathbf{\varpi }^{\mathfrak{m}}\left(
s,t;\rho \right) )}}\left( f\left( \rho \right) \right) \ .  \label{sasdadsa}
\end{equation}%
Similar (point-wise) identities for $\partial _{s}V_{t,s}^{\mathfrak{m}}$
like 
\begin{equation}
\partial _{s}V_{t,s}^{\mathfrak{m}}\left( f\right) =-\daleth ^{\mathfrak{m}%
(s)}\circ V_{t,s}^{\mathfrak{m}}(f)  \label{ewrbis}
\end{equation}%
are not at all obvious. In fact, no unified theory of non-auto%
\-%
nomous evolution equations that gives a complete characterization of the
existence of fundamental solutions in terms of properties of generators,
analogously to the Hille-Yosida generation theorems for the autonomous case,
is available. See, e.g., \cite{Katobis,Caps,Schnaubelt1,Pazy,Bru-Bach} and
references therein.

An important, \emph{highly non-trivial}, result in that direction is proven
in the next theorem, which depends in a crucial way on Lieb-Robinson bounds
for multi-commutators of \cite[Theorems 4.11, 5.4]{brupedraLR}:

\begin{theorem}[Non-autonomous classical dynamics]
\label{classical dynamics I}\mbox{ }\newline
Fix $\Lambda \in \mathcal{P}_{f}$, $\varsigma ,\epsilon \in \mathbb{R}^{+}$
and take 
\begin{equation*}
\mathbf{F}\left( x,y\right) =\mathrm{e}^{-2\varsigma \left\vert
x-y\right\vert }(1+\left\vert x-y\right\vert )^{-(d+\epsilon )}\ ,\qquad
x,y\in \mathfrak{L}\ ,
\end{equation*}%
as the decay function, see (\ref{examples}). Then, $(V_{t,s}^{\mathfrak{m}%
})_{s,t\in \mathbb{R}}$ is a strongly continuous two-parameter family of $%
\ast $-automorphisms of $\mathfrak{C}$ satisfying, on the dense $\ast $%
-subalgebra $\mathfrak{C}_{\mathcal{U}_{0}}$ (\ref{CU0}),%
\begin{equation}
\forall s,t\in {\mathbb{R}}:\qquad \partial _{t}V_{t,s}^{\mathfrak{m}}\left(
f\right) |_{E_{\mathrm{p}}}=\lim_{L\rightarrow \infty }V_{t,s}^{\mathfrak{m}%
}\left( \{\mathrm{h}_{L}^{\mathfrak{m}(t)},f\}\right) |_{E_{\mathrm{p}}}\
,\qquad V_{s,s}^{\mathfrak{m}}=\mathbf{1}_{\mathfrak{C}}\ ,  \label{ewr}
\end{equation}%
for any $\mathfrak{m}\in C_{b}(\mathbb{R};\mathcal{M}_{\Lambda }\cap 
\mathcal{M}_{1})$ while, for any $\mathfrak{m}\in C_{b}(\mathbb{R};\mathcal{M%
}_{\Lambda })$, $s,t\in {\mathbb{R}}$ and $f\in \mathfrak{C}_{\mathcal{U}%
_{0}}$, $V_{t,s}^{\mathfrak{m}}(f)\in \mathfrak{Y}$ with%
\begin{equation}
\forall s,t\in {\mathbb{R}}:\qquad \partial _{s}V_{t,s}^{\mathfrak{m}}\left(
f\right) |_{E_{\mathrm{p}}}=-\lim_{L\rightarrow \infty }\{\mathrm{h}_{L}^{%
\mathfrak{m}(s)},V_{t,s}^{\mathfrak{m}}(f)\}|_{E_{\mathrm{p}}}\ ,\qquad
V_{t,t}^{\mathfrak{m}}=\mathbf{1}_{\mathfrak{C}}\ .  \label{ewrbis1}
\end{equation}%
All limits have to be understood point-wise on the weak$^{\ast }$-dense
subspace $E_{\mathrm{p}}\subseteq E^{+}$ of all periodic states.
\end{theorem}

\begin{proof}
Fix all parameters of the theorem. Equation (\ref{ewr}) results from (\ref%
{invariance translation}), (\ref{limit1}), (\ref{sasdadsa}), Lemma \ref%
{density of periodic states copy(2)} and the fact that $\mathfrak{m}\in
C_{b}(\mathbb{R};\mathcal{M}_{\Lambda }\cap \mathcal{M}_{1})$ (cf. (\ref%
{translation invaration encore})). In order to prove $V_{t,s}^{\mathfrak{m}%
}(f)\in \mathfrak{Y}$ and (\ref{ewrbis1}), it suffices to invoke Lemma \ref%
{Corollary bije+cocylbaby copy(1)}, which says that%
\begin{equation*}
\partial _{s}V_{t,s}^{\mathfrak{m}}(\hat{A})|_{E_{\mathrm{p}%
}}=-\lim_{L\rightarrow \infty }\{\mathrm{h}_{L}^{\mathfrak{m}(s)},V_{t,s}^{%
\mathfrak{m}}(\hat{A})\}|_{E_{\mathrm{p}}}
\end{equation*}%
for any $s,t\in \mathbb{R}$ and $A\in \mathcal{U}_{0}$. Since $(V_{t,s}^{%
\mathfrak{m}})_{s,t\in \mathbb{R}}$ is a family of $\ast $-automorphisms of $%
\mathfrak{C}$, by using the (bi)linearity and Leibniz's rule satisfied by
the derivatives and the bracket $\{\cdot ,\cdot \}$, we deduce that $%
V_{t,s}^{\mathfrak{m}}(f)\in \mathfrak{Y}$ and (\ref{ewrbis1}) for all
polynomial functions $f\in \mathfrak{C}_{\mathcal{U}_{0}}$ and times $s,t\in 
\mathbb{R}$.
\end{proof}

By Corollary \ref{Closed Poissonian symmetric derivations copy(1)} and
Theorem \ref{classical dynamics I}, Equation (\ref{ewrbis}) seems to hold
true, but it cannot be directly deduced from Theorem \ref{classical dynamics
I}. We refrain from doing such a study in this paper, since Theorem \ref%
{classical dynamics I} already shows that we have a non-autonomous classical
dynamics in the usual sense. Note that, in the case $\mathfrak{L}=\mathbb{Z}%
^{d}$, we have to consider the limit $L\rightarrow \infty $. This is the
classical counterpart of the thermodynamic limit $L\rightarrow \infty $ in
the derivations of Corollary \ref{Lemma cigare1}.

In the autonomous situation, as already suggested by Corollary \ref{Closed
Poissonian symmetric derivations copy(1)}, we obtain from Theorem \ref%
{classical dynamics I} the usual (autonomous) dynamics of classical
mechanics written in terms of Poisson brackets (see, e.g., \cite[Proposition
10.2.3]{classical-dynamics}), i.e., \emph{Liouville's equation}:

\begin{corollary}[Liouville's equation]
\label{classical dynamics I copy(1)}\mbox{
}\newline
Under conditions of Theorem \ref{classical dynamics I}, for any $t\in 
\mathbb{R}$ and $f\in \mathfrak{C}_{\mathcal{U}_{0}}$,%
\begin{equation*}
\partial _{t}V_{t,0}^{\mathfrak{m}}\left( f\right) =V_{t,0}^{\mathfrak{m}%
}\circ \daleth ^{\mathfrak{m}}\left( f\right) =V_{t,0}^{\mathfrak{m}}\left(
\lim_{L\rightarrow \infty }\{\mathrm{h}_{L}^{\mathfrak{m}},f\}\right)
=\lim_{L\rightarrow \infty }\left\{ \mathrm{h}_{L}^{\mathfrak{m}},V_{t,0}^{%
\mathfrak{m}}(f)\right\} =\daleth ^{\mathfrak{m}}\circ V_{t,0}^{\mathfrak{m}%
}\left( f\right) \ ,
\end{equation*}%
where all limits have to be understood point-wise on the weak$^{\ast }$%
-dense subspace $E_{\mathrm{p}}\subseteq E^{+}$ of all periodic states.
\end{corollary}

\begin{proof}
Combine Corollary \ref{Closed Poissonian symmetric derivations copy(1)} with
Theorem \ref{classical dynamics I}.
\end{proof}

Writing the classical dynamics in terms of Liouville's equation, as in
Corollary \ref{classical dynamics I copy(1)}, is conceptually illuminating
and also very useful from a purely mathematical point of view. For instance,
the Gross-Pitaevskii and Hartree hierarchies mathematically derived from
Bose gases with \emph{mean-field} interactions are \emph{infinite} systems
of coupled PDEs and therefore, a direct proof of the uniqueness of its
solutions is technically quite demanding, usually involving Feynman graphs,
multilinear estimates, etc. In \cite{Ammari2018}, the authors show that a
solution to these hierarchies is basically a family of time-dependent
correlation functions associated with a certain positive measure on the unit
ball of a $L^{2}$-space, whose dynamical evolution is driven by Liouville's
equation, similarly to Corollary \ref{classical dynamics I copy(1)}.
Uniqueness of a solution to such Liouville's equation can be proven in a
general setting, as shown in 2018 \cite{Ammari2018,Ammari2018-0}, implying
the uniqueness of a solution to the Gross-Pitaevskii and Hartree hierarchies
without highly technical issues depending on the particularities of the
hierarchies under consideration.

\subsection{Quantum Part of Long-Range Dynamics\label{Quantum Part}}

The classical part of the dynamics of lattice-fermion systems with
long-range interactions, which is defined within the classical $C^{\ast }$%
-algebras $\mathfrak{C}$ of continuous complex-valued functions on states,
is shown to result from self-consistency equations, as explained in Theorem %
\ref{theorem sdfkjsdklfjsdklfj}. Since $\mathfrak{C}$ can be seen as a
subalgebra (\ref{subset}) of the quantum $C^{\ast }$-algebras $\mathfrak{U}$
of continuous $\mathcal{U}$-valued functions on states, defined by (\ref%
{metaciagre set})-(\ref{metaciagre set bis}), there is a natural extension
of the classical dynamics on $\mathfrak{C}$: The continuous family $\mathbf{%
\varpi }^{\mathfrak{m}}$ of Theorem \ref{theorem sdfkjsdklfjsdklfj} yields a
family $(\mathfrak{V}_{t,s}^{\mathfrak{m}})_{s,t\in \mathbb{R}}$ of $\ast $%
-automorphisms of $\mathfrak{U}$ defined by%
\begin{equation*}
\mathfrak{V}_{t,s}^{\mathfrak{m}}\left( f\right) \doteq f\circ \mathbf{%
\varpi }^{\mathfrak{m}}\left( s,t\right) \ ,\qquad f\in \mathfrak{U},\
s,t\in \mathbb{R}\ .
\end{equation*}%
In particular, by (\ref{classical evolution family}), $\mathfrak{V}_{t,s}^{%
\mathfrak{m}}|_{\mathfrak{C}}=V_{t,s}^{\mathfrak{m}}$\ for any $s,t\in 
\mathbb{R}$.

\emph{However, it is not what we have in mind here}: Emphasizing rather the
inclusion $\mathcal{U}\subseteq \mathfrak{U}$, in the long-range dynamics,
the classical algebra $\mathfrak{C}$ becomes a subalgebra of the fixed-point
algebra of the state-dependent long-range dynamics on $\mathfrak{U}$. In 
\cite{Bru-pedra-MF-III} we describe in detail the quantum part of the
long-range dynamics, which is defined in a representation of the $C^{\ast }$%
-algebra $\mathfrak{U}$. We give below a few key points of this study:
\bigskip

\noindent \underline{(i):} As soon as the classical part of the long-range
dynamics is concerned, there is no need to impose any additional property on
initial states to define it. By contrast, for the quantum part, periodicity
of initial states is needed. Note, however, that the set $E_{\mathrm{p}}$ of
all periodic states is still a weak$^{\ast }$-dense subset of the physically
relevant set $E^{+}$ of all even states, by Proposition \ref{density of
periodic states}.\bigskip

\noindent \underline{(ii):} From now on, fix $\vec{\ell}\in \mathbb{N}^{d}$
and consider the set $E_{\vec{\ell}}$ of $\vec{\ell}$-periodic states
defined by (\ref{periodic invariant states}). Note that any $\vec{\ell}$%
-periodic state $\rho \in E_{\vec{\ell}}\subseteq \mathcal{U}^{\ast }$
naturally extends to a state on $\mathfrak{U}$: There is a natural
conditional expectation $\Xi $ from $\mathfrak{U}$ to $\mathfrak{C}\subseteq 
\mathfrak{U}$ defined by 
\begin{equation*}
\Xi \left( f\right) \left( \rho \right) =\rho \left( f\left( \rho \right)
\right) \ ,\qquad \rho \in E\ .
\end{equation*}%
Then, since $E_{\vec{\ell}}$ is a Choquet simplex, any state $\rho \in E_{%
\vec{\ell}}$ can be uniquely identified with its Choquet measure $\mu _{\rho
}$ which, in turn, is a state of $\mathfrak{C}$, canonically viewed as a
measure on $E$. The state of $\mathfrak{U}$ extending $\rho \in E_{\vec{\ell}%
}$ is the state $\mu _{\rho }\circ \Xi $. Moreover, if $\left( \mathcal{H}%
_{\rho },\pi _{\rho },\Omega _{\rho }\right) $ is a cyclic representation of 
$\rho \in E_{\vec{\ell}}$, then, by \cite[Proposition 4.2]{Bru-pedra-MF-III}%
, there is a representation 
\begin{equation*}
\Pi _{\rho }:\mathfrak{U}\rightarrow \mathcal{B}\left( \mathcal{H}_{\rho
}\right)
\end{equation*}%
such that%
\begin{equation*}
\lbrack \Pi _{\rho }\left( \mathfrak{U}\right) ]^{\prime \prime }=[\Pi
_{\rho }\left( \mathcal{U}\right) ]^{\prime \prime }\qquad \text{and}\qquad
\Pi _{\rho }\left( A\right) =\pi _{\rho }\left( A\right)
\end{equation*}%
for any $A\in \mathcal{U}\subseteq \mathfrak{U}$. In particular, 
\begin{equation*}
\lbrack \Pi _{\rho }\left( \mathfrak{C}\right) ]^{\prime \prime }\subseteq
\lbrack \pi _{\rho }\left( \mathcal{U}\right) ]^{\prime }\cap \lbrack \pi
_{\rho }\left( \mathcal{U}\right) ]^{\prime \prime }\ .
\end{equation*}%
In fact, for $\rho \in E_{\vec{\ell}}$, $\left( \mathcal{H}_{\rho },\Pi
_{\rho },\Omega _{\rho }\right) $ is a cyclic representation of the state $%
\mu _{\rho }\circ \Xi \in \mathfrak{U}^{\ast }$. The existence of such an
extension of $\pi _{\rho }$ strongly depends on the orthogonality of the ($%
\vec{\ell}$-) ergodic decomposition of $\vec{\ell}$-periodic states, as it
is explained in \cite{Bru-pedra-MF-III}. The ergodicity property of extreme
states of $E_{\vec{\ell}}$ is pivotal in order to get the limit dynamics
stated in the third point. \bigskip

\noindent \underline{(iii):} Let $\mathfrak{m}\in C_{b}(\mathbb{R};\mathcal{M%
}_{\Lambda }\cap \mathcal{M}_{1})$ for some $\Lambda \in \mathcal{P}_{f}$.
Assume that the state at initial time $s\in \mathbb{R}$ is $\rho \in E_{\vec{%
\ell}}$. Then, by taking the above cyclic representation $\left( \mathcal{H}%
_{\rho },\Pi _{\rho },\Omega _{\rho }\right) $ of $\rho $, seen as the state 
$\mu _{\rho }\circ \Xi $ on $\mathfrak{U}$, we show in \cite[Theorem 4.3]%
{Bru-pedra-MF-III} that, for any $t\in \mathbb{R}$ and $A\in \mathcal{U}%
\subseteq \mathfrak{U}$, in the thermodynamic limit $L\rightarrow \infty $,%
\begin{equation}
\Pi _{\rho }\left( \mathfrak{T}_{t,s}^{\mathbf{\Phi }^{(\mathfrak{m},\mathbf{%
\varpi }^{\mathfrak{m}}\left( \alpha ,\cdot \right) )}}\left( A\right)
|_{\alpha =s}-\tau _{t,s}^{(L,\mathfrak{m})}\left( A\right) \right)
\label{limit dynamics}
\end{equation}%
converges to $0$ in the $\sigma $-weak topology within $\mathcal{B}(\mathcal{%
H}_{\rho })$. In particular, the restriction to $\mathcal{U}$ of the state 
\begin{equation}
\rho \circ \mathfrak{T}_{t,s}^{\mathbf{\Phi }^{(\mathfrak{m},\mathbf{\varpi }%
^{\mathfrak{m}}\left( \alpha ,\cdot \right) )}}|_{\alpha =s}
\label{limit dynamics2}
\end{equation}%
can be seen as the state of the system at any time $t\in \mathbb{R}$ when $%
\rho $ is the (initial) state at time $t=s$. Here, 
\begin{equation*}
\mathbf{\varpi }^{\mathfrak{m}}\in C\left( \mathbb{R}^{2};\mathrm{Aut}\left(
E\right) \right)
\end{equation*}%
\ results from Theorem \ref{theorem sdfkjsdklfjsdklfj}, $\mathbf{\Phi }^{(%
\mathfrak{m},\mathbf{\xi })}$ is the state-dependent interaction of
Definition \ref{definition BCS-type model approximated} for 
\begin{equation*}
\mathbf{\xi }\in C\left( \mathbb{R};\mathrm{Aut}\left( E\right) \right) \ ,
\end{equation*}%
$(\mathfrak{T}_{t,s}^{\mathbf{\Psi }})_{_{s,t\in \mathbb{R}}}$ is the
strongly continuous two-para%
\-%
meter family of $\ast $-auto%
\-%
morphisms of $\mathfrak{U}$ of Proposition \ref{Lemma cigare1 copy(5)} for $%
\mathbf{\Psi }\in C(\mathbb{R};\mathfrak{W}^{\mathbb{R}})$, and $(\tau
_{t}^{(L,\mathfrak{m})})_{_{t\in \mathbb{R}}}$ is the strongly continuous
one-parameter group of $\ast $-auto%
\-%
morphisms of $\mathcal{U}$ defined by (\ref{definition fininte vol dynam}).
Note that, even if the local dynamics is autonomous, the limit dynamics can
still be \emph{non-autonomous}.

A similar result holds true for non-autonomous long-range dynamics, i.e.,
for all time-dependent models 
\begin{equation*}
\mathfrak{m}\in C_{b}(\mathbb{R};\mathcal{M}_{\Lambda }\cap \mathcal{M}%
_{1})\ ,\qquad \Lambda \in \mathcal{P}_{f}\ .
\end{equation*}%
In this case, one replaces the autonomous local dynamics in (\ref{limit
dynamics}) with the non-autonomous local one, similar to (\ref{cauchy1})-(%
\ref{cauchy2}).

\section{Technical Proofs\label{Well-posedness sect}}

The aim of this section is to prove Theorems \ref{theorem sdfkjsdklfjsdklfj}
and \ref{classical dynamics I}. In fact, we prove here stronger results than
these theorems. The proof of Theorem \ref{theorem sdfkjsdklfjsdklfj} is done
in five lemmata and two corollaries. The proof of Theorem \ref{classical
dynamics I} is a direct consequence of Lemma \ref{Corollary bije+cocylbaby
copy(1)}. Note that those proofs are a much more involved version of the
ones performed in \cite[Section 7]{Bru-pedra-MF-I} to prove \cite[Theorems
4.1 and 4.6]{Bru-pedra-MF-I}.

We start with preliminary definitions: In all the present section, fix once
and for all $\vec{\ell}\in \mathbb{N}^{d}$, $\Lambda \in \mathcal{P}_{f}$
and a time-dependent model 
\begin{equation*}
\mathfrak{m=((}\Phi \left( t\right) ,\mathfrak{\mathfrak{a}}\left( t\right) 
\mathfrak{))}_{t\in \mathbb{R}}\in C_{b}(\mathbb{R};\mathcal{M}_{\Lambda })\
,\qquad \text{with}\qquad \mathcal{M}_{\Lambda }\doteq \mathcal{W}^{\mathbb{R%
}}\times \mathcal{S}_{\Lambda }\subseteq \mathcal{M}\ ,
\end{equation*}%
(see (\ref{S00bis})), $\mathcal{S}_{\Lambda }\subseteq \mathcal{S}$ being
defined by (\ref{S0bis}). By (\ref{eq:enpersitebisbis})-(\ref%
{eq:enpersitebisbisbis}), we can assume without loss of generality that $%
\mathfrak{e}_{\Phi ,\vec{\ell}}\in \mathcal{U}_{\Lambda }$. In order to
simplify mathematical expressions, we use the standard notation%
\begin{equation}
\left\Vert \mathfrak{m}\right\Vert _{\infty }\equiv \left\Vert \mathfrak{m}%
\right\Vert _{C_{b}\left( \mathbb{R};\mathcal{M}\right) }\doteq \sup_{t\in 
\mathbb{R}}\left\Vert \mathfrak{m}\left( t\right) \right\Vert _{\mathcal{M}%
}\ ,\qquad \mathfrak{m}\in C_{b}(\mathbb{R};\mathcal{M}_{\Lambda })\ .
\label{norm m}
\end{equation}

Recall that $E_{\Lambda }\subseteq \mathcal{U}_{\Lambda }^{\ast }$, which is
defined by (\ref{local states}), is the norm-compact set of states on the
finite-dimensional $C^{\ast }$-algebra $\mathcal{U}_{\Lambda }$. For every
continuous function $\zeta \in C\left( \mathbb{R};E_{\Lambda }\right) $, let%
\begin{equation}
\Psi ^{\zeta }\left( t\right) \doteq \Phi \left( t\right) +\sum_{n\in 
\mathbb{N}}\int_{(\mathbb{S}\cap \mathcal{W}_{\Lambda })^{n}}\ \left\lfloor
\zeta \left( t\right) ;\Psi ^{(1)},\ldots ,\Psi ^{(n)}\right\rfloor _{\vec{%
\ell}}\ \mathfrak{a}\left( t\right) _{n}\left( \mathrm{d}\Psi ^{(1)},\ldots ,%
\mathrm{d}\Psi ^{(n)}\right) \,,\qquad t\in \mathbb{R}\ ,  \label{def utile}
\end{equation}%
with $\mathcal{W}_{\Lambda }\subseteq \mathcal{W}_{1}$ being defined by (\ref%
{eq:enpersitebis}). In particular, by Definition \ref{definition BCS-type
model approximated} and Equation (\ref{S0bis}), for all continuous functions 
$\xi \in C\left( \mathbb{R};E\right) $,%
\begin{equation}
\Psi ^{\xi |_{\mathcal{U}_{\Lambda }}}\left( t\right) =\Phi ^{(\mathfrak{m}%
,\xi )}\left( t\right) \,,\qquad t\in \mathbb{R}\ .
\label{equality cigare00}
\end{equation}%
Note that such approximating interactions can be used to define a strongly
continuous two-para%
\-%
meter family $(\tau _{t,s}^{\Psi ^{\zeta }})_{s,t\in \mathbb{R}}$ of $\ast $%
-auto%
\-%
morphisms of $\mathcal{U}$ for any $\zeta \in C\left( \mathbb{R};E_{\Lambda
}\right) $, by Proposition \ref{Theorem Lieb-Robinson} and Lemma \ref%
{definition BCS-type model approximated copy(2)}. Such approximating
dynamics satisfy the following estimate:

\begin{lemma}[Estimates on approximating dynamics]
\label{estimate useful}\mbox{ }\newline
For any $s,t\in \mathbb{R}$, $\tilde{\Lambda}\in \mathcal{P}_{f}$, $A\in 
\mathcal{U}_{\tilde{\Lambda}}$ and $\zeta _{1},\zeta _{2}\in C\left( \mathbb{%
R};E_{\Lambda }\right) \subseteq C\left( \mathbb{R};\mathcal{U}_{\Lambda
}^{\ast }\right) $,%
\begin{equation*}
\left\Vert \left( \tau _{t,s}^{\Psi ^{\zeta _{1}}}-\tau _{t,s}^{\Psi ^{\zeta
_{2}}}\right) \left( A\right) \right\Vert _{\mathcal{U}}\leq 2|\tilde{\Lambda%
}|\left\Vert A\right\Vert _{\mathcal{U}}\left\Vert \mathbf{F}\right\Vert _{1,%
\mathfrak{L}}\left\Vert \mathfrak{m}\right\Vert _{\infty }\mathrm{e}^{2%
\mathbf{D}\left\Vert \mathfrak{m}\right\Vert _{\infty }\left\vert
t-s\right\vert }\int_{t\wedge s}^{t\vee s}\left\Vert \zeta _{1}\left( \alpha
\right) -\zeta _{2}\left( \alpha \right) \right\Vert _{\mathcal{U}_{\Lambda
}^{\ast }}\mathrm{d}\alpha \ .
\end{equation*}%
Here, $\mathcal{U}_{\Lambda }^{\ast }$ is endowed with the usual norm for
continuous linear functionals.
\end{lemma}

\begin{proof}
Similar to Inequality (\ref{totototot}), by Equations (\ref{e phi}), (\ref%
{definition 0})-(\ref{definition 0bis}) and (\ref{equality cigare00}), note
that, for any $\zeta _{1},\zeta _{2}\in C\left( \mathbb{R};E_{\Lambda
}\right) $,%
\begin{equation*}
\left\Vert \Psi ^{\zeta _{1}}\left( \alpha \right) -\Psi ^{\zeta _{2}}\left(
\alpha \right) \right\Vert _{\mathcal{W}}\leq \left\Vert \mathfrak{m}\left(
\alpha \right) \right\Vert _{\mathcal{M}}\left\Vert \zeta _{1}\left( \alpha
\right) -\zeta _{2}\left( \alpha \right) \right\Vert _{\mathcal{U}_{\Lambda
}^{\ast }}\ ,\qquad \alpha \in \mathbb{R}\ .
\end{equation*}%
Combining this inequality with Proposition \ref{Theorem Lieb-Robinson
copy(3)} (iii) and Lemma \ref{definition BCS-type model approximated copy(2)}%
, we obtain the assertion. Note that $C\left( \mathbb{R};E_{\Lambda }\right)
\subseteq C\left( \mathbb{R};\mathcal{U}_{\Lambda }^{\ast }\right) $,
because the norm and weak$^{\ast }$ topologies of $\mathcal{U}_{\Lambda
}^{\ast }$ are the same, by finite dimensionality of $\mathcal{U}_{\Lambda }$
for $\Lambda \in \mathcal{P}_{f}$.
\end{proof}

We now show the existence and uniqueness of the solution to the
self-consistency equation:

\begin{lemma}[Self-consistency equations]
\label{Solution self}\mbox{ }\newline
For any $s\in \mathbb{R}$ and $\rho \in E$, there is a unique solution $%
\varpi _{\rho ,s}$ to the following equation in $\xi \in C\left( \mathbb{R}%
;E\right) $:%
\begin{equation}
\forall t\in {\mathbb{R}}:\qquad \xi \left( t\right) =\rho \circ \tau
_{t,s}^{\Psi ^{\xi |_{\mathcal{U}_{\Lambda }}}}\ .
\label{self consitence equation1}
\end{equation}%
Moreover, $\varpi _{\rho ,s}(t)=\varpi _{\varpi _{\rho ,s}(r),r}(t)$ for any 
$r,s,t\in {\mathbb{R}}$.
\end{lemma}

\begin{proof}
The proof is similar to the one of \cite[Lemma 7.3]{Bru-pedra-MF-I}: Fix the
initial time $s\in \mathbb{R}$ and state $\rho \in E$. The existence and
uniqueness of a solution $\varpi _{\rho ,s}$ to (\ref{self consitence
equation1}) is proven via the Banach fixed point theorem: \medskip

\noindent \underline{Step 1:} Fix $T\in \mathbb{R}^{+}$ and observe that $%
C\left( [s-T,s+T];E_{\Lambda }\right) $ is a closed bounded subset of the
Banach space $C\left( [s-T,s+T];\mathcal{U}_{\Lambda }^{\ast }\right) $,
with $\mathcal{U}_{\Lambda }^{\ast }$ being endowed with the usual norm for
continuous linear functionals and 
\begin{equation*}
\left\Vert \zeta \right\Vert _{C([s-T,s+T];\mathcal{U}_{\Lambda }^{\ast
})}\doteq \sup_{t\in \lbrack s-T,s+T]}\left\Vert \zeta \left( t\right)
\right\Vert _{\mathcal{U}_{\Lambda }^{\ast }}\ ,\qquad \zeta \in C\left(
[s-T,s+T];\mathcal{U}_{\Lambda }^{\ast }\right) \ .
\end{equation*}%
Define the mapping $\mathfrak{F}$ from $C\left( [s-T,s+T];E_{\Lambda
}\right) $ to itself by%
\begin{equation}
\mathfrak{F}\left( \zeta \right) \left( t\right) \doteq \left. \rho \circ
\tau _{t,s}^{\Psi ^{\zeta }}\right\vert _{\mathcal{U}_{\Lambda }}\ ,\qquad
t\in \lbrack s-T,s+T]\ .  \label{label}
\end{equation}%
The existence of such a $\ast $-automorphism $\tau _{t,s}^{\Psi ^{\zeta }}$
for any $\zeta \in C\left( [s-T,s+T];E_{\Lambda }\right) $ and $t\in \lbrack
s-T,s+T]$ follows from Proposition \ref{Theorem Lieb-Robinson} and Lemma \ref%
{definition BCS-type model approximated copy(2)}. By Proposition \ref%
{Theorem Lieb-Robinson} (or Proposition \ref{Theorem Lieb-Robinson copy(3)}
(iv)), (\ref{label}) defines a mapping from $C\left( [s-T,s+T];E_{\Lambda
}\right) $ to itself. Moreover, by (\ref{def utile}) and (\ref{label}), we
infer from Lemma \ref{estimate useful} that, for any $\zeta _{1},\zeta
_{2}\in C\left( [s-T,s+T];E_{\Lambda }\right) $,%
\begin{eqnarray}
&&\left\Vert \mathfrak{F}\left( \zeta _{1}\right) \left( t\right) -\mathfrak{%
F}\left( \zeta _{2}\right) \left( t\right) \right\Vert _{C([s-T,s+T];%
\mathcal{U}_{\Lambda }^{\ast })}  \label{bound0002} \\
&\leq &4T\left\vert \Lambda \right\vert \left\Vert \mathbf{F}\right\Vert _{1,%
\mathfrak{L}}\left\Vert \mathfrak{m}\right\Vert _{\infty }\mathrm{e}^{4%
\mathbf{D}T\left\Vert \mathfrak{m}\right\Vert _{\infty }}\left\Vert \zeta
_{1}-\zeta _{2}\right\Vert _{C([s-T,s+T];\mathcal{U}_{\Lambda }^{\ast })}\ .
\notag
\end{eqnarray}%
Therefore, by fixing the time parameter $T\in \mathbb{R}^{+}$ such that 
\begin{equation}
T\leq \frac{\mathrm{e}^{-4\mathbf{D}T\left\Vert \mathfrak{m}\right\Vert
_{\infty }}}{8\left\vert \Lambda \right\vert \left\Vert \mathbf{F}%
\right\Vert _{1,\mathfrak{L}}\left\Vert \mathfrak{m}\right\Vert _{\infty }}\
,  \label{bound0003}
\end{equation}%
the function $\mathfrak{F}$ is a contraction. Hence, we obtain a unique
solution $\gimel _{\rho ,s}\in C\left( [s-T,s+T];E_{\Lambda }\right) $ to
Equation (\ref{self consitence equation1}) with $\xi |_{\mathcal{U}_{\Lambda
}}=\gimel _{\rho ,s}$ at fixed $s\in \mathbb{R}$ and $\rho \in E$. \medskip

\noindent \underline{Step 2:} By (\ref{equality cigare00}), the restriction
of any solution $\varpi _{\rho ,s}\in C\left( [s-T,s+T];E\right) $ to (\ref%
{self consitence equation1}) to the subspace $\mathcal{U}_{\Lambda
}\subseteq \mathcal{U}$ must equal $\gimel _{\rho ,s}\in C\left(
[s-T,s+T];E_{\Lambda }\right) $ and $\Phi ^{(\mathfrak{m},\varpi _{\rho
,s}\left( t\right) )}=\Psi ^{\gimel _{\rho ,s}}\left( t\right) $. With this
observation, we see that 
\begin{equation*}
\varpi _{\rho ,s}\left( t\right) \doteq \rho \circ \tau _{t,s}^{\Psi
^{\gimel _{\rho ,s}}},\qquad t\in \lbrack s-T,s+T]\ ,
\end{equation*}%
is the unique solution in $C\left( [s-T,s+T];E\right) $ to (\ref{self
consitence equation1}) at fixed initial time $s\in \mathbb{R}$ and state $%
\rho \in E$. \medskip

\noindent \underline{Step 3:} In the same way we prove the existence and
uniqueness of a solution to (\ref{self consitence equation1}) at fixed $s\in 
\mathbb{R}$ and $\rho \in E$, one shows that, for each $r\in \left[ s-T,s+T%
\right] $, the self-consistency equation 
\begin{equation}
\forall t\in \lbrack r-\tilde{T},r+\tilde{T}]:\qquad \xi \left( t\right)
=\varpi _{\rho ,s}\left( r\right) \circ \tau _{t,r}^{\Psi ^{\xi |_{\mathcal{U%
}_{\Lambda }}}}\ ,  \label{solutionplus}
\end{equation}%
has also a unique solution $\varpi _{\varpi _{\rho ,s}(r),r}$ in $C([r-%
\tilde{T},r+\tilde{T}];E)$ for any $\tilde{T}\in (0,T]$. By the reverse
cocycle property (\ref{reverse cocycle}), at fixed $s\in \mathbb{R}$ and $%
\rho \in E$, $\varpi _{\rho ,s}$ solves the self-consistency equation (\ref%
{solutionplus}) for any $r\in (s-T,s+T)$ and $t\in \lbrack s-\tilde{T},s+%
\tilde{T}]$ with $\tilde{T}=T-|s-r|\in \mathbb{R}^{+}$, whence 
\begin{equation*}
\varpi _{\rho ,s}(t)=\varpi _{\varpi _{\rho ,s}(r),r}(t)
\end{equation*}%
for any $r\in (s-T,s+T)$ and $t\in \lbrack s-\tilde{T},s+\tilde{T}]$.\medskip

\noindent \underline{Step 4:} Assume the existence and uniqueness of a
solution $\varpi _{\rho ,s}$ in $C\left( [s-T_{0},s+T_{0}];E\right) $ to
Equation (\ref{self consitence equation1}) for some parameter $T_{0}\in 
\mathbb{R}^{+}$. Take 
\begin{equation*}
r\in (s-T_{0},s-T_{0}+T)\cup (s+T_{0}-T,s+T_{0})\ .
\end{equation*}%
By combining the existence and uniqueness of a solution $\varpi _{\varpi
_{\rho ,s}(r),r}$ to (\ref{solutionplus}) in $C([r-\tilde{T},r+\tilde{T}];E)$
for any $\tilde{T}\in (0,T]$ together with the reverse cocycle property (\ref%
{reverse cocycle}), we deduce that 
\begin{equation*}
\varpi _{\rho ,s}(t)=\varpi _{\varpi _{\rho ,s}(r),r}(t)\ ,\qquad t\in
(s-T_{0},s+T_{0})\ ,
\end{equation*}%
as well as the existence of a unique solution $\varpi _{\rho ,s}$ to (\ref%
{self consitence equation1}) in $C\left( [s-T_{0}-T,s+T_{0}+T];E\right) $.
As a consequence, one can infer from a contradiction argument the existence
and uniqueness of a solution in $\xi \in C\left( \mathbb{R};E\right) $ to
the self-consistency equation (\ref{self consitence equation1}). Moreover,
this solution must satisfy the equality $\varpi _{\rho ,s}(t)=\varpi
_{\varpi _{\rho ,s}(r),r}(t)$ for any $r,s,t\in {\mathbb{R}}$.
\end{proof}

\begin{corollary}[Bijectivity of the solution to the self-consistency
equation]
\label{bijectivity}\mbox{ }\newline
For any $s,t\in \mathbb{R}$, $\varpi _{s}\left( t\right) \equiv (\varpi
_{\rho ,s}\left( t\right) )_{\rho \in E}$ is a bijective mapping from $E$ to
itself.
\end{corollary}

\begin{proof}
The assertion is a direct consequence of Lemma \ref{Solution self}, in
particular the equality $\varpi _{\rho ,s}(t)=\varpi _{\varpi _{\rho
,s}(r),r}(t)$ for $r,s,t\in {\mathbb{R}}$.
\end{proof}

By combining Lemma \ref{Solution self} with Proposition \ref{Theorem
Lieb-Robinson}, note that, for any $s\in \mathbb{R}$, $\rho \in E$ and $A\in 
\mathcal{U}_{0}$, 
\begin{equation*}
\partial _{t}\left\{ \varpi _{\rho ,s}\left( t\right) \left( A\right)
\right\} =\rho \circ \tau _{t,s}^{\Psi ^{\gimel }}\circ \delta ^{\Psi
^{\gimel }\left( t\right) }\left( A\right) \ ,\qquad t\in \mathbb{R}\ ,
\end{equation*}%
with $\gimel \doteq \varpi _{\rho ,s}|_{\mathcal{U}_{\Lambda }}$. This
property can be strengthened within the local $C^{\ast }$-algebra $\mathcal{U%
}_{\Lambda }$:

\begin{lemma}[Differentiability of the solution -- I]
\label{Differentiability}\mbox{ }\newline
For any $s\in \mathbb{R}$ and $\rho \in E$, $\gimel \doteq \varpi _{\rho
,s}|_{\mathcal{U}_{\Lambda }}\in C^{1}\left( \mathbb{R};\mathcal{U}_{\Lambda
}^{\ast }\right) $ with derivative given by%
\begin{equation*}
\partial _{t}\left\{ \varpi _{\rho ,s}\left( t\right) |_{\mathcal{U}%
_{\Lambda }}\right\} =\rho \circ \tau _{t,s}^{\Psi ^{\gimel }}\circ \delta
^{\Psi ^{\gimel }\left( t\right) }|_{\mathcal{U}_{\Lambda }}\ ,\qquad t\in 
\mathbb{R}\ .
\end{equation*}%
Here, $\mathcal{U}_{\Lambda }^{\ast }$ is endowed with the usual norm for
continuous linear functionals.
\end{lemma}

\begin{proof}
To prove that $\gimel \doteq \varpi _{\rho ,s}|_{\mathcal{U}_{\Lambda }}\in
C^{1}\left( \mathbb{R};\mathcal{U}_{\Lambda }^{\ast }\right) $ at fixed $%
s\in \mathbb{R}$ and $\rho \in E$, we first remark that, for any $A\in 
\mathcal{U}_{\Lambda }$ and $h\in \mathbb{R}\backslash \{0\}$, 
\begin{eqnarray}
&&\left\vert h^{-1}\left( \rho \circ \tau _{t+h,s}^{\Psi ^{\gimel }}\left(
A\right) -\rho \circ \tau _{t,s}^{\Psi ^{\gimel }}\left( A\right) \right)
-\rho \circ \tau _{t,s}^{\Psi ^{\gimel }}\circ \delta ^{\Psi ^{\gimel \left(
t\right) }}\left( A\right) \right\vert  \label{machin} \\
&\leq &\max_{\alpha \in \left[ t-h,t+h\right] }\left\Vert \delta ^{\Psi
^{\gimel \left( \alpha \right) }-\Psi ^{\gimel \left( t\right) }}\left(
A\right) \right\Vert _{\mathcal{U}}+\max_{\alpha \in \left[ t-h,t+h\right]
}\left\Vert \left( \tau _{\alpha ,s}^{\Psi ^{\gimel }}-\tau _{t,s}^{\Psi
^{\gimel }}\right) \circ \delta ^{\Psi ^{\gimel \left( t\right) }}\left(
A\right) \right\Vert _{\mathcal{U}}\ .  \notag
\end{eqnarray}%
Since $\gimel \in C\left( \mathbb{R};E_{\Lambda }\right) $, by Corollary \ref%
{Lemma cigare1}, for any $A\in \mathcal{U}_{\Lambda }$ with $\left\Vert
A\right\Vert _{\mathcal{U}}=1$, 
\begin{equation*}
\max_{\alpha \in \left[ t-h,t+h\right] }\left\Vert \delta ^{\Psi ^{\gimel
\left( \alpha \right) }-\Psi ^{\gimel \left( t\right) }}\left( A\right)
\right\Vert _{\mathcal{U}}\leq 2\left\vert \Lambda \right\vert \left\Vert 
\mathbf{F}\right\Vert _{1,\mathfrak{L}}\max_{\alpha \in \left[ t-h,t+h\right]
}\left\Vert \Psi ^{\gimel \left( \alpha \right) }-\Psi ^{\gimel \left(
t\right) }\right\Vert _{\mathcal{W}}
\end{equation*}%
and hence, by using Lemma \ref{definition BCS-type model approximated
copy(2)}, we arrive at 
\begin{equation}
\lim_{h\rightarrow 0}\sup_{A\in \mathcal{U}_{\Lambda },\left\Vert
A\right\Vert _{\mathcal{U}}=1}\max_{\alpha \in \left[ t-h,t+h\right]
}\left\Vert \delta ^{\Psi ^{\gimel \left( \alpha \right) }-\Psi ^{\gimel
\left( t\right) }}\left( A\right) \right\Vert _{\mathcal{U}}=0\ .
\label{tactac}
\end{equation}%
Meanwhile, by Proposition \ref{Lemma cigare0}, for any $A\in \mathcal{U}%
_{\Lambda }$ satisfying $\left\Vert A\right\Vert _{\mathcal{U}}=1$, 
\begin{multline*}
\left\Vert \left( \tau _{\alpha ,s}^{\Psi ^{\gimel }}-\tau _{t,s}^{\Psi
^{\gimel }}\right) \circ \delta ^{\Psi ^{\gimel }\left( t\right) }\left(
A\right) \right\Vert _{\mathcal{U}}\leq \left\Vert \left( \tau _{\alpha
,s}^{\Psi ^{\gimel }}-\tau _{t,s}^{\Psi ^{\gimel }}\right) \circ \delta
_{L}^{\Psi ^{\gimel \left( t\right) }}\left( A\right) \right\Vert _{\mathcal{%
U}} \\
+2\left\vert \Lambda \right\vert \left\Vert \Psi ^{\gimel \left( t\right)
}\right\Vert _{\mathcal{W}}\sup_{y\in \Lambda }\sum\limits_{x\in \Lambda
_{L}^{c}}\mathbf{F}\left( x,y\right) \ .
\end{multline*}%
Thus, by (\ref{(3.1) NS}), for any fixed $\varepsilon \in \mathbb{R}^{+}$,
there is $L\in \mathbb{N}$ such that, for any $A\in \mathcal{U}_{\Lambda }$
with $\left\Vert A\right\Vert _{\mathcal{U}}=1$,%
\begin{equation*}
\max_{\alpha \in \left[ t-h,t+h\right] }\left\Vert \left( \tau _{\alpha
,s}^{\Psi ^{\gimel }}-\tau _{t,s}^{\Psi ^{\gimel }}\right) \circ \delta
^{\Psi ^{\gimel }\left( t\right) }\left( A\right) \right\Vert _{\mathcal{U}%
}\leq \max_{\alpha \in \left[ t-h,t+h\right] }\left\Vert \left( \tau
_{\alpha ,s}^{\Psi ^{\gimel }}-\tau _{t,s}^{\Psi ^{\gimel }}\right) \circ
\delta _{L}^{\Psi ^{\gimel \left( t\right) }}\left( A\right) \right\Vert _{%
\mathcal{U}}+\varepsilon \ ,
\end{equation*}%
while we obtain from Proposition \ref{Theorem Lieb-Robinson copy(3)} (iv)
that 
\begin{equation*}
\lim_{h\rightarrow 0}\sup_{A\in \mathcal{U}_{\Lambda },\left\Vert
A\right\Vert _{\mathcal{U}}=1}\max_{\alpha \in \left[ t-h,t+h\right]
}\left\Vert \left( \tau _{\alpha ,s}^{\Psi ^{\gimel }}-\tau _{t,s}^{\Psi
^{\gimel }}\right) \circ \delta _{L}^{\Psi ^{\gimel \left( t\right) }}\left(
A\right) \right\Vert _{\mathcal{U}}=0\ .
\end{equation*}%
It follows that 
\begin{equation*}
\lim_{h\rightarrow 0}\sup_{A\in \mathcal{U}_{\Lambda },\left\Vert
A\right\Vert _{\mathcal{U}}=1}\max_{\alpha \in \left[ t-h,t+h\right]
}\left\Vert \left( \tau _{\alpha ,s}^{\Psi ^{\gimel }}-\tau _{t,s}^{\Psi
^{\gimel }}\right) \circ \delta ^{\Psi ^{\gimel }\left( t\right) }\left(
A\right) \right\Vert _{\mathcal{U}}=0\ .
\end{equation*}%
We finally combine the last limit with (\ref{machin})-(\ref{tactac}) to
deduce that $\gimel \in C^{1}\left( \mathbb{R};\mathcal{U}_{\Lambda }^{\ast
}\right) $ with derivative given by%
\begin{equation*}
\partial _{t}\left\{ \varpi _{\rho ,s}\left( t\right) |_{\mathcal{U}%
_{\Lambda }}\right\} =\partial _{t}\gimel \left( t\right) =\rho \circ \tau
_{t,s}^{\Psi ^{\gimel }}\circ \delta ^{\Psi ^{\gimel }\left( t\right) }|_{%
\mathcal{U}_{\Lambda }}\ ,\qquad t\in \mathbb{R}\ ,
\end{equation*}%
at any fixed $s\in \mathbb{R}$ and $\rho \in E$.
\end{proof}

\begin{lemma}[Continuity with respect to the initial condition]
\label{lemma contnuity}\mbox{ }\newline
For any $s,t\in \mathbb{R}$, $\varpi _{s}\left( t\right) \equiv (\varpi
_{\rho ,s}\left( t\right) )_{\rho \in E}\in C\left( E;E\right) $.
\end{lemma}

\begin{proof}
Take $s\in \mathbb{R}$ and two states $\rho _{1},\rho _{2}\in E$. Then,
define the quantity 
\begin{equation}
\mathbf{X}\left( T\right) \doteq \max_{t\in \left[ s-T,s+T\right]
}\left\Vert \left( \varpi _{\rho _{1},s}\left( t\right) -\varpi _{\rho
_{2},s}\left( t\right) \right) |_{\mathcal{U}_{\Lambda }}\right\Vert _{%
\mathcal{U}_{\Lambda }^{\ast }}\ ,\qquad T\in \mathbb{R}^{+}\ .
\label{sdfsdf}
\end{equation}%
By Proposition \ref{Theorem Lieb-Robinson copy(3)} (ii) and Lemma \ref%
{definition BCS-type model approximated copy(2)} together with Equations (%
\ref{(3.1) NS}) and (\ref{equality cigare00}), for any $T,\varepsilon \in 
\mathbb{R}^{+}$, there is $L\in \mathbb{N}$ such that, for any $A\in 
\mathcal{U}_{\Lambda }$ with $\left\Vert A\right\Vert _{\mathcal{U}}=1$, 
\begin{equation}
\sup_{t\in \left[ s-T,s+T\right] }\sup_{\zeta \in C\left( \mathbb{R}%
;E_{\Lambda }\right) }\left\Vert \tau _{t,s}^{\Psi ^{\zeta }}\left( A\right)
-\tau _{t,s}^{(L,\Psi ^{\zeta })}\left( A\right) \right\Vert _{\mathcal{U}%
}\leq \varepsilon \ .  \label{bound000}
\end{equation}%
By combining Lemmata \ref{estimate useful} (for $\mathfrak{L}=\Lambda _{L}$)
and \ref{Solution self} with (\ref{bound000}), we thus obtain the bound%
\begin{equation*}
\mathbf{X}\left( T\right) \leq 2\varepsilon +\left\Vert \left( \rho
_{1}-\rho _{2}\right) |_{\mathcal{U}_{\Lambda _{L}}}\right\Vert _{\mathcal{U}%
_{\Lambda _{L}}^{\ast }}+2\left\vert \Lambda \right\vert \left\Vert \mathbf{F%
}\right\Vert _{1,\mathfrak{L}}\left\Vert \mathfrak{m}\right\Vert _{\infty }%
\mathrm{e}^{2\mathbf{D}T\left\Vert \mathfrak{m}\right\Vert _{\infty
}}\int_{0}^{T}\mathbf{X}\left( \alpha \right) \mathrm{d}\alpha \ ,
\end{equation*}%
where $\left\Vert \mathfrak{m}\right\Vert _{\infty }$ is defined by (\ref%
{norm m}). By Gr\"{o}nwall's inequality, 
\begin{equation}
\mathbf{X}\left( T\right) \leq \left( 2\varepsilon +\left\Vert \left( \rho
_{1}-\rho _{2}\right) |_{\mathcal{U}_{\Lambda _{L}}}\right\Vert _{\mathcal{U}%
_{\Lambda _{L}}^{\ast }}\right) \left( 1+2\left\vert \Lambda \right\vert
\left\Vert \mathbf{F}\right\Vert _{1,\mathfrak{L}}\left\Vert \mathfrak{m}%
\right\Vert _{\infty }\mathrm{e}^{2\mathbf{D}T\left\Vert \mathfrak{m}%
\right\Vert _{\infty }}\mathbf{Y}\left( T\right) \right)  \label{gronwal}
\end{equation}%
with 
\begin{equation}
\mathbf{Y}\left( T\right) \doteq \int_{0}^{T}\exp \left\{ 2\left\vert
\Lambda \right\vert \left\Vert \mathbf{F}\right\Vert _{1,\mathfrak{L}%
}\left\Vert \mathfrak{m}\right\Vert _{\infty }\mathrm{e}^{2\mathbf{D}%
T\left\Vert \mathfrak{m}\right\Vert _{\infty }}\alpha \right\} \mathrm{d}%
\alpha \ .  \label{gronwal2}
\end{equation}%
By finite dimensionality of $\mathcal{U}_{\Lambda _{L}}$, the norm and weak$%
^{\ast }$ topologies of $\mathcal{U}_{\Lambda _{L}}^{\ast }$ are the same
and, by the weak$^{\ast }$ continuity property of $\varpi _{s}\left(
t\right) $, we infer from (\ref{sdfsdf}) and (\ref{gronwal})-(\ref{gronwal2}%
) that 
\begin{equation}
(\varpi _{\rho ,s}\left( t\right) |_{\mathcal{U}_{\Lambda }})_{\rho \in
E}\in C\left( E;\mathcal{U}_{\Lambda }^{\ast }\right) \ ,\qquad s,t\in 
\mathbb{R}\ .  \label{norm continuity}
\end{equation}%
The continuity is even uniform for times $t$ in compact sets. Now, by\ Lemma %
\ref{Solution self}, for any $s,t\in \mathbb{R}$, $\rho _{1},\rho _{2}\in E$
and $A\in \mathcal{U}$, one obviously gets from the triangle inequality that%
\begin{eqnarray}
&&\left\vert \left( \varpi _{\rho _{1},s}\left( t\right) -\varpi _{\rho
_{2},s}\left( t\right) \right) \left( A\right) \right\vert
\label{sdfsfsdfdffffffffffffffffff} \\
&\leq &\left\vert \left( \rho _{1}-\rho _{2}\right) \circ \tau _{t,s}^{\Psi
^{\varpi _{\rho _{1},s}|_{\mathcal{U}_{\Lambda }}}}\left( A\right)
\right\vert +\left\Vert \left( \tau _{t,s}^{\Psi ^{\varpi _{\rho _{1},s}|_{%
\mathcal{U}_{\Lambda }}}}-\tau _{t,s}^{\Psi ^{\varpi _{\rho _{2},s}|_{%
\mathcal{U}_{\Lambda }}}}\right) \left( A\right) \right\Vert _{\mathcal{U}}\
.  \notag
\end{eqnarray}%
Combined with Lemma \ref{estimate useful}, Lebesgue's dominated convergence
theorem and Equation (\ref{norm continuity}), this last inequality leads to 
\begin{equation*}
(\varpi _{\rho ,s}\left( t\right) \left( A\right) )_{\rho \in E}\in C\left(
E;\mathbb{C}\right) \ ,\qquad A\in \mathcal{U}_{0}\ .
\end{equation*}%
By density of $\mathcal{U}_{0}\subseteq \mathcal{U}$ and the fact that any
state $\rho \in E$ satisfies $\left\Vert \rho \right\Vert _{\mathcal{U}%
^{\ast }}=1$, the assertion follows.
\end{proof}

\begin{corollary}[Solution to the self-consistency equation as
self-homeomorphisms]
\label{Corollary bije+cocyl}\mbox{ }\newline
At fixed $s,t\in \mathbb{R}$, $\varpi _{s}\left( t\right) \equiv (\varpi
_{\rho ,s}\left( t\right) )_{\rho \in E}\in \mathrm{Aut}\left( E\right) $,
i.e., $\varpi _{s}\left( t\right) $ is an automorphism of the state space $E$%
. Moreover, it satisfies a cocycle property:%
\begin{equation}
\forall s,r,t\in \mathbb{R}:\qquad \varpi _{s}\left( t\right) =\varpi
_{r}\left( t\right) \circ \varpi _{s}\left( r\right) \ .
\label{cocycle property}
\end{equation}
\end{corollary}

\begin{proof}
The proof is the same as the one of \cite[Corollary 7.7]{Bru-pedra-MF-I}: By
Corollary \ref{bijectivity}, for any $s,t\in \mathbb{R}$, $\varpi _{s}\left(
t\right) $ is a weak$^{\ast }$-continuous bijective mapping from $E$ to
itself. Since $E$ is weak$^{\ast }$-compact, its inverse is also weak$^{\ast
}$-continuous. The cocycle property is a rewriting of the equality $\varpi
_{\rho ,s}(t)=\varpi _{\varpi _{\rho ,s}(r),r}(t)$ of Lemma \ref{Solution
self}.
\end{proof}

Before stating the next lemma, recall that the topology used in the subspace 
$\mathrm{Aut}\left( E\right) \varsubsetneq C\left( E;E\right) $ of all
automorphisms of $E$ is the one of uniform convergence of weak$^{\ast }$%
-continuous functions, as stated in Equation (\ref{uniform convergence weak*}%
).

\begin{lemma}[Well-posedness of the self-consistency equation]
\label{lemma well}\mbox{ }\newline
For any $s\in \mathbb{R}$,%
\begin{equation*}
\mathbf{\varpi }_{s}^{\mathfrak{m}}\equiv (\varpi _{s}\left( t\right)
)_{t\in \mathbb{R}}\equiv ((\varpi _{\rho ,s}\left( t\right) )_{\rho \in
E})_{t\in \mathbb{R}}\in C\left( \mathbb{R};\mathrm{Aut}\left( E\right)
\right) \ .
\end{equation*}%
At fixed $s\in \mathbb{R}$, the mapping $t\mapsto \varpi _{s}(t)$ is
uniformly continuous for times $t$ in compact sets, i.e., for any $%
T,\varepsilon \in \mathbb{R}^{+}$ and $A\in \mathcal{U}$ there is $\eta \in 
\mathbb{R}^{+}$ such that, for all $t_{1},t_{2}\in \lbrack -T,T]$ with $%
\left\vert t_{1}-t_{2}\right\vert \leq \eta $,%
\begin{equation*}
\max_{\rho \in E}\left\vert \varpi _{s}\left( t_{1}\right) \left( A\right)
-\varpi _{s}\left( t_{2}\right) \left( A\right) \right\vert \leq \varepsilon
\ .
\end{equation*}
\end{lemma}

\begin{proof}
The proof is not exactly the same as the one of \cite[Lemma 7.8]%
{Bru-pedra-MF-I}, but it is similar: Take any sequence $(t_{n})_{n\in 
\mathbb{N}}\subseteq \mathbb{R}$ converging to $t\in \mathbb{R}$. Assume
that $\varpi _{s}\left( t_{n}\right) $ does not converge to $\varpi
_{s}\left( t\right) $, uniformly. So, by density of $\mathcal{U}_{0}$, there
are $(\rho _{n})_{n\in \mathbb{N}}\subseteq E$, $k\in \mathbb{N}$, $L\in 
\mathbb{N}$, $\varepsilon _{1},\ldots ,\varepsilon _{k}\in \mathbb{R}^{+}$
and $A_{1},\ldots ,A_{k}\in \mathcal{U}_{\Lambda _{L}}$ such that 
\begin{equation}
\liminf_{n\rightarrow \infty }\left\vert \left[ \varpi _{\rho _{n},s}\left(
t_{n}\right) -\varpi _{\rho _{n},s}\left( t\right) \right] \left(
A_{j}\right) \right\vert \geq \varepsilon _{j}>0\ ,\qquad j\in \left\{
1,\ldots ,k\right\} \ .  \label{machin2}
\end{equation}%
By weak$^{\ast }$-compactness and metrizability of $E$, we can assume
without loss of generality that the sequence $(\rho _{n})_{n\in \mathbb{N}}$
weak$^{\ast }$-converges to some $\rho \in E$. By Lemma \ref{Solution self}
and Equation (\ref{machin2}), this in turn implies that 
\begin{equation}
\liminf_{n\rightarrow \infty }\left\vert \left[ \rho _{n}\circ \tau
_{t_{n},s}^{\Psi ^{\varpi _{\rho _{n},s}|_{\mathcal{U}_{\Lambda }}}}-\varpi
_{\rho ,s}\left( t\right) \right] \left( A_{j}\right) \right\vert \geq
\varepsilon _{j}>0\ ,\qquad j\in \left\{ 1,\ldots ,k\right\} \ ,
\label{machin2bis}
\end{equation}%
since $\varpi _{s}\left( t\right) \in C\left( E;E\right) $. Using Lemma \ref%
{estimate useful}, Lebesgue's dominated convergence theorem and Equation (%
\ref{norm continuity}), we obtain from (\ref{machin2bis}) that 
\begin{equation}
\liminf_{n\rightarrow \infty }\left\vert \left[ \rho _{n}\circ \tau
_{t_{n},s}^{\Psi ^{\varpi _{\rho ,s}|_{\mathcal{U}_{\Lambda }}}}-\varpi
_{\rho ,s}\left( t\right) \right] \left( A_{j}\right) \right\vert \geq
\varepsilon _{j}>0\ ,\qquad j\in \left\{ 1,\ldots ,k\right\} \ .
\label{estimate}
\end{equation}%
But this is a contradiction because $(\tau _{t,s}^{\Psi })_{s,t\in {\mathbb{R%
}}}$ is a strongly continuous two-para%
\-%
meter family (Proposition \ref{Theorem Lieb-Robinson}) and hence, 
\begin{equation*}
\lim_{n\rightarrow \infty }\rho _{n}\circ \tau _{t_{n},s}^{\Psi ^{\varpi
_{\rho ,s}|_{\mathcal{U}_{\Lambda }}}}\left( A_{j}\right) =\rho \circ \tau
_{t,s}^{\Psi ^{\varpi _{\rho ,s}|_{\mathcal{U}_{\Lambda }}}}\left(
A_{j}\right) =\left[ \varpi _{\rho ,s}\left( t\right) \right] \left(
A_{j}\right)
\end{equation*}%
for all $j\in \left\{ 1,\ldots ,k\right\} $. By (\ref{sdfsdf}) and (\ref%
{gronwal})-(\ref{gronwal2}) taken for some arbitrarily large (but finite)
cubic box $\Lambda _{L}\supseteq \Lambda $, note that $(\varpi _{\rho
,s}\left( t\right) |_{\mathcal{U}_{\Lambda }})_{\rho \in E}\in C\left( E;%
\mathcal{U}_{\Lambda _{L}}^{\ast }\right) $ is uniformly continuous for
times $t$ in compact sets and, as $n\rightarrow \infty $, 
\begin{equation*}
(\tau _{t,s}^{\Psi ^{\varpi _{\rho _{n},s}|_{\mathcal{U}_{\Lambda }}}}-\tau
_{t,s}^{\Psi ^{\varpi _{\rho ,s}|_{\mathcal{U}_{\Lambda }}}})\left( A\right)
\ ,\qquad A\in \mathcal{U}_{0}\ ,
\end{equation*}%
converges to $0\in \mathcal{U}$ uniformly with respect to times $t$ in
compact sets, by Lemma \ref{estimate useful}. It means that the mapping $%
t\mapsto \varpi _{s}(t)$ is in fact uniformly continuous for times $t$ in
compact sets, at fixed $s\in \mathbb{R}$.
\end{proof}

\begin{lemma}[Joint continuity with respect to initial and final times]
\label{lemma well copy(1)}\mbox{ }\newline
The solution to the self-consistency equation is jointly continuous with
respect to initial and final times: 
\begin{equation*}
\mathbf{\varpi }^{\mathfrak{m}}\equiv (\mathbf{\varpi }_{s}^{\mathfrak{m}%
})_{s\in \mathbb{R}}\equiv (\varpi _{s}\left( t\right) )_{s,t\in \mathbb{R}%
}\equiv ((\varpi _{\rho ,s}\left( t\right) )_{\rho \in E})_{s,t\in \mathbb{R}%
}\in C\left( \mathbb{R}^{2};\mathrm{Aut}\left( E\right) \right) \ .
\end{equation*}
\end{lemma}

\begin{proof}
The proof is not the same as the one of \cite[Lemma 7.9]{Bru-pedra-MF-I},
but it is similar: Fix $\rho \in E$, $s\in \mathbb{R}$ and $T\in \mathbb{R}%
^{+}$ and remark that $C\left( [s-T,s+T]^{2};E_{\Lambda }\right) $ is a
closed bounded subset of the Banach space $C\left( [s-T,s+T]^{2};\mathcal{U}%
_{\Lambda }^{\ast }\right) $, where%
\begin{equation*}
\left\Vert \zeta \right\Vert _{C\left( [s-T,s+T]^{2};\mathcal{U}_{\Lambda
}^{\ast }\right) }\doteq \sup_{\alpha ,t\in \lbrack s-T,s+T]}\left\Vert
\zeta \left( \alpha ,t\right) \right\Vert _{\mathcal{U}_{\Lambda }^{\ast }}\
,\qquad \zeta \in C\left( [s-T,s+T]^{2};\mathcal{U}_{\Lambda }^{\ast
}\right) \ .
\end{equation*}%
Similar to (\ref{label}), we define the mapping $\mathfrak{F}$ from $C\left(
[s-T,s+T]^{2};E_{\Lambda }\right) $ to itself by%
\begin{equation*}
\mathfrak{F}\left( \zeta \right) \left( \alpha ,t\right) \doteq \rho \circ
\tau _{t,\alpha }^{\Psi ^{\zeta }}|_{\mathcal{U}_{\Lambda }}\ ,\qquad \alpha
,t\in \lbrack s-T,s+T]\ .
\end{equation*}%
See Proposition \ref{Theorem Lieb-Robinson} and Lemma \ref{definition
BCS-type model approximated copy(2)}. By Inequality (\ref{bound0002}), $%
\mathfrak{F}$ is also a contraction when the time $T\in \mathbb{R}^{+}$
satisfies (\ref{bound0003}). Hence, in this case, for any $\rho \in E$ and $%
s\in \mathbb{R}$, there is a unique 
\begin{equation}
\tilde{\gimel}\in C^{2}\left( [s-T,s+T]^{2};E_{\Lambda }\right)
\label{sdfkh0}
\end{equation}%
such that 
\begin{equation*}
\forall \alpha ,t\in \lbrack s-T,s+T]:\qquad \tilde{\gimel}\left( \alpha
,t\right) =\rho \circ \tau _{t,\alpha }^{\Psi ^{\tilde{\gimel}\left( \alpha
,\cdot \right) }}|_{\mathcal{U}_{\Lambda }}\ .
\end{equation*}%
Now, by uniqueness of the solution to (\ref{self consitence equation1}), 
\begin{equation}
\varpi _{\rho ,\alpha }\left( t\right) |_{\mathcal{U}_{\Lambda }}=\tilde{%
\gimel}\left( \alpha ,t\right) =\rho \circ \tau _{t,\alpha }^{\Psi ^{\tilde{%
\gimel}\left( \alpha ,\cdot \right) }}|_{\mathcal{U}_{\Lambda }}
\label{sdfkh}
\end{equation}%
for any $\alpha ,t\in \lbrack s-T,s+T]$. Similar to (\ref%
{sdfsfsdfdffffffffffffffffff}), for any $A\in \mathcal{U}$ and $\alpha
_{1},\alpha _{2},t_{1},t_{2}\in \lbrack s-T,s+T]^{2}$, note that 
\begin{multline*}
\left\vert \left( \varpi _{\rho ,\alpha _{1}}\left( t_{1}\right) -\varpi
_{\rho ,\alpha _{2}}\left( t_{2}\right) \right) \left( A\right) \right\vert
\leq \left\Vert \left( \tau _{t_{1},\alpha _{1}}^{\Psi ^{\tilde{\gimel}%
\left( \alpha _{1},\cdot \right) }}-\tau _{t_{2},\alpha _{1}}^{\Psi ^{\tilde{%
\gimel}\left( \alpha _{1},\cdot \right) }}\right) \left( A\right)
\right\Vert _{\mathcal{U}} \\
+\left\Vert \left( \tau _{t_{2},\alpha _{1}}^{\Psi ^{\tilde{\gimel}\left(
\alpha _{1},\cdot \right) }}-\tau _{t_{2},\alpha _{2}}^{\Psi ^{\tilde{\gimel}%
\left( \alpha _{1},\cdot \right) }}\right) \left( A\right) \right\Vert _{%
\mathcal{U}}+\left\Vert \left( \tau _{t_{2},\alpha _{2}}^{\Psi ^{\tilde{%
\gimel}\left( \alpha _{1},\cdot \right) }}-\tau _{t_{2},\alpha _{2}}^{\Psi ^{%
\tilde{\gimel}\left( \alpha _{2},\cdot \right) }}\right) \left( A\right)
\right\Vert _{\mathcal{U}}\ .
\end{multline*}%
Using this elementary inequality together with Proposition \ref{Theorem
Lieb-Robinson copy(3)} (iv), Lemmata \ref{definition BCS-type model
approximated copy(2)}, \ref{estimate useful} and Equations (\ref{sdfkh0})-(%
\ref{sdfkh}), one gets that, for any local element $A\in \mathcal{U}_{0}$,
the mapping 
\begin{equation*}
\left( \alpha ,t\right) \mapsto \varpi _{\rho ,\alpha }\left( t\right)
\left( A\right)
\end{equation*}%
from $[s-T,s+T]^{2}$ to $\mathbb{C}$ is continuous. By density of $\mathcal{U%
}_{0}\varsubsetneq \mathcal{U}$, we deduce that 
\begin{equation*}
(\varpi _{\rho ,\alpha }\left( t\right) )_{\alpha ,t\in \mathbb{R}}\in
C\left( [s-T,s+T]^{2};E\right) \ ,\qquad \rho \in E,\ s\in \mathbb{R}\ ,
\end{equation*}%
the parameter $T\in \mathbb{R}^{+}$ satisfying (\ref{bound0003}). Via
Corollary \ref{Corollary bije+cocyl} and Lemma \ref{lemma well}, we then
deduce that%
\begin{equation}
(\varpi _{\rho ,s}\left( t\right) )_{s,t\in \mathbb{R}}\in C\left( \mathbb{R}%
^{2};E\right) \ ,\qquad \rho \in E\ .  \label{norm continuity s}
\end{equation}

To get the assertion, it only remains to reproduce the compactness argument
performed in the proof of Lemma \ref{lemma well}: Take two sequences $%
(s_{n})_{n\in \mathbb{N}},(t_{n})_{n\in \mathbb{N}}\subseteq \mathbb{R}$
converging to $s,t\in \mathbb{R}$, respectively. Assume that $\varpi
_{s_{n}}\left( t_{n}\right) $ does not converge to $\varpi _{s}\left(
t\right) $, uniformly, which corresponds to have Equation (\ref{machin2bis}%
), the term $\tau _{t_{n},s}^{\Psi ^{\varpi _{\rho _{n},s}|_{\mathcal{U}%
_{\Lambda }}}}$ being replaced with $\tau _{t_{n},s_{n}}^{\Psi ^{\varpi
_{\rho _{n},s_{n}}|_{\mathcal{U}_{\Lambda }}}}$. Thanks to the triangle
inequality and Lemma \ref{estimate useful}, for any $L\in \mathbb{N}$, $A\in 
\mathcal{U}_{\Lambda _{L}}$, sufficiently large $T\in \mathbb{R}^{+}$ and $%
n\in \mathbb{N}$ such that $s,t\in (-T,T)$ and $s_{n},t_{n}\in \lbrack -T,T]$%
, 
\begin{multline*}
\left\Vert \left( \tau _{t_{n},s_{n}}^{\Psi ^{\varpi _{\rho _{n},s_{n}}|_{%
\mathcal{U}_{\Lambda }}}}-\tau _{t_{n},s_{n}}^{\Psi ^{\varpi _{\rho ,s}|_{%
\mathcal{U}_{\Lambda }}}}\right) \left( A\right) \right\Vert _{\mathcal{U}%
}\leq 2\left\vert \Lambda _{L}\right\vert \left\Vert A\right\Vert _{\mathcal{%
U}}\left\Vert \mathbf{F}\right\Vert _{1,\mathfrak{L}}\left\Vert \mathfrak{m}%
\right\Vert _{\infty }\mathrm{e}^{4\mathbf{D}T\left\Vert \mathfrak{m}%
\right\Vert _{\infty }} \\
\times \int_{-T}^{T}\left( \left\Vert \left( \varpi _{\rho _{n},s_{n}}\left(
\alpha \right) -\varpi _{\rho ,s_{n}}\left( \alpha \right) \right) |_{%
\mathcal{U}_{\Lambda }}\right\Vert _{\mathcal{U}_{\Lambda }^{\ast
}}+\left\Vert \left( \varpi _{\rho ,s_{n}}\left( \alpha \right) -\varpi
_{\rho ,s}\left( \alpha \right) \right) |_{\mathcal{U}_{\Lambda
}}\right\Vert _{\mathcal{U}_{\Lambda }^{\ast }}\right) \mathrm{d}\alpha \ .
\end{multline*}%
As a consequence, by using (\ref{gronwal})-(\ref{gronwal2}) and (\ref{norm
continuity s}) together with Lebesgue's dominated convergence theorem, one
arrives from (\ref{machin2bis}) at Equation (\ref{estimate}), $\tau
_{t_{n},s}^{\Psi ^{\varpi _{\rho ,s}|_{\mathcal{U}_{\Lambda }}}}$ being
replaced with $\tau _{t_{n},s_{n}}^{\Psi ^{\varpi _{\rho ,s}|_{\mathcal{U}%
_{\Lambda }}}}$. This is not possible because $(\tau _{t,s}^{\Psi })_{s,t\in 
{\mathbb{R}}}$ is a strongly continuous two-para%
\-%
meter family, by Proposition \ref{Theorem Lieb-Robinson}.
\end{proof}

The proof of Theorem \ref{theorem sdfkjsdklfjsdklfj}, being a consequence of
Lemmata \ref{Solution self} and \ref{lemma well copy(1)}, is finished. In
order to prove Theorem \ref{classical dynamics I}, we give now several
technical assertions. Concerning the first one, recall Definition \ref%
{convex Frechet derivative}: For any $f\in \mathfrak{C}$ and $\rho \in E$, $%
\mathrm{d}f\left( \rho \right) :E\rightarrow \mathbb{C}$ is the (unique)
convex weak$^{\ast }$-continuous G\^{a}teaux derivative of $f$ at $\rho \in
E $ if $\mathrm{d}f\left( \rho \right) \in \mathcal{A}(E;\mathbb{C})$ and%
\begin{equation*}
\lim_{\lambda \rightarrow 0^{+}}\lambda ^{-1}\left( f\left( \left( 1-\lambda
\right) \rho +\lambda \upsilon \right) -f\left( \rho \right) \right) =\left[ 
\mathrm{d}f\left( \rho \right) \right] \left( \upsilon \right) \ ,\qquad
\upsilon \in E\ .
\end{equation*}%
By Equation (\ref{clear2}) extended to the complex case, for any $f\in 
\mathfrak{Y}$, there is a unique $\mathrm{D}f\in C(E;\mathcal{U})$ such that%
\begin{equation*}
\mathrm{d}f\left( \rho \right) \left( \upsilon \right) =\widehat{\mathrm{D}%
f\left( \rho \right) }\left( \upsilon \right) =\upsilon \left( \mathrm{D}%
f\left( \rho \right) \right) \ ,\qquad \rho ,\upsilon \in E\ .
\end{equation*}%
In the next lemma, we compute the convex weak$^{\ast }$-continuous G\^{a}%
teaux derivative of the classical evolution (\ref{classical evolution family}%
) of elementary functions defined by (\ref{fA}) for local elements of $%
\mathcal{U}$:

\begin{lemma}[Differentiability of the solution -- II]
\label{lemma contnuity copy(1)}\mbox{ }\newline
For any $s,t\in \mathbb{R}$ and $A\in \mathcal{U}_{0}$, 
\begin{equation}
(\varpi _{\rho ,s}\left( t\right) (A))_{\rho \in E}\equiv (\varpi _{\rho
,s}\left( t,A\right) )_{\rho \in E}\in C^{1}\left( E;\mathbb{C}\right)
\label{eqsdfkljsdklfj}
\end{equation}%
and, for any $\upsilon \in E$, 
\begin{equation*}
\left[ \mathrm{d}\varpi _{\rho ,s}\left( t,A\right) \right] \left( \upsilon
\right) =\upsilon \left( \mathrm{D}\varpi _{\rho ,s}\left( t,A\right)
\right) =\left( \upsilon -\rho \right) \circ \tau _{t,s}^{\Psi ^{\gimel
}}\left( A\right) +\mathfrak{\maltese }_{A}\left[ \mathrm{d}\varpi _{\rho
,s}\left( \cdot ,\cdot \right) \left( \upsilon \right) \right] \ ,
\end{equation*}%
where $\gimel \equiv \gimel _{\rho ,s}\doteq \varpi _{\rho ,s}|_{\mathcal{U}%
_{\Lambda }}$ and, for any continuous function $\xi :\mathbb{R}\times 
\mathcal{U}_{\Lambda }\rightarrow \mathbb{C}$,%
\begin{align}
\mathfrak{\maltese }_{A}\left[ \xi \right] & \doteq \sum_{n\in \mathbb{N}%
}\sum\limits_{\mathcal{Z}\in \mathcal{P}_{f}}\int_{s}^{t}\mathrm{d}\alpha
\int_{(\mathbb{S}\cap \mathcal{W}_{\Lambda })^{n}}\mathfrak{a}\left( \alpha
\right) _{n}\left( \mathrm{d}\Psi ^{(1)},\ldots ,\mathrm{d}\Psi ^{(n)}\right)
\notag \\
& \sum_{m_{1},m_{2}=1,m_{2}\neq m_{1}}^{n}\xi (\alpha ,\mathfrak{e}_{\Psi
^{(m_{2})},\vec{\ell}})\rho \left( i\left[ \tau _{\alpha ,s}^{\Psi ^{\gimel
}}\left( \Psi _{\mathcal{Z}}^{(m_{1})}\right) ,\tau _{t,s}^{\Psi ^{\gimel
}}\left( A\right) \right] \right)  \notag \\
& \prod\limits_{j\in \left\{ 1,\ldots ,n\right\} \backslash
\{m_{1},m_{2}\}}\varpi _{\rho ,s}\left( \alpha ,\mathfrak{e}_{\Psi ^{(j)},%
\vec{\ell}}\right) \ .  \label{def dysons}
\end{align}%
The above series is absolutely summable. Moreover, for any $s\in \mathbb{R}$%
, $\rho \in E$ and $\tilde{\Lambda}\in \mathcal{P}_{f}$, the mapping $%
(t,A)\mapsto \mathrm{D}\varpi _{\rho ,s}\left( t,A\right) $ from $\mathbb{R}%
\times \mathcal{U}_{\tilde{\Lambda}}$ to $\mathcal{U}$ is continuous.
\end{lemma}

\begin{proof}
We start with a preliminary observation: By Proposition \ref{Theorem
Lieb-Robinson} and (\ref{cauchy1}) together with Lieb-Robinson bounds for
multi-commutators (cf. Proposition \ref{Theorem Lieb-Robinson copy(3)} (i)
and \cite[Theorems 4.11, 5.4]{brupedraLR}) and Lebesgue's dominated
convergence theorem, observe that 
\begin{equation}
\left( \tau _{t,s}^{\Psi _{1}}-\tau _{t,s}^{\Psi _{2}}\right) \left(
A\right) =\sum\limits_{\mathcal{Z}\in \mathcal{P}_{f}}\int_{s}^{t}\mathrm{d}%
\alpha \tau _{\alpha ,s}^{\Psi _{1}}\left( i\left[ \left( \Psi _{1}\left(
\alpha \right) -\Psi _{2}\left( \alpha \right) \right) _{\mathcal{Z}},\tau
_{t,\alpha }^{\Psi _{2}}\left( A\right) \right] \right)
\label{equality rivial}
\end{equation}%
for any $s,t\in \mathbb{R}$, $A\in \mathcal{U}_{0}$\ and $\Psi _{1},\Psi
_{2}\in C(\mathbb{R};\mathcal{W}^{\mathbb{R}}\cap \mathcal{W}_{\Lambda })$.
(Note that, in this case, $\Psi _{1},\Psi _{2}$ are finite-range
interactions.) The above series is absolutely summable, because of
Lieb-Robinson bounds, as stated in Proposition \ref{Theorem Lieb-Robinson
copy(3)} (i). In order to arrive at (\ref{equality rivial}), we also use the
equality 
\begin{equation*}
\left( \tau _{t,s}^{\Psi _{1}}-\tau _{t,s}^{\Psi _{2}}\right) |_{\mathcal{U}%
_{\tilde{\Lambda}}}=\lim_{L\rightarrow \infty }\int_{s}^{t}\tau _{\alpha
,s}^{\Psi _{1}}\circ \left( \delta ^{\Psi _{1}\left( \alpha \right) }-\delta
_{L}^{\Psi _{2}\left( \alpha \right) }\right) \circ \tau _{t,\alpha }^{(\Psi
_{2},L)}|_{\mathcal{U}_{\tilde{\Lambda}}}\mathrm{d}\alpha
\end{equation*}%
for any $\tilde{\Lambda}\in \mathcal{P}_{f}$, $s,t\in \mathbb{R}$ and $\Psi
_{1},\Psi _{2}\in C(\mathbb{R};\mathcal{W}^{\mathbb{R}})$, deduced from (\ref%
{cauchy1}) and Proposition \ref{Theorem Lieb-Robinson}. See also Definitions %
\ref{definition fininte vol dynam0}, \ref{dynamic series} and Corollary \ref%
{Lemma cigare1}. In particular, to get (\ref{equality rivial}), in the light
of Equations (\ref{assertion bisbisbisbis0100})-(\ref{assertion
bisbisbisbis01}), we have to estimate multi-commutators of order \emph{three}%
. This is done by the extension to multicommutators of the Lieb-Robinson
bounds, contributed in \cite[Theorems 4.11, 5.4]{brupedraLR}. Lieb-Robinson
bounds for multi-commutators of order three stated in \cite[Theorems 4.11,
5.4]{brupedraLR} require sufficient polynomial decays of interactions. An
obvious sufficient condition for such decays is to take $\Psi _{1}\left(
t\right) ,\Psi _{2}\left( t\right) \in \mathcal{W}_{\Lambda }$ for all $t\in 
\mathbb{R}$, meaning that they are all finite-range. See (\ref%
{eq:enpersitebis}).

Fix now all parameters of the lemma. For any $s,t\in \mathbb{R}$, $\rho
,\upsilon \in E$, $h\in (0,1]$ and $A\in \mathcal{U}_{0}$,%
\begin{eqnarray*}
\beth \left( h,t,A;\upsilon \right) &\doteq &h^{-1}\left( \varpi _{\left(
1-h\right) \rho +h\upsilon ,s}\left( t,A\right) -\varpi _{\rho ,s}\left(
t,A\right) \right) \\
&=&\left( \upsilon -\rho \right) \circ \tau _{t,s}^{\Psi ^{\gimel _{\rho
,s}}}\left( A\right) +h^{-1}\left( \left( 1-h\right) \rho +h\upsilon \right)
\circ \left( \tau _{t,s}^{\Psi ^{\gimel _{\left( 1-h\right) \rho +h\upsilon
,s}}}-\tau _{t,s}^{\Psi ^{\gimel _{\rho ,s}}}\right) \left( A\right) \ .
\end{eqnarray*}%
Using now (\ref{def aussi utile}), (\ref{def utile}) and (\ref{equality
rivial}) we deduce that, for any $\upsilon \in E$ and $A\in \mathcal{U}_{0}$%
, 
\begin{eqnarray}
\beth \left( h,t,A;\upsilon \right) &=&\left( \upsilon -\rho \right) \circ
\tau _{t,s}^{\Psi ^{\gimel _{\rho ,s}}}\left( A\right)
\label{equality nontrivial} \\
&&+\sum_{n\in \mathbb{N}}\sum\limits_{\mathcal{Z}\in \mathcal{P}%
_{f}}\int_{s}^{t}\mathrm{d}\alpha \int_{(\mathbb{S}\cap \mathcal{W}_{\Lambda
})^{n}}\mathfrak{a}\left( \alpha \right) _{n}\left( \mathrm{d}\Psi
^{(1)},\ldots ,\mathrm{d}\Psi ^{(n)}\right)  \notag \\
&&\qquad \sum_{m_{1},m_{2}=1,m_{2}\neq m_{1}}^{n}\beth \left( h,\alpha ,%
\mathfrak{e}_{\Psi ^{(m_{2})},\vec{\ell}};\upsilon \right)  \notag \\
&&\qquad \qquad \times \left( \left( 1-h\right) \rho +h\upsilon \right)
\circ \tau _{\alpha ,s}^{\Psi ^{\gimel _{_{\left( 1-h\right) \rho +h\upsilon
,s}}}}\left( i\left[ \Psi _{\mathcal{Z}}^{(m_{1})},\tau _{t,\alpha }^{\Psi
^{\gimel _{\rho ,s}}}\left( A\right) \right] \right)  \notag \\
&&\qquad \qquad \qquad \times \prod\limits_{j\in \left\{ 1,\ldots
,m_{2}-1\right\} \backslash \{m_{1}\}}\varpi _{\rho ,s}\left( \alpha ,%
\mathfrak{e}_{\Psi ^{(j)},\vec{\ell}}\right)  \notag \\
&&\qquad \qquad \qquad \qquad \times \prod\limits_{j\in \left\{
m_{2}+1,\ldots ,n\right\} \backslash \{m_{1}\}}\varpi _{\left( 1-h\right)
\rho +h\upsilon ,s}\left( \alpha ,\mathfrak{e}_{\Psi ^{(j)},\vec{\ell}%
}\right) \ ,  \notag
\end{eqnarray}%
where the two products over $j$ are, by definition, equal to $1$ when $j$
ranges over the empty set\footnote{%
This happens when $m_{2}\in \{1,n\}$.}. Note that (\ref{equality rivial})
can be used here because $\mathfrak{m}\in C_{b}(\mathbb{R};\mathcal{M}%
_{\Lambda })$, implying that $\Psi ^{\gimel }\in C_{b}(\mathbb{R};\mathcal{W}%
^{\mathbb{R}}\cap \mathcal{W}_{\Lambda })$. See Lemma \ref{definition
BCS-type model approximated copy(2)} and Equation (\ref{equality cigare00}).

From Equation (\ref{equality nontrivial}), one sees that $\beth \left(
\lambda ,t,A;\upsilon \right) $ is given by a Dyson-type series which is
absolutely summable, uniformly with respect to $h\in (0,1]$. To show that,
use $\mathfrak{m}\in C_{b}(\mathbb{R};\mathcal{M}_{\Lambda })$ (see (\ref%
{definition 0bis}) and (\ref{S00bis})-(\ref{S0bis})), Lemma \ref{definition
BCS-type model approximated copy(2)} together with Equation (\ref{equality
cigare00}), the fact that $(\tau _{t,s}^{\Psi ^{\xi }})_{s,t\in \mathbb{R}}$
is a family of $\ast $-automorphisms of $\mathcal{U}$ for any $\xi \in
C\left( \mathbb{R};E\right) $, Inequality (\ref{e phi}) and the (usual)
Lieb-Robinson bounds (Proposition \ref{Theorem Lieb-Robinson copy(3)} (i)).
Note that, for any $\tilde{\Lambda}\in \mathcal{P}_{f}$, the mapping $%
\mathbb{R}\times \mathcal{U}_{\tilde{\Lambda}}$ to $\mathbb{C}$ defined by 
\begin{equation*}
\left( t,A\right) \mapsto \left( \upsilon -\rho \right) \circ \tau
_{t,s}^{\Psi ^{\gimel _{\rho ,s}}}\left( A\right)
\end{equation*}%
is continuous. By Lemma \ref{estimate useful} and Lebesgue's dominated
convergence theorem,%
\begin{equation*}
\lim_{h\rightarrow 0^{+}}\left( \left( 1-h\right) \rho +h\upsilon \right)
\circ \tau _{\alpha ,s}^{\Psi ^{\gimel _{_{\left( 1-h\right) \rho +h\upsilon
,s}}}}\left( i\left[ \Psi _{\mathcal{Z}}^{(m_{1})},\tau _{t,\alpha }^{\Psi
^{\gimel _{\rho ,s}}}\left( A\right) \right] \right) =\rho \circ \tau
_{\alpha ,s}^{\Psi ^{\gimel _{_{\rho ,s}}}}\left( i\left[ \Psi _{\mathcal{Z}%
}^{(m_{1})},\tau _{t,\alpha }^{\Psi ^{\gimel _{\rho ,s}}}\left( A\right) %
\right] \right) \ ,
\end{equation*}%
uniformly for $\alpha $ in a compact set and, by Equations (\ref{sdfsdf})
and (\ref{gronwal})-(\ref{gronwal2}),%
\begin{equation*}
\lim_{h\rightarrow 0^{+}}\max_{\alpha \in \left[ s-T,s+T\right] }\left\Vert
\left( \varpi _{\left( 1-h\right) \rho +h\upsilon ,s,s}\left( \alpha \right)
-\varpi _{\rho ,s}\left( \alpha \right) \right) |_{\mathcal{U}_{\Lambda
}}\right\Vert _{\mathcal{U}_{\Lambda }^{\ast }}=0\ ,\qquad T\in \mathbb{R}%
^{+}\ .
\end{equation*}%
Hence, using again Lebesgue's dominated convergence theorem, we deduce that 
\begin{equation}
\beth \left( 0,t,A;\upsilon \right) \doteq \lim_{h\rightarrow 0^{+}}\beth
\left( h,t,A;\upsilon \right) =\lim_{h\rightarrow 0^{+}}h^{-1}\left( \varpi
_{\left( 1-h\right) \rho +h\upsilon ,s}\left( t,A\right) -\varpi _{\rho
,s}\left( t,A\right) \right)  \label{lebesgue chat}
\end{equation}%
exists for all $s,t\in \mathbb{R}$, $\rho ,\upsilon \in E$ and $A\in 
\mathcal{U}_{0}$, as given by a Dyson-type series. In particular, for any $%
\upsilon \in E$ and $\tilde{\Lambda}\in \mathcal{P}_{f}$, the complex-valued
function $(t,A)\mapsto \beth \left( 0,t,A,\upsilon \right) $ on $\mathbb{R}%
\times \mathcal{U}_{\tilde{\Lambda}}$ is the unique solution in $\xi \in
C\left( \mathbb{R}\times \mathcal{U}_{\tilde{\Lambda}};\mathbb{C}\right) $
to the equation 
\begin{equation}
\xi \left( t,A\right) =\left( \upsilon -\rho \right) \circ \tau _{t,s}^{\Psi
^{\gimel }}\left( A\right) +\mathfrak{\maltese }_{A}\left[ \xi \right]
\label{integral equation 0}
\end{equation}%
with $\mathfrak{\maltese }_{A}$ defined by (\ref{def dysons}). Compare with (%
\ref{equality nontrivial}) taken at $h=0$.

The uniqueness of the solution in $\xi \in C\left( \mathbb{R}\times \mathcal{%
U}_{\tilde{\Lambda}};\mathbb{C}\right) $ to (\ref{integral equation 0}) is a
consequence of the fact that one can iterate Equation (\ref{integral
equation 0}) in order to prove that $\xi $ is the Dyson-type series $\beth
\left( 0,t,A,\upsilon \right) $: First, it suffices to show this fact for $%
\tilde{\Lambda}=\Lambda $ and $A\in \mathcal{U}_{\Lambda }$ with $\left\Vert
A\right\Vert _{\mathcal{U}}\leq \left\Vert \mathbf{F}\right\Vert _{1,%
\mathfrak{L}}$ because this case fixes all $\xi (t,\mathfrak{e}_{\Psi ,\vec{%
\ell}})$ for all $\Psi \in \mathbb{S}\cap \mathcal{W}_{\Lambda }$ and $t\in 
\mathbb{R}$\ (cf. (\ref{e phi})). Secondly, observe that any function $\xi
\in C\left( \mathbb{R}\times \mathcal{U}_{\Lambda };\mathbb{C}\right) $ is
bounded on compact sets of $\mathbb{R}\times \mathcal{U}_{\Lambda }$ and any
norm-closed ball of $\mathcal{U}_{\Lambda }$ is compact, by finite
dimensionality of $\mathcal{U}_{\Lambda }$. Using these observations
together with (\ref{(3.1) NS})-(\ref{(3.2) NS}), (\ref{e phi}), (\ref%
{definition 0bis}), Lieb-Robinson bounds (Proposition \ref{Theorem
Lieb-Robinson copy(3)} (i)), Lemma \ref{definition BCS-type model
approximated copy(2)} and tedious computations, one checks that, for any $%
A\in \mathcal{U}_{\Lambda }$ satisfying $\left\Vert A\right\Vert _{\mathcal{U%
}}\leq \left\Vert \mathbf{F}\right\Vert _{1,\mathfrak{L}}$, arbitrary time $%
T\in \mathbb{R}^{+}$ and $t\in \left[ s-T,s+T\right] $, 
\begin{equation}
\left\vert \mathfrak{\maltese }_{A}\left[ \xi \right] \right\vert \leq
2\left\Vert \mathbf{F}\right\Vert _{1,\mathfrak{L}}\mathrm{e}^{4\mathbf{D}%
T\left\Vert \mathfrak{m}\right\Vert _{\infty }}\left\vert \Lambda
\right\vert \left\Vert \mathfrak{m}\right\Vert _{\infty }\sup_{t\in \left[
s-T,s+T\right] }\sup_{B\in \mathcal{U}_{\Lambda }:\left\Vert B\right\Vert _{%
\mathcal{U}}\leq \left\Vert \mathbf{F}\right\Vert _{1,\mathfrak{L}%
}}\left\vert \xi \left( t,B\right) \right\vert \left( \int_{s}^{t}\mathrm{d}%
\alpha \right) \ .  \label{dyson chiant}
\end{equation}%
More generally, we can reconstruct from (\ref{integral equation 0}) the $k$%
th first terms of the Dyson-type series $\beth \left( 0,t,A,\upsilon \right) 
$ and, in the same way one obtains (\ref{dyson chiant}) (i.e., the $\left(
0+1\right) $th remaining term), the $\left( k+1\right) $th remaining term is
bounded by 
\begin{equation*}
\sup_{t\in \left[ s-T,s+T\right] }\sup_{B\in \mathcal{U}_{\Lambda
}:\left\Vert B\right\Vert _{\mathcal{U}}\leq \left\Vert \mathbf{F}%
\right\Vert _{1,\mathfrak{L}}}\left\vert \xi \left( t,B\right) \right\vert
\left( 2\left\Vert \mathbf{F}\right\Vert _{1,\mathfrak{L}}\mathrm{e}^{4%
\mathbf{D}T\left\Vert \mathfrak{m}\right\Vert _{\infty }}\left\vert \Lambda
\right\vert \left\Vert \mathfrak{m}\right\Vert _{\infty }\right)
^{k+1}\int_{s}^{t}\mathrm{d}\alpha _{1}\cdots \int_{s}^{\alpha _{k-1}}%
\mathrm{d}\alpha _{k}\ ,
\end{equation*}%
and thus vanishes in the limit $k\rightarrow \infty $. So, by Lebesgue's
dominated convergence theorem, any solution $\xi \in C\left( \mathbb{R}%
\times \mathcal{U}_{\tilde{\Lambda}};\mathbb{C}\right) $ to (\ref{integral
equation 0}) is equal to $\beth \left( 0,t,A,\upsilon \right) $.

Similar to (\ref{equality nontrivial}), using in particular the
Lieb-Robinson bounds (Proposition \ref{Theorem Lieb-Robinson copy(3)} (i)),
note that the integral equation 
\begin{eqnarray}
\mathfrak{D}\left( t,A\right) &=&\tau _{t,s}^{\Psi ^{\gimel }}\left(
A\right) -\rho \circ \tau _{t,s}^{\Psi ^{\gimel }}\left( A\right) \mathfrak{1%
}  \label{chat0} \\
&&+\sum_{n\in \mathbb{N}}\sum\limits_{\mathcal{Z}\in \mathcal{P}%
_{f}}\int_{s}^{t}\mathrm{d}\alpha \int_{(\mathbb{S}\cap \mathcal{W}_{\Lambda
})^{n}}\mathfrak{a}\left( \alpha \right) _{n}\left( \mathrm{d}\Psi
^{(1)},\ldots ,\mathrm{d}\Psi ^{(n)}\right)  \notag \\
&&\qquad \sum_{m_{1},m_{2}=1,m_{2}\neq m_{1}}^{n}\mathfrak{D}\left( \alpha ,%
\mathfrak{e}_{\Psi ^{(m_{2})},\vec{\ell}}\right)  \notag \\
&&\qquad \qquad \times \rho \circ \tau _{\alpha ,s}^{\Psi ^{\gimel }}\left(
i \left[ \Psi _{\mathcal{Z}}^{(m_{1})},\tau _{t,\alpha }^{\Psi ^{\gimel
}}\left( A\right) \right] \right)  \notag \\
&&\qquad \qquad \qquad \qquad \qquad \times \prod\limits_{j\in \left\{
1,\ldots ,n\right\} \backslash \{m_{1},m_{2}\}}\varpi _{\rho ,s}\left(
\alpha ,\mathfrak{e}_{\Psi ^{(j)},\vec{\ell}}\right) \ .  \notag
\end{eqnarray}%
uniquely determines, by absolutely summable (in $\mathfrak{U}$) Dyson-type
series, a continuous mapping $(t,A)\mapsto \mathfrak{D}\left( t,A\right) $
from $\mathbb{R}\times \mathcal{U}_{\tilde{\Lambda}}$ to $\mathcal{U}$ for
any $\tilde{\Lambda}\in \mathcal{P}_{f}$, which, via (\ref{integral equation
0}), satisfies 
\begin{equation}
\upsilon \left( \mathfrak{D}\left( t,A\right) \right) =\beth \left(
0,t,A;\upsilon \right) \doteq \lim_{h\rightarrow 0^{+}}h^{-1}\left( \varpi
_{\left( 1-h\right) \rho +h\upsilon ,s}\left( t,A\right) -\varpi _{\rho
,s}\left( t,A\right) \right)  \label{chat1}
\end{equation}%
for all $s,t\in \mathbb{R}$, $\rho ,\upsilon \in E$ and $A\in \mathcal{U}%
_{0} $. Observe that the integrals in the corresponding Dyson-type series
are well-defined as Bochner integrals: For any $\tilde{\Lambda}\in \mathcal{P%
}_{f}$, the mapping from $\mathbb{R}\times \mathcal{U}_{\tilde{\Lambda}}$ to 
$\mathcal{U}$ defined by 
\begin{equation}
\left( t,A\right) \mapsto \tau _{t,s}^{\Psi ^{\gimel }}\left( A\right) -\rho
\circ \tau _{t,s}^{\Psi ^{\gimel }}\left( A\right) \mathfrak{1}
\label{mapping0}
\end{equation}%
is continuous. Since the measures $\mathfrak{a}\left( \alpha \right) _{n}$, $%
n\in \mathbb{N}$, are finite and $\mathcal{U}$ is a separable Banach space,
by \cite[Theorems 1.1 and 1.2]{pettis}, all terms appearing in the arguments
of the integrals are Bochner-integrable. By Definition \ref{convex Frechet
derivative}, the assertion follows.
\end{proof}

For the last assertions, we assume that the decay function equals 
\begin{equation}
\mathbf{F}\left( x,y\right) =\mathrm{e}^{-2\varsigma \left\vert
x-y\right\vert }(1+\left\vert x-y\right\vert )^{-(d+\epsilon )}\ ,\qquad
x,y\in \mathfrak{L}\ ,  \label{fix decay}
\end{equation}%
for some fixed $\varsigma ,\epsilon \in \mathbb{R}^{+}$. See (\ref{examples}%
).

\begin{lemma}[Graph norm continuity of dynamics on local elements]
\label{lemma extra1}\mbox{ }\newline
Assume (\ref{fix decay}). For any $\Psi \in C(\mathbb{R};\mathcal{W}^{%
\mathbb{R}})$, $\Phi \in \mathcal{W}$, $s,t\in \mathbb{R}$ and $A\in 
\mathcal{U}_{0}$, $\tau _{t,s}^{\Psi }\left( A\right) \in \mathrm{dom}%
(\delta ^{\Phi })$ and%
\begin{equation*}
\delta ^{\Phi }\circ \tau _{t,s}^{\Psi }\left( A\right) =\lim_{L\rightarrow
\infty }\delta ^{\Phi }\circ \tau _{t,s}^{(L,\Psi )}\left( A\right)
=\lim_{L\rightarrow \infty }\delta _{L}^{\Phi }\circ \tau _{t,s}^{\Psi
}\left( A\right) \ ,
\end{equation*}%
uniformly for $s,t$ on compacta. Additionally, for any $s\in \mathbb{R}$ and 
$\tilde{\Lambda}\in \mathcal{P}_{f}$, the mapping $(t,A)\mapsto \delta
^{\Phi }\circ \tau _{t,s}^{\Psi }\left( A\right) $ from $\mathbb{R}\times 
\mathcal{U}_{\tilde{\Lambda}}$ to $\mathcal{U}$ is continuous.
\end{lemma}

\begin{proof}
For any $\Psi \in C(\mathbb{R};\mathcal{W}^{\mathbb{R}})$, $\Phi \in 
\mathcal{W}$, $s,t\in \mathbb{R}$ and $A\in \mathcal{U}_{0}$, a simple
adaptation of \cite[Equation (5.47), $m=1$]{brupedraLR} (using Lieb-Robinson
bounds for multi-commutators of order three \cite[Theorems 4.11, 5.4]%
{brupedraLR}) implies that 
\begin{equation}
\lim_{L_{0}\rightarrow \infty }\sup_{L\in \mathbb{N}}\left\Vert \delta
_{L}^{\Phi }\circ \left( \tau _{t,s}^{\Psi }-\tau _{t,s}^{(L_{0},\Psi
)}\right) \left( A\right) \right\Vert _{\mathcal{U}}=0\ .  \label{hiant1}
\end{equation}%
Using again a similar argument together with Corollary \ref{Lemma cigare1}
and the closedness of $\delta ^{\Phi }$, we also deduce that%
\begin{equation}
\delta ^{\Phi }\circ \tau _{t,s}^{\Psi }\left( A\right) =\lim_{L\rightarrow
\infty }\delta ^{\Phi }\circ \tau _{t,s}^{(L,\Psi )}\left( A\right) \qquad 
\text{and}\qquad \tau _{t,s}^{\Psi }\left( A\right) \in \mathrm{dom}(\delta
^{\Phi })\ .  \label{hiant2}
\end{equation}%
Note that the limits in Equations (\ref{hiant1})-(\ref{hiant2}) are uniform
for $s,t$ on compacta (cf. \cite[Theorem 5.6 (i)]{brupedraLR}). For any $%
\Psi \in C(\mathbb{R};\mathcal{W}^{\mathbb{R}})$, $\Phi \in \mathcal{W}$, $%
L,L_{0}\in \mathbb{N}$, $s,t\in \mathbb{R}$ and $A\in \mathcal{U}_{0}$,
observe meanwhile that%
\begin{eqnarray*}
\left\Vert \delta ^{\Phi }\circ \tau _{t,s}^{\Psi }\left( A\right) -\delta
_{L}^{\Phi }\circ \tau _{t,s}^{\Psi }\left( A\right) \right\Vert _{\mathcal{U%
}} &\leq &\left\Vert \delta ^{\Phi }\circ \left( \tau _{t,s}^{\Psi }-\tau
_{t,s}^{(L_{0},\Psi )}\right) \left( A\right) \right\Vert _{\mathcal{U}} \\
&&+\left\Vert \delta _{L}^{\Phi }\circ \left( \tau _{t,s}^{\Psi }-\tau
_{t,s}^{(L_{0},\Psi )}\right) \left( A\right) \right\Vert _{\mathcal{U}} \\
&&+\left\Vert \left( \delta ^{\Phi }-\delta _{L}^{\Phi }\right) \circ \tau
_{t,s}^{(L_{0},\Psi )}\left( A\right) \right\Vert _{\mathcal{U}}\ .
\end{eqnarray*}%
By (\ref{hiant1})-(\ref{hiant2}) and Corollary \ref{Lemma cigare1}, it
follows that, for any $\Phi \in \mathcal{W}$, $s,t\in \mathbb{R}$ and $A\in 
\mathcal{U}_{0}$, 
\begin{equation*}
\delta ^{\Phi }\circ \tau _{t,s}^{\Psi }\left( A\right) =\lim_{L\rightarrow
\infty }\delta _{L}^{\Phi }\circ \tau _{t,s}^{\Psi }\left( A\right) \ ,
\end{equation*}%
uniformly for $s,t$ on compacta.

For any $L\in \mathbb{R}$, the mapping from $\mathbb{R}\times \mathcal{U}%
_{0} $ to $\mathcal{U}$ defined by 
\begin{equation}
\left( t,A\right) \mapsto \delta ^{\Phi }\circ \tau _{t,s}^{(L,\Psi )}\left(
A\right)  \label{mapping L}
\end{equation}%
is continuous, by Corollary \ref{Lemma cigare1}. By (\ref{hiant2}), as $%
L\rightarrow \infty $, this mapping on any fixed $A\in \mathcal{U}_{0}$
converges uniformly for $s,t$ on compacta to the mapping from $\mathbb{R}%
\times \mathcal{U}_{0}$ to $\mathcal{U}$ defined by 
\begin{equation}
\left( t,A\right) \mapsto \delta ^{\Phi }\circ \tau _{t,s}^{\Psi }\left(
A\right) \ ,  \label{mapping plus}
\end{equation}%
which is thus continuous with respect to $t\in \mathbb{R}$. By finite
dimensionality of $\mathcal{U}_{\tilde{\Lambda}}$ for $\tilde{\Lambda}\in 
\mathcal{P}_{f}$ and linearity of the mapping (\ref{mapping plus}) with
respect to $A\in \mathcal{U}_{0}$, the function (\ref{mapping plus}) is
continuous on $\mathbb{R}\times \mathcal{U}_{\tilde{\Lambda}}$.
\end{proof}

\begin{proposition}[Graph norm continuity of convex derivatives]
\label{lemma extra2}\mbox{ }\newline
Assume (\ref{fix decay}). For any $\Phi \in \mathcal{W}$, $s,t\in \mathbb{R}$%
, $\rho \in E$ and $A\in \mathcal{U}_{0}$, $\mathrm{D}\varpi _{\rho
,s}\left( t,A\right) \in \mathrm{dom}(\delta ^{\Phi })$ and%
\begin{equation*}
\delta ^{\Phi }\left( \mathrm{D}\varpi _{\rho ,s}\left( t,A\right) \right)
=\lim_{L\rightarrow \infty }\delta _{L}^{\Phi }\left( \mathrm{D}\varpi
_{\rho ,s}\left( t,A\right) \right) \ .
\end{equation*}%
Additionally, for any $s\in \mathbb{R}$, $\rho \in E$ and $\tilde{\Lambda}%
\in \mathcal{P}_{f}$, the mapping $(t,A)\mapsto \delta ^{\Phi }(\mathrm{D}%
\varpi _{\rho ,s}\left( t,A\right) )$ from $\mathbb{R}\times \mathcal{U}_{%
\tilde{\Lambda}}$ to $\mathcal{U}$ is continuous.
\end{proposition}

\begin{proof}
Fix all parameters of the proposition. Let%
\begin{equation*}
\mathcal{G\doteq }\left( \mathrm{dom}(\delta ^{\Phi }),\left\Vert \cdot
\right\Vert _{\mathcal{G}}\right)
\end{equation*}%
be the Banach space obtained by endowing the domain of $\delta ^{\Phi }$
with its graph norm. By \cite[Theorem 4.8 (ii)]{brupedraLR}, $\mathcal{U}%
_{0} $ is a core of the derivation $\delta ^{\Phi }$ and hence, $\mathcal{G}$
is a separable Banach space.

By using Lemma \ref{lemma extra1}, the closedness of $\delta ^{\Phi }$ and
the fact that all terms appearing in the arguments of the integrals in the
Dyson-type series of $\mathfrak{D}\left( t,A\right) $ are
Bochner-integrable, one checks that the Dyson-type series deduced from (\ref%
{chat0}) yields an element $\mathfrak{D}\left( t,A\right) \in \mathcal{G}$
with $(t,A)\mapsto \delta ^{\Phi }\left( \mathfrak{D}\left( t,A\right)
\right) $ on $\mathbb{R}\times \mathcal{U}_{\tilde{\Lambda}}$ ($\tilde{%
\Lambda}\in \mathcal{P}_{f}$) being the unique solution in $\mathfrak{E}\in
C\left( \mathbb{R}\times \mathcal{U}_{\tilde{\Lambda}};\mathcal{U}\right) $
to the equation%
\begin{eqnarray}
\mathfrak{E}\left( t,A\right) &=&\delta ^{\Phi }\circ \tau _{t,s}^{\Psi
^{\gimel }}\left( A\right)  \label{chat00} \\
&&+\sum_{n\in \mathbb{N}}\sum\limits_{\mathcal{Z}\in \mathcal{P}%
_{f}}\int_{s}^{t}\mathrm{d}\alpha \int_{(\mathbb{S}\cap \mathcal{W}_{\Lambda
})^{n}}\mathfrak{a}\left( \alpha \right) _{n}\left( \mathrm{d}\Psi
^{(1)},\ldots ,\mathrm{d}\Psi ^{(n)}\right)  \notag \\
&&\qquad \sum_{m_{1},m_{2}=1,m_{2}\neq m_{1}}^{n}\mathfrak{E}\left( \alpha ,%
\mathfrak{e}_{\Psi ^{(m_{2})},\vec{\ell}}\right)  \notag \\
&&\qquad \qquad \times \rho \circ \tau _{\alpha ,s}^{\Psi ^{\gimel }}\left(
i \left[ \Psi _{\mathcal{Z}}^{(m_{1})},\tau _{t,\alpha }^{\Psi ^{\gimel
}}\left( A\right) \right] \right)  \notag \\
&&\qquad \qquad \qquad \qquad \qquad \times \prod\limits_{j\in \left\{
1,\ldots ,n\right\} \backslash \{m_{1},m_{2}\}}\varpi _{\rho ,s}\left(
\alpha ,\mathfrak{e}_{\Psi ^{(j)},\vec{\ell}}\right)  \notag
\end{eqnarray}%
for $t\in \mathbb{R}$, $A\in \mathcal{U}_{\tilde{\Lambda}}$ and $\tilde{%
\Lambda}\in \mathcal{P}_{f}$. Compare this equation with (\ref{chat0}). The
integrals in (\ref{chat00}) are again well-defined as Bochner integrals:
Combine Lemma \ref{lemma extra1} and \cite[Theorems 1.1 and 1.2]{pettis}
with the fact that the measure $\mathfrak{a}\left( \alpha \right) _{n}$, $%
n\in \mathbb{N}$, are finite and $\mathcal{U}$ is separable.

Additionally, for any $s\in \mathbb{R}$, $\rho \in E$ and $A\in \mathcal{U}%
_{0}$, the mapping $t\mapsto \delta ^{\Phi }(\mathrm{D}\varpi _{\rho
,s}\left( t,A\right) )$ from $\mathbb{R}$ to $\mathcal{U}$ is continuous,
again by Lemma \ref{lemma extra1}, the last Dyson-type series and the
identity 
\begin{equation}
\mathrm{D}\varpi _{\rho ,s}\left( t,A\right) =\mathfrak{D}\left( t,A\right)
\ .  \label{chat0chat0}
\end{equation}%
By linearity\ with respect to $A\in \mathcal{U}_{0}$ and finite
dimensionality of $\mathcal{U}_{\tilde{\Lambda}}$ for $\tilde{\Lambda}\in 
\mathcal{P}_{f}$, it follows that the mapping $(t,A)\mapsto \delta ^{\Phi }(%
\mathrm{D}\varpi _{\rho ,s}\left( t,A\right) )$ from $\mathbb{R}\times 
\mathcal{U}_{\tilde{\Lambda}}$ to $\mathcal{U}$ is continuous.

By combining Lemma \ref{lemma extra1} with Lebesgue's dominated convergence
theorem and the fact that the unique solution $\mathfrak{E}=\delta ^{\Phi
}\circ \mathfrak{D}$ to (\ref{chat00}) is given by a Dyson-type series, we
arrive at 
\begin{equation}
\delta ^{\Phi }\left( \mathfrak{D}\left( t,A\right) \right)
=\lim_{L\rightarrow \infty }\delta _{L}^{\Phi }\left( \mathfrak{D}\left(
t,A\right) \right)  \label{limit important}
\end{equation}%
in $\mathcal{U}$. Note that one can interchange the local derivation $\delta
_{L}^{\Phi }$ (which is a commutator) in the right-hand side of (\ref{limit
important}) with every integral over $(\mathbb{S}\cap \mathcal{W}_{\Lambda
})^{n}$ for any $n\in \mathbb{N}$, because it is a bounded operator on $%
\mathcal{U}$. By (\ref{chat0chat0}), this concludes the proof of the
proposition.
\end{proof}

\begin{lemma}[Differentiability of the solution -- III]
\label{Differentiability copy(1)}\mbox{ }\newline
Assume (\ref{fix decay}). Then, for any $t\in \mathbb{R}$, $\rho \in E$ and $%
A\in \mathcal{U}_{0}$, 
\begin{equation*}
(\varpi _{\rho ,s}\left( t\right) (A))_{s\in \mathbb{R}}\equiv (\varpi
_{\rho ,s}\left( t,A\right) )_{s\in \mathbb{R}}\in C^{1}\left( \mathbb{R};%
\mathbb{C}\right)
\end{equation*}%
with derivative satisfying, for any $A\in \mathcal{U}_{0}$,%
\begin{equation}
\partial _{s}\varpi _{\rho ,s}\left( t,A\right) =-\rho \circ \delta ^{\Psi
^{\varpi _{\rho ,s}|_{\mathcal{U}_{\Lambda }}}\left( s\right) }\circ \tau
_{t,s}^{\Psi ^{\varpi _{\rho ,s}|_{\mathcal{U}_{\Lambda }}}}\left( A\right) +%
\mathfrak{\maltese }_{A}\left[ \partial _{s}\varpi _{\rho ,s}\right] \ .
\label{equation integral}
\end{equation}%
Here, $\mathfrak{\maltese }_{A}$ is defined by (\ref{def dysons}).
Additionally, for any $\tilde{\Lambda}\in \mathcal{P}_{f}$, $(t,A)\mapsto
\partial _{s}\varpi _{\rho ,s}(t,A)$ is a continuous function on $\mathbb{R}%
\times \mathcal{U}_{\tilde{\Lambda}}$.
\end{lemma}

\begin{proof}
Fix all parameters of the Lemma. By Lemma \ref{Solution self}, for any $\rho
\in E$, $s,t\in \mathbb{R}$, $A\in \mathcal{U}_{0}$ and $\varepsilon \in 
\mathbb{R}\backslash \{0\}$, 
\begin{eqnarray*}
\tilde{\beth}\left( \varepsilon ,t,A\right) &\doteq &\varepsilon ^{-1}\left(
\varpi _{\rho ,s+\varepsilon }\left( t,A\right) -\varpi _{\rho ,s}\left(
t,A\right) \right) \\
&=&\varepsilon ^{-1}\rho \circ (\tau _{t,s+\varepsilon }^{\Psi ^{\gimel
_{\rho ,s}}}-\tau _{t,s}^{\Psi ^{\gimel _{\rho ,s}}})\left( A\right)
+\varepsilon ^{-1}\rho \circ (\tau _{t,s+\varepsilon }^{\Psi ^{\gimel _{\rho
,s+\varepsilon }}}-\tau _{t,s+\varepsilon }^{\Psi ^{\gimel _{\rho
,s}}})\left( A\right)
\end{eqnarray*}%
with $\gimel \equiv \gimel _{\rho ,s}\doteq \varpi _{\rho ,s}|_{\mathcal{U}%
_{\Lambda }}$. Similar to (\ref{equality nontrivial}), via Equations (\ref%
{def aussi utile}), (\ref{def utile}) and (\ref{equality rivial}) (keeping
in mind that $\mathfrak{m}\in C_{b}(\mathbb{R};\mathcal{M}_{\Lambda })$), we
deduce that%
\begin{eqnarray}
\tilde{\beth}\left( \varepsilon ,t,A\right) &=&\varepsilon ^{-1}\rho \circ
(\tau _{t,s+\varepsilon }^{\Psi ^{\gimel _{\rho ,s}}}-\tau _{t,s}^{\Psi
^{\gimel _{\rho ,s}}})\left( A\right)  \label{sdkljsdlkjf} \\
&&+\sum_{n\in \mathbb{N}}\sum\limits_{\mathcal{Z}\in \mathcal{P}%
_{f}}\int_{s+\varepsilon }^{t}\mathrm{d}\alpha \int_{(\mathbb{S}\cap 
\mathcal{W}_{\Lambda })^{n}}\mathfrak{a}\left( \alpha \right) _{n}\left( 
\mathrm{d}\Psi ^{(1)},\ldots ,\mathrm{d}\Psi ^{(n)}\right)  \notag \\
&&\qquad \sum_{m_{1},m_{2}=1,m_{2}\neq m_{1}}^{n}\tilde{\beth}\left(
\varepsilon ,\alpha ,\mathfrak{e}_{\Psi ^{(m_{2})},\vec{\ell}}\right)  \notag
\\
&&\qquad \qquad \times \rho \circ \tau _{\alpha ,s}^{\Psi ^{\gimel _{_{\rho
,s+\varepsilon }}}}\left( i\left[ \Psi _{\mathcal{Z}}^{(m_{1})},\tau
_{t,\alpha }^{\Psi ^{\gimel _{\rho ,s}}}\left( A\right) \right] \right) 
\notag \\
&&\qquad \qquad \qquad \times \prod\limits_{j\in \left\{ 1,\ldots
,m_{2}-1\right\} \backslash \{m_{1}\}}\varpi _{\rho ,s}\left( \alpha ,%
\mathfrak{e}_{\Psi ^{(j)},\vec{\ell}}\right)  \notag \\
&&\qquad \qquad \qquad \qquad \times \prod\limits_{j\in \left\{
m_{2}+1,\ldots ,n\right\} \backslash \{m_{1}\}}\varpi _{\rho ,s+\varepsilon
}\left( \alpha ,\mathfrak{e}_{\Psi ^{(j)},\vec{\ell}}\right) \ ,  \notag
\end{eqnarray}%
where the two products over $j$ are, by definition, equal to $1$ when $j$
ranges over the empty set\footnote{%
This happens when $m_{2}\in \{1,n\}$.}. Again, using the same arguments as
for $\beth \left( h,t,A;\upsilon \right) $ in (\ref{equality nontrivial}),
one sees from Equation (\ref{sdkljsdlkjf}) that $\tilde{\beth}\left(
\varepsilon ,t,A\right) $ is given by a Dyson-type series which is
absolutely summable, uniformly with respect to $\varepsilon $ in a bounded
set when $\mathfrak{m}\in C_{b}(\mathbb{R};\mathcal{M}_{\Lambda })$.

Now, assuming (\ref{fix decay}), we can apply \cite[Theorem 5.5]{brupedraLR}
to the interaction $\Psi ^{\gimel _{\rho ,s}}$: 
\begin{equation}
\forall s,r,t\in {\mathbb{R}}:\qquad \partial _{r}\tau _{t,r}^{\Psi ^{\gimel
_{\rho ,s}}}=-\delta ^{\Psi ^{\gimel _{\rho ,s}}\left( r\right) }\circ \tau
_{t,r}^{\Psi ^{\gimel _{\rho ,s}}}\ ,  \label{sdgsdgsgsg}
\end{equation}%
in the strong sense on the dense set $\mathcal{U}_{0}$. It is a highly
non-trivial outcome resulting again from Lieb-Robinson bounds for
multi-commutators of order three given in \cite[Theorems 4.11, 5.4]%
{brupedraLR}. Similar to (\ref{lebesgue chat}), by Lemmata \ref{estimate
useful} and \ref{lemma well copy(1)} together with Lebesgue's dominated
convergence theorem, we deduce from the Dyson-type series coming from (\ref%
{sdkljsdlkjf}) that%
\begin{equation*}
\partial _{s}\varpi _{\rho ,s}\left( t,A\right) \doteq \lim_{\varepsilon
\rightarrow 0}\tilde{\beth}\left( \varepsilon ,t,A\right) =\lim_{\varepsilon
\rightarrow 0}\varepsilon ^{-1}\left( \varpi _{\rho ,s+\varepsilon }\left(
t,A\right) -\varpi _{\rho ,s}\left( t,A\right) \right)
\end{equation*}%
exists, for all $s,t\in \mathbb{R}$, $\rho \in E$ and $A\in \mathcal{U}_{0}$%
, and is also given by a Dyson-type series. Note that the Dyson-type series
are well-defined because, for any $\tilde{\Lambda}\in \mathcal{P}_{f}$, the
mapping from $\mathbb{R}\times \mathcal{U}_{\tilde{\Lambda}}$ to $\mathbb{C}$
defined by 
\begin{equation}
\left( t,A\right) \mapsto \rho \circ \delta ^{\Psi ^{\gimel _{\rho
,s}}\left( s\right) }\circ \tau _{t,s}^{\Psi ^{\gimel _{\rho ,s}}}\left(
A\right)  \label{mapping Linfinit}
\end{equation}%
is continuous, by Lemma \ref{lemma extra1}.

By Proposition \ref{Theorem Lieb-Robinson copy(3)} (i), for any $\tilde{%
\Lambda}\in \mathcal{P}_{f}$, the complex-valued function $(t,A)\mapsto
\partial _{s}\varpi _{\rho ,s}\left( t,A\right) $ on $\mathbb{R}\times 
\mathcal{U}_{\tilde{\Lambda}}$ is the unique solution in $\xi \in C\left( 
\mathbb{R}\times \mathcal{U}_{\tilde{\Lambda}};\mathbb{C}\right) $ to the
equation 
\begin{equation}
\xi \left( t,A\right) =-\rho \circ \delta ^{\Psi ^{\gimel _{\rho ,s}}\left(
s\right) }\circ \tau _{t,s}^{\Psi ^{\gimel _{\rho ,s}}}\left( A\right) +%
\mathfrak{\maltese }_{A}\left[ \xi \right]  \label{toto final}
\end{equation}%
with $\mathfrak{\maltese }_{A}$ defined by (\ref{def dysons}). Compare with (%
\ref{sdkljsdlkjf}) taken at $\varepsilon =0$. To prove the uniqueness of a
solution in $\xi \in C\left( \mathbb{R}\times \mathcal{U}_{\tilde{\Lambda}};%
\mathbb{C}\right) $ to (\ref{toto final}), we use the same arguments than
for (\ref{integral equation 0}), keeping in mind Lemma \ref{lemma extra1}.
This is again a consequence of Lieb-Robinson bounds for multi-commutators of
order three given in \cite[Theorems 4.11, 5.4]{brupedraLR}.
\end{proof}

We conclude this section with the derivation of Liouville's equation for
(elementary) continuous and affine functions defined by (\ref{fA}), from
which Theorem \ref{classical dynamics I} is deduced.

\begin{lemma}[Liouville's equation for affine functions]
\label{Corollary bije+cocylbaby copy(1)}\mbox{ }\newline
Assume (\ref{fix decay}). Then, 
\begin{equation*}
\partial _{s}V_{t,s}^{\mathfrak{m}}(\hat{A})\left( \rho \right)
=-\lim_{L\rightarrow \infty }\{\mathrm{h}_{L}^{\mathfrak{m}(s)},V_{t,s}^{%
\mathfrak{m}}(\hat{A})\}\left( \rho \right) \ ,\qquad s,t\in \mathbb{R},\
A\in \mathcal{U}_{0},\ \vec{\ell}\in \mathbb{N}^{d},\ \rho \in E_{\vec{\ell}%
}\ ,
\end{equation*}%
with $\hat{A}\in \mathfrak{C}$ being defined by (\ref{fA}).
\end{lemma}

\begin{proof}
Fix $s\in \mathbb{R}$ and $\rho \in E$. By (\ref{fA}) and (\ref{classical
evolution family}), note that 
\begin{equation*}
V_{t,s}^{\mathfrak{m}}(\hat{A})=\varpi _{\rho ,s}\left( t\right) \left(
A\right) \equiv \varpi _{\rho ,s}\left( t,A\right) \ ,\qquad t\in \mathbb{R}%
,\ A\in \mathcal{U}\ .
\end{equation*}%
Therefore, by Lemma \ref{lemma contnuity copy(1)} and (\ref{chat0}), 
\begin{equation*}
\mathrm{D}V_{t,s}^{\mathfrak{m}}(\hat{A})\left( \rho \right) =\mathrm{D}%
\varpi _{\rho ,s}\left( t;A\right) =\mathfrak{D}\left( t,A\right) \ ,\qquad
t\in \mathbb{R},\ A\in \mathcal{U}_{0}\ .
\end{equation*}%
See also Definition \ref{convex Frechet derivative} and Equation (\ref%
{clear2}). By Proposition \ref{lemma extra2} and (\ref{chat00}), the
continuous complex-valued function 
\begin{equation*}
(t,A)\mapsto -\rho \circ \delta ^{\Psi ^{\varpi _{\rho ,s}|_{\mathcal{U}%
_{\Lambda }}}\left( s\right) }\left( \mathrm{D}V_{t,s}^{\mathfrak{m}}(\hat{A}%
)\left( \rho \right) \right)
\end{equation*}%
on $\mathbb{R}\times \mathcal{U}_{\tilde{\Lambda}}$ solves Equation (\ref%
{toto final}), like the well-defined continuous mapping 
\begin{equation*}
(t,A)\mapsto \partial _{s}V_{t,s}^{\mathfrak{m}}(\hat{A})\left( \rho \right)
=\partial _{s}\varpi _{\rho ,s}\left( t\right) \left( A\right) \equiv
\partial _{s}\varpi _{\rho ,s}\left( t,A\right)
\end{equation*}%
from $\mathbb{R}\times \mathcal{U}_{\tilde{\Lambda}}$ to $\mathbb{C}$ (Lemma %
\ref{Differentiability copy(1)}), at any fixed $\tilde{\Lambda}\in \mathcal{P%
}_{f}$. By uniqueness of the solution to (\ref{toto final}),%
\begin{equation*}
\partial _{s}V_{t,s}^{h}(\hat{A})\left( \rho \right) =-\rho \circ \delta
^{\Psi ^{\varpi _{\rho ,s}|_{\mathcal{U}_{\Lambda }}}\left( s\right) }\left( 
\mathrm{D}V_{t,s}^{\mathfrak{m}}(\hat{A})\left( \rho \right) \right) \
,\qquad s,t\in \mathbb{R},\ A\in \mathcal{U}_{0}\ .
\end{equation*}%
By tedious computations using Definitions \ref{dynamic series}, \ref{convex
Frechet derivative copy(1)}, Corollary \ref{Lemma cigare1}, Lemma \ref%
{density of periodic states copy(2)}, Proposition \ref{lemma extra2},
Lebesgue's dominated convergence theorem and Equations (\ref{def aussi utile}%
), (\ref{derivative classical}) and (\ref{def utile}), one meanwhile checks
that, for any $\vec{\ell}\in \mathbb{N}^{d}$ and $\rho \in E_{\vec{\ell}}$,%
\begin{equation*}
\rho \circ \delta ^{\Psi ^{\varpi _{\rho ,s}|_{\mathcal{U}_{\Lambda
}}}\left( s\right) }\left( \mathrm{D}V_{t,s}^{\mathfrak{m}}(\hat{A})\left(
\rho \right) \right) =\lim_{L\rightarrow \infty }\{\mathrm{h}_{L}^{\mathfrak{%
m}(s)},V_{t,s}^{\mathfrak{m}}(\hat{A})\}\left( \rho \right) \ ,\qquad s,t\in 
\mathbb{R},\ A\in \mathcal{U}_{0}\ .
\end{equation*}
\end{proof}

\section{Equivalent Definition of Translation-Invariant Long-Range Models 
\label{Long-range models}}

Recall that $\mathfrak{L}=\mathbb{Z}^{d}$, see Section \ref{Long-rande gef}.
In \cite[Definition 2.1]{BruPedra2}, we give a definition of
translation-invariant long-range models that differs from Equation (\ref%
{translatino invariatn long range models}). It turns out that any model of%
\begin{equation*}
\mathcal{M}_{1}^{(2)}\doteq \left\{ (\Phi ,\mathfrak{a})\equiv (\Phi ,(%
\mathfrak{a}_{n})_{n\in \mathbb{N}})\in \mathcal{W}_{1}^{\mathbb{R}}\times 
\mathcal{S}:\forall n\in \mathbb{N}\backslash \{2\},\mathfrak{a}%
_{n}=0\right\} \ ,
\end{equation*}%
where $\mathcal{W}_{1}^{\mathbb{R}}\doteq \mathcal{W}_{1}\cap \mathcal{W}^{%
\mathbb{R}}$, can be identified with a long-range model in the sense of \cite%
[Definition 2.1]{BruPedra2}, and vice-versa. (The notation $\mathcal{W}_{1}$
in \cite{BruPedra2} corresponds here to $\mathcal{W}_{1}^{\mathbb{R}}$.)
This identification can be done in a such a way that the sequences of local
Hamiltonians associated with each long-range model are the same in both
cases:\bigskip

\noindent \underline{(i):} We start with a preliminary observation which
simplifies the arguments. At $L\in \mathbb{N}$, the local Hamiltonian of any
model $\mathfrak{m}=\left( \Phi ,(0,\mathfrak{a}_{2},0,\ldots )\right) \in 
\mathcal{M}_{1}^{(2)}$ of Definition \ref{definition long range energy} is
equal to%
\begin{equation*}
U_{L}^{\mathfrak{m}}=U_{L}^{\Phi }+\frac{1}{2\left\vert \Lambda
_{L}\right\vert }\int_{\mathbb{S}^{2}}\left( U_{L}^{\Psi ^{(1)}}U_{L}^{\Psi
^{(2)}}+(U_{L}^{\Psi ^{(2)}})^{\ast }(U_{L}^{\Psi ^{(1)}})^{\ast }\right) 
\mathfrak{a}_{2}\left( \mathrm{d}\Psi ^{(1)},\mathrm{d}\Psi ^{(2)}\right) \ ,
\end{equation*}%
because $\mathfrak{a}_{2}$ is, by definition, self-adjoint, meaning that it
equals its pushforward through the homeomorphism (\ref{push forward
self-adjoint}) for $n=2$. Since, for any $A,B\in \mathcal{U}$, 
\begin{equation*}
AB+B^{\ast }A^{\ast }=\frac{1}{2}\left( \left( A^{\ast }+B\right) ^{\ast
}\left( A^{\ast }+B\right) -\left( A^{\ast }-B\right) ^{\ast }\left( A^{\ast
}-B\right) \right) \ ,
\end{equation*}%
observe that%
\begin{equation*}
U_{L}^{\mathfrak{m}}=U_{L}^{\Phi }+\frac{1}{\left\vert \Lambda
_{L}\right\vert }\int_{\mathbb{S}^{2}}\left( |U_{L}^{\Psi ^{(1)\ast }+\Psi
^{(2)}}|^{2}-|U_{L}^{\Psi ^{(1)\ast }-\Psi ^{(2)}}|^{2}\right) \mathfrak{a}%
_{2}\left( \mathrm{d}\Psi ^{(1)},\mathrm{d}\Psi ^{(2)}\right)
\end{equation*}%
with $|C|^{2}\doteq C^{\ast }C$ for $C\in \mathcal{U}$. Let $\mathbb{B}%
\supseteq \mathbb{S}$ be the unit closed ball of the Banach space $\mathcal{W%
}_{1}$\ of translation-invariant (complex) interactions and define the
continuous functions $F^{\pm }:\mathbb{S}^{2}\rightarrow \mathbb{B}$ by%
\begin{equation*}
F^{\pm }\left( \Psi ^{(1)},\Psi ^{(2)}\right) \doteq \frac{1}{2}\left( \Psi
^{(1)\ast }\pm \Psi ^{(2)}\right) \ ,\qquad \Psi ^{(1)},\Psi ^{(2)}\in 
\mathbb{S}\ .
\end{equation*}%
Denoting by $F_{\ast }^{\pm }\left( \mathfrak{a}_{2}\right) $ the two
pushforwards of the measure $\mathfrak{a}_{2}$ through the continuous
functions $F^{\pm }$, we arrive at the equality 
\begin{equation}
U_{L}^{\mathfrak{m}}=U_{L}^{\Phi }+\frac{1}{\left\vert \Lambda
_{L}\right\vert }\int_{\mathbb{B}}|U_{L}^{\Psi }|^{2}\mathfrak{a}\left( 
\mathrm{d}\Psi \right) \ ,\qquad \mathfrak{a}\doteq F_{\ast }^{+}\left( 
\mathfrak{a}_{2}\right) -F_{\ast }^{-}\left( \mathfrak{a}_{2}\right) \ .
\label{Hamiltonian local simple}
\end{equation}

Then, any model of $\mathcal{M}_{1}^{(2)}$ can be identified with a
long-range model in the sense of \cite[Definition 2.1]{BruPedra2}:

\begin{itemize}
\item The measure space of \cite[Definition 2.1]{BruPedra2} is $(\mathbb{B}%
,\Sigma ,|\mathfrak{a}|)$, with $\Sigma \equiv \Sigma _{\mathbb{B}}$ being
the Borel $\sigma $-algebra associated with $\mathbb{B}$. $\Sigma $ is
countably generated, by separability of $\mathcal{W}\supseteq \mathbb{B}$.
Ergo, by \cite[Proposition\ 3.4.5]{Cohn-Measure-Theory}, the space $L^{2}(%
\mathbb{B};\mathbb{C})\doteq L^{2}(\mathbb{B},|\mathfrak{a}|;\mathbb{C})$ of
square-integrable complex-valued functions on $\mathbb{B}$ is a separable
Hilbert space, i.e., $(\mathbb{B},\Sigma ,|\mathfrak{a}|)$ is a separable
measure space.

\item The corresponding $\mathcal{L}^{2}$-functions $(\Phi _{\Psi })_{\Psi
\in \mathbb{B}},(\Phi _{\Psi }^{\prime })_{\Psi \in \mathbb{B}}$ of \cite[%
Definition 2.1]{BruPedra2} are defined by 
\begin{equation*}
\Phi _{\Psi }\doteq \mathrm{Re}\left\{ \Psi \right\} \qquad \text{and}\qquad
\Phi _{\Psi }^{\prime }\doteq \mathrm{Im}\left\{ \Psi \right\} \qquad \text{%
with}\qquad \Psi \in \mathbb{B}\ .
\end{equation*}%
See (\ref{real-im part interaction}) for the definition of real and
imaginary parts of interactions. These two functions are Bochner measurable,
by \cite[Theorems 1.1 and 1.2]{pettis}, for they are continuous functions
and $\mathcal{W}$ is a separable Banach space. Additionally, they are $%
\mathcal{L}^{2}$-functions because there are bounded on a space of finite
measure.

\item The corresponding measurable function $\gamma _{\Psi }\in \{-1,1\}$ of 
\cite[Definition 2.3]{BruPedra2} is obtained from any Hahn decomposition $P_{%
\mathfrak{a}},N_{\mathfrak{a}}\in \Sigma $ of the signed measure $\mathfrak{a%
}$ (\ref{Hamiltonian local simple}) by 
\begin{equation*}
\gamma _{\Psi }=\mathbf{1}[\Psi \in P_{\mathfrak{a}}]-\mathbf{1}[\Psi \in N_{%
\mathfrak{a}}]\ ,\qquad \Psi \in \mathbb{B}\ ,
\end{equation*}%
where $P_{\mathfrak{a}}$ and $N_{\mathfrak{a}}$ are respectively positive
and negative sets for $\mathfrak{a}$.\bigskip
\end{itemize}

\noindent \underline{(ii):} Conversely, let $(\mathcal{A},\mathfrak{A},%
\mathfrak{a}_{0})$ be a separable measure space with $\mathfrak{A}$ and $%
\mathfrak{a}_{0}:\mathfrak{A}\rightarrow \mathbb{R}_{0}^{+}$ being
respectively some $\sigma $-algebra on $\mathcal{A}$\ and some measure on $%
\mathfrak{A}$. Fix a measurable function $a\mapsto \gamma _{a}\in \{-1,1\}$
on $\mathcal{A}$ (see \cite[Definition 2.3]{BruPedra2}) and a model 
\begin{equation}
(\Phi ,(\Phi _{a})_{a\in \mathcal{A}},(\Phi _{a}^{\prime })_{a\in \mathcal{A}%
})\in \mathcal{W}_{1}^{\mathbb{R}}\times \mathcal{L}^{2}\left( \mathcal{A},%
\mathcal{W}_{1}^{\mathbb{R}}\right) \times \mathcal{L}^{2}\left( \mathcal{A},%
\mathcal{W}_{1}^{\mathbb{R}}\right)  \label{totot}
\end{equation}%
in the sense of \cite[Definition 2.1]{BruPedra2}. ($\mathcal{W}_{1}$ in \cite%
{BruPedra2} refers here to $\mathcal{W}_{1}^{\mathbb{R}}$.)

Define the set%
\begin{equation*}
\mathcal{A}_{\Phi ,\Phi ^{\prime }}\doteq \left\{ a\in \mathcal{A}:\Phi
_{a}+i\Phi _{a}^{\prime }\neq 0\right\} \ .
\end{equation*}%
The mapping 
\begin{equation*}
a\mapsto \Phi _{a}+i\Phi _{a}^{\prime }
\end{equation*}%
from $\mathcal{A}_{\Phi ,\Phi ^{\prime }}$ to $\mathbb{R}^{+}$ is
measurable, since the vector space operations in $\mathcal{W}$ are jointly
continuous. Hence, $\mathcal{A}_{\Phi ,\Phi _{a}^{\prime }}$, the preimage
of $\mathcal{W}_{1}\backslash \{0\}$ by this mapping, is an element of the $%
\sigma $-algebra $\mathfrak{A}$. As a consequence, by \cite[Definition 2.3]%
{BruPedra2}, at $L\in \mathbb{N}$, the local Hamiltonian associated with a
model of \cite[Definition 2.1]{BruPedra2} equals 
\begin{equation}
U_{L}=U_{L}^{\Phi }+\frac{1}{|\Lambda _{L}|}\int_{\mathcal{A}_{\Phi ,\Phi
^{\prime }}}\gamma _{a}|U_{\Lambda _{L}}^{\Phi _{a}+i\Phi _{a}^{\prime
}}|^{2}\mathfrak{a}_{0}\left( \mathrm{d}a\right) \ ,  \label{local model}
\end{equation}%
recalling that $|C|^{2}\doteq C^{\ast }C$ for $C\in \mathcal{U}$.

Define the functions $G:\mathcal{A}_{\Phi ,\Phi ^{\prime }}\rightarrow 
\mathbb{S}$ by%
\begin{equation*}
G\left( a\right) \doteq g\left( a\right) \left( \Phi _{a}+i\Phi _{a}^{\prime
}\right) \ ,\qquad a\in \mathcal{A}_{\Phi ,\Phi ^{\prime }}\ ,
\end{equation*}%
where 
\begin{equation*}
g\left( a\right) \doteq \left\Vert \Phi _{a}+i\Phi _{a}^{\prime }\right\Vert
_{\mathcal{W}}^{-1}\ ,\qquad a\in \mathcal{A}_{\Phi ,\Phi ^{\prime }}\ .
\end{equation*}%
The mapping $g$ from $\mathcal{A}_{\Phi ,\Phi ^{\prime }}$ to $\mathbb{R}%
^{+} $ is measurable, as a composition of a measurable function from $%
\mathcal{A}_{\Phi ,\Phi ^{\prime }}$ to $\mathcal{W}_{1}\backslash \{0\}$
with a continuous one from $\mathcal{W}_{1}\backslash \{0\}$ to $\mathbb{R}%
^{+}$. Again by the joint continuity of the vector space operations in $%
\mathcal{W}$, $G$ is a measurable function. As a consequence, at $L\in 
\mathbb{N}$, by (\ref{local model}), the local Hamiltonian associated with a
model of \cite[Definition 2.1]{BruPedra2} equals 
\begin{equation*}
U_{L}=U_{L}^{\Phi }+\frac{1}{|\Lambda _{L}|}\int_{\mathbb{S}}|U_{\Lambda
_{L}}^{\Psi }|^{2}\mathfrak{a}\left( \mathrm{d}\Psi \right) \ ,\qquad 
\mathfrak{a}\doteq G_{\ast }(\gamma g^{-1}\mathfrak{a}_{0})\ ,
\end{equation*}%
where $G_{\ast }\left( \gamma g^{-2}\mathfrak{a}_{0}\right) $ is the
pushforward of the signed measure $\gamma g^{-2}\mathfrak{a}_{0}$ through
the measurable function $G$. Note that $\mathfrak{a}$ is a finite measure on 
$\mathbb{S}$ because 
\begin{equation*}
\int_{\mathbb{S}}\left\vert \mathfrak{a}\right\vert \left( \mathrm{d}\Psi
\right) =\int_{\mathbb{S}}\left\Vert \Phi _{a}+i\Phi _{a}^{\prime
}\right\Vert _{\mathcal{W}}^{2}\mathfrak{a}_{0}\left( \mathrm{d}a\right)
\leq 2\int_{\mathbb{S}}\left\Vert \Phi _{a}\right\Vert _{\mathcal{W}}^{2}%
\mathfrak{a}_{0}\left( \mathrm{d}a\right) +2\int_{\mathbb{S}}\left\Vert \Phi
_{a}^{\prime }\right\Vert _{\mathcal{W}}^{2}\mathfrak{a}_{0}\left( \mathrm{d}%
a\right) <\infty \ .
\end{equation*}%
Now, using the continuous mapping $K$\ from $\mathbb{S}$ to $\mathbb{S\times
S}$ defined by%
\begin{equation*}
K\left( \Psi \right) \doteq \left( \Psi ^{\ast },\Psi \right) \ ,\qquad \Psi
\in \mathbb{S}\ ,
\end{equation*}%
we define 
\begin{equation*}
\mathfrak{a}_{2}\doteq \frac{1}{2}\left( K_{\ast }(\mathfrak{a})+K_{\ast }(%
\mathfrak{a})^{\ast }\right)
\end{equation*}%
to be the real part of the pushforward of the signed measure $\mathfrak{a}$
through the measurable function $K$. Then, by construction, $\mathfrak{a}%
_{2}=\mathfrak{a}_{2}^{\ast }\mathfrak{\ }$is self-adjoint and 
\begin{equation*}
U_{L}=U_{L}^{\Phi }+\frac{1}{2\left\vert \Lambda _{L}\right\vert }\int_{%
\mathbb{S}^{2}}\left( U_{L}^{\Psi ^{(1)}}U_{L}^{\Psi ^{(2)}}+(U_{L}^{\Psi
^{(2)}})^{\ast }(U_{L}^{\Psi ^{(1)}})^{\ast }\right) \mathfrak{a}_{2}\left( 
\mathrm{d}\Psi ^{(1)},\mathrm{d}\Psi ^{(2)}\right)
\end{equation*}%
for any $L\in \mathbb{N}$.

\section{The BCS Theory as a Long-Range Model\label{Section applications}}

The most general form of a translation invariant model for fermions with
two-body interactions\ in a cubic box $\Lambda _{L}\doteq \{\mathbb{Z}\cap %
\left[ -L,L\right] \}^{d}$ (see (\ref{eq:def lambda n})) of volume $|\Lambda
_{L}|$, $L\in \mathbb{N}$, is given in momentum space by%
\begin{equation}
\mathrm{H}_{L}^{Full}=\underset{k\in \Lambda _{L}^{\ast },\ \mathrm{s}\in 
\mathrm{S}}{\sum }\left( \varepsilon _{k}-\mu \right) \tilde{a}_{k}^{\ast }%
\tilde{a}_{k}+\frac{1}{\left\vert \Lambda _{L}\right\vert }\underset{\mathrm{%
s}_{1},\mathrm{s}_{2},\mathrm{s}_{3},\mathrm{s}_{4}\in \mathrm{S}}{\underset{%
k,k^{\prime },q\in \Lambda _{L}^{\ast }}{\sum }}g_{\mathrm{s}_{1},\mathrm{s}%
_{2},\mathrm{s}_{3},\mathrm{s}_{4}}\left( k,k^{\prime },q\right) \tilde{a}%
_{k+q,\mathrm{s}_{1}}^{\ast }\tilde{a}_{k^{\prime }-q,\mathrm{s}_{2}}^{\ast }%
\tilde{a}_{k^{\prime },\mathrm{s}_{3}}\tilde{a}_{k,\mathrm{s}_{4}}\ .
\label{hamil general}
\end{equation}%
See \cite[Eq. (2.1)]{Metzner}. Recall that $\mathrm{S}$ is some finite
(spin) set representing the internal degrees of freedom of quantum
particles. The set 
\begin{equation*}
\Lambda _{l}^{\ast }\doteq \frac{2\pi }{(2l+1)}\Lambda _{l}\subseteq \left[
-\pi ,\pi \right] ^{d}
\end{equation*}%
is the reciprocal lattice of quasi-momenta (periodic boundary conditions)
associated with $\Lambda _{L}$ and the operator 
\begin{equation*}
\tilde{a}_{k,\mathrm{s}}^{\ast }\doteq \frac{1}{\left\vert \Lambda
_{L}\right\vert ^{1/2}}\underset{x\in \Lambda _{L}}{\sum }\mathrm{e}%
^{-ik\cdot x}a_{x,\mathrm{s}}^{\ast }
\end{equation*}%
(respectively $\tilde{a}_{k,\mathrm{s}}$) creates (respectively annihilates)
a fermion with spin $\mathrm{s}\in \mathrm{S}$ and (quasi-) momentum $k\in
\Lambda _{L}^{\ast }$. The function $\varepsilon _{k}$ represents the
kinetic energy of a fermion with (quasi-) momentum $k$. In physics, it is
usually the Fourier transform of the discrete Laplacian. The real number $%
\mu $ is the chemical potential. The last term of (\ref{hamil general})
corresponds to a translation-invariant two-body interaction written in the
momentum space.

One important example of a lattice-fermion system with long-range
interactions is given in the scope of the celebrated BCS theory -- proposed
in the late 1950s (1957) to explain conventional type I superconductors. The
lattice version of this theory is obtained from (\ref{hamil general}) by
taking $\mathrm{S}\doteq \{\uparrow ,\downarrow \}$ and imposing 
\begin{equation*}
g_{\mathrm{s}_{1},\mathrm{s}_{2},\mathrm{s}_{3},\mathrm{s}_{4}}\left(
k,k^{\prime },q\right) =\delta _{k,-k^{\prime }}\delta _{\mathrm{s}%
_{1},\uparrow }\delta _{\mathrm{s}_{2},\downarrow }\delta _{\mathrm{s}%
_{3},\downarrow }\delta _{\mathrm{s}_{4},\uparrow }f\left( k,-k,q\right)
\end{equation*}%
for some function $f$: It corresponds to the so-called (reduced) BCS\
Hamiltonian%
\begin{equation}
\mathrm{H}_{L}^{BCS}\doteq \sum\limits_{k\in \Lambda _{L}^{\ast },\ \mathrm{s%
}\in \mathrm{S}}\left( \varepsilon _{k}-\mu \right) \tilde{a}_{k,\mathrm{s}%
}^{\ast }\tilde{a}_{k,\mathrm{s}}-\frac{1}{\left\vert \Lambda
_{L}\right\vert }\sum_{k,q\in \Lambda _{L}^{\ast }}\gamma _{k,q}\tilde{a}%
_{k,\uparrow }^{\ast }\tilde{a}_{-k,\downarrow }^{\ast }\tilde{a}%
_{-q,\downarrow }\tilde{a}_{q,\uparrow }\ ,  \label{BCS Hamilt}
\end{equation}%
where $\gamma _{k,q}$ is a positive\footnote{%
The positivity of $\gamma _{k,q}$ imposes constraints on the choice of the
function $f$.} function. Because of the term $\delta _{k,-k^{\prime }}$, the
interaction of this model has a long-range character, in position space. In
physics, one usually takes 
\begin{equation*}
\gamma _{k,q}\doteq \left\{ 
\begin{array}{l}
\gamma \geq 0 \\ 
0%
\end{array}%
\begin{array}{l}
\mathrm{for\ }\left\vert k-q\right\vert \leq \mathrm{C} \\ 
\mathrm{for\ }\left\vert k-q\right\vert >\mathrm{C}%
\end{array}%
\right.
\end{equation*}%
with $\mathrm{C}\in \left( 0,\infty \right] $. The simple choice $\mathrm{C}%
=\infty $, i.e., $\gamma _{k,q}=\gamma >0$ in (\ref{BCS Hamilt}), is still
physically very interesting since, even when $\varepsilon _{k}=0$, the BCS\
Hamiltonian qualitatively displays most of basic properties of real
conventional type I superconductors. See, e.g. \cite[Chapter VII, Section 4]%
{Thou}.

The BCS\ Hamiltonian $\mathrm{H}_{L}^{BCS}$ with $\gamma _{k,q}=\gamma >0$
can be rewritten in the $x$--space as%
\begin{equation}
\mathrm{H}_{L}^{BCS}=\sum\limits_{x,y\in \Lambda _{L},\ \mathrm{s}\in 
\mathrm{S}}h\left( x-y\right) a_{x,\mathrm{s}}^{\ast }a_{y,\mathrm{s}}-\mu
\sum\limits_{x\in \Lambda _{L},\ \mathrm{s}\in \mathrm{S}}a_{x,\mathrm{s}%
}^{\ast }a_{x,\mathrm{s}}-\frac{\gamma }{\left\vert \Lambda _{L}\right\vert }%
\sum_{x,y\in \Lambda _{N}}a_{x,\uparrow }^{\ast }a_{x,\downarrow }^{\ast
}a_{y,\downarrow }a_{y,\uparrow }  \label{BCS model}
\end{equation}%
for some appropriate function $h:\mathbb{Z}^{d}\rightarrow \mathbb{C}$
representing the (kernel of) one-particle operator (or the Fourier transform
of the function $k\mapsto \varepsilon _{k}-\mu $). Note that the long-range
character of the BCS interaction in (\ref{BCS model}) is in this case clear
since it can be seen as an hopping at any distance of (Cooper) pairs of
fermions with spins $\downarrow $ and $\uparrow $, respectively. It is also
a mean-field term since it can be seen as a space average of interactions: 
\begin{equation*}
-\frac{1}{\left\vert \Lambda _{L}\right\vert }\sum_{x,y\in \Lambda
_{L}}a_{x,\uparrow }^{\ast }a_{x,\downarrow }^{\ast }a_{y,\downarrow
}a_{y,\uparrow }=-\sum_{y\in \Lambda _{L}}\left( \frac{1}{\left\vert \Lambda
_{L}\right\vert }\sum_{x\in \Lambda _{L}}a_{x,\uparrow }^{\ast
}a_{x,\downarrow }^{\ast }\right) a_{y,\downarrow }a_{y,\uparrow }\ .
\end{equation*}%
Therefore, the BCS\ Hamiltonian $\mathrm{H}_{L}^{BCS}$ with $\gamma
_{k,q}=\gamma >0$ can be explicitly expressed as the local Hamiltonian $%
U_{L}^{\mathfrak{m}_{0}}$ of a long-range model $\mathfrak{m}_{0}\in 
\mathcal{M}$: Let $\Psi ^{\left( h\right) }$ be defined by 
\begin{equation*}
\Psi _{\Lambda }^{\left( h\right) }\doteq h\left( x-y\right) \sum\limits_{%
\mathrm{s}\in \mathrm{S}}a_{x,\mathrm{s}}^{\ast }a_{y,\mathrm{s}}+\left(
1-\delta _{x,y}\right) h\left( x-y\right) \sum\limits_{\mathrm{s}\in \mathrm{%
S}}a_{y,\mathrm{s}}^{\ast }a_{x,\mathrm{s}}
\end{equation*}%
whenever $\Lambda =\left\{ x,y\right\} $ for $x,y\in \mathbb{Z}^{d}$, and $%
\Psi _{\Lambda }^{\left( h\right) }\doteq 0$ otherwise. The interaction $%
\Psi ^{\left( h\right) }$ belongs to $\mathcal{W}^{\mathbb{R}}$ as soon as%
\begin{equation}
\underset{x,y\in \mathbb{Z}^{d}}{\sup }\frac{\left\vert h\left( x-y\right)
\right\vert }{\mathbf{F}\left( x,y\right) }<\infty .
\label{condition facile}
\end{equation}%
See Equation (\ref{iteration0}). Such a condition is trivially satisfied for
all kinetic terms used in condensed matter physics for many-fermions on
lattices, by taking $\mathbf{F}$ as in Equation (\ref{examples}). Let $%
\tilde{\Psi}\in \mathbb{S}$ be defined by $\tilde{\Psi}_{\left\{ x\right\}
}\doteq a_{x,\downarrow }a_{x,\uparrow }$ for $x\in \mathbb{Z}^{d}$ and $%
\tilde{\Psi}_{\Lambda }\doteq 0$ otherwise. Let $\mathfrak{\tilde{a}}$ be
defined, for all Borel subset $\mathfrak{B}\subseteq \mathbb{S}$, by 
\begin{equation*}
\mathfrak{\tilde{a}}\left( \mathfrak{B}\right) =-\gamma \mathbf{1}[\tilde{%
\Psi}\in \mathfrak{B}].
\end{equation*}%
If Inequality (\ref{condition facile}) holds true, then $\mathfrak{m}%
_{0}=(\Psi ^{\left( h\right) },\mathfrak{\tilde{a}})\in \mathcal{M}$ and, by
Definition \ref{definition long range energy}, 
\begin{equation*}
U_{L}^{\mathfrak{m}_{0}}=\mathrm{H}_{L}^{BCS},\qquad L\in \mathbb{N}.
\end{equation*}%
As a consequence, the dynamics at infinite volume of the BCS\ Hamiltonian $%
\mathrm{H}_{L}^{BCS}$ with $\gamma _{k,q}=\gamma >0$ can be explicitly
computed from results of this paper, in which concerns its classical part,
and from \cite{Bru-pedra-MF-III} for its quantum part.

Other examples can also be found in \cite[Section 2.2]{BruPedra2} (in
relation with the forward scattering approximation) as well as in \cite%
{Bru-pedra-proceeding,Bru-pedra-MF-IV}, which explain the dynamics of the
strong-coupling BCS-Hubbard model. The latter is an interesting model
because it predicts the existence of a superconductor-Mott insulator phase
transition, like in cuprates which must be doped to become superconductors.
See \cite{BruPedra1} for more details. \bigskip

\noindent \textit{Acknowledgments:} This work is supported by CNPq
(308337/2017-4), FAPESP (2017/22340-9), as well as by the Basque Government
through the grant IT641-13 and the BERC 2018-2021 program, and by the
Spanish Ministry of Science, Innovation and Universities: BCAM Severo Ochoa
accreditation SEV-2017-0718, MTM2017-82160-C2-2-P. The formulation of
long-range models as pairs $\left( \Phi ,\mathfrak{a}\right) \in \mathcal{W}%
_{1}^{\mathbb{R}}\times \mathcal{S}(\mathbb{S)}$ with associated
Hamiltonians defined by (\ref{Hamiltonian local simple}), instead of \cite[%
Definitions 2.1 and 2.3]{BruPedra2}, has been proposed to us by S. Breteaux.
We thank him very much for this hint.

\end{document}